\documentclass[leqno]{article}

\usepackage{siunitx,etoolbox}
\usepackage{array}
\newcolumntype{H}{>{\setbox0=\hbox\bgroup}c<{\egroup}@{}}
\renewrobustcmd{\bfseries}{\fontseries{b}\selectfont}
\newrobustcmd{\B}{\bfseries}
\newcommand{\commentout}[1]{}
\usepackage{amsfonts,amsmath}
\usepackage{amssymb}
\usepackage{amsthm}
\usepackage{paralist}
\usepackage{multirow}
\usepackage{xspace}
\usepackage{graphicx}
\usepackage[numbers, sort&compress]{natbib}
\usepackage[noend,ruled,algo2e]{algorithm2e}
\usepackage{subcaption}
\usepackage{float}
\usepackage{nicefrac}
\usepackage{tabularx}
\usepackage{mathtools}
\usepackage{tikz}
\usetikzlibrary{patterns,decorations.pathreplacing}
\usetikzlibrary{positioning}
\usepackage{booktabs}
\usepackage{mathtools}
\usepackage{authblk}

\usetikzlibrary{shapes}

\usepackage[font=small]{caption}

\makeatletter
\def\NAT@spacechar{~}
\makeatother

\usepackage[pagebackref]{hyperref}
\usepackage{cleveref}

\crefname{theorem}{Theorem}{Theorems}

\newtheorem{proposition}{Proposition}
\newtheorem{definition}{Definition}
\newtheorem{theorem}{Theorem}
\newtheorem{lemma}{Lemma}

\newtheorem{example}{Example}
\crefname{example}{Example}{Examples}

\newcommand{\DTW}{\textsc{DTW-Mean}\xspace}
\newcommand{\wDTW}{\textsc{Weighted~DTW-Mean}\xspace}

\newcommand{\BDTW}{\textsc{Binary DTW-Mean}\xspace}

\newcommand{\MSA}{\textsc{Multiple Sequence Alignment}\xspace}
\newcommand{\MSAacr}{\textsc{MSA}\xspace}
\newcommand{\STS}{\textsc{Steiner String}\xspace}
\newcommand{\STSacr}{\textsc{STS}\xspace}

\newcommand{\N}{\mathbb{N}}
\newcommand{\Q}{\mathbb{Q}}

\newcommand{\T}{\mathcal{T}}

\DeclareMathOperator{\dtw}{dtw}
\providecommand{\norm}[1]{\left\Vert #1 \right\Vert}
\DeclareMathOperator*{\argmin}{arg\,min}

\newcommand{\Optproblemdef}[3]{
	\begin{center}
  \begin{minipage}{0.95\textwidth}
    \noindent
    \textsc{#1}

			\vspace{2pt}
			\setlength{\tabcolsep}{3pt}
			\begin{tabularx}{\textwidth}{@{}lX@{}}
					\textbf{Input:} 		& #2 \\
					\textbf{Task:} 	& #3
				\end{tabularx}
  \end{minipage}
	\end{center}
}

\usepackage{etoolbox}

\newcommand{\aedp}{\textsc{EDP}\xspace}
\newcommand{\asym}{\textsc{MAL}\xspace}
\newcommand{\adba}{\textsc{DBA}\xspace}
\newcommand{\asoft}{\textsc{SDTW}\xspace}
\newcommand{\absg}{\textsc{BSG}\xspace}
\newcommand{\assg}{\textsc{SSG}\xspace}

\renewcommand{\S}[1]{{\mathcal{#1}}}
\newcommand{\rw}{\text{rw}}
\newcommand{\ucr}{\text{ucr}}
\newcommand{\ucm}{\text{ucm}}
\newcommand{\ncm}{\text{ncm}}
\newcommand{\cm}{\text{cm}}
\newcommand{\med}{\text{med}}
\newcommand{\avg}{\text{avg}}
\newcommand{\std}{\text{std}}
\newcommand{\total}{\text{total}}
 
\newcommand{\struct}{\text{struct}}

\newcommand{\abs}[1]{\mathop{\left\lvert #1 \right\rvert}} 
\newcommand{\args}[1]{\mathop{\left( #1 \right)}} 
\newcommand{\cbrace}[1]{\mathop{\left\{ #1 \right\}}}

\newcommand{\Oh}{O}

\newcommand{\thetitle}{Exact Mean Computation in Dynamic Time Warping Spaces}
\title{\thetitle}

\author[1]{Markus~Brill}
\author[2]{Till~Fluschnik}
\author[2]{Vincent~Froese}
\author[3]{Brijnesh~Jain}
\author[2]{Rolf~Niedermeier}
\author[3]{David~Schultz}

\affil[1]{\small Efficient Algorithms, TU~Berlin, Germany, \texttt{brill@tu-berlin.de}}
\affil[2]{\small Algorithmics and Computational Complexity, TU~Berlin, Germany, \texttt{\{till.fluschnik,vincent.froese,rolf.niedermeier\}@tu-berlin.de}}
\affil[3]{\small Distributed Artificial Intelligence Laboratory, TU Berlin, Germany, \texttt{\{brijnesh-johannes.jain,david.schultz\}@dai-labor.de}}

\date{}

\begin{document}
\maketitle

\begin{abstract}
  Dynamic time warping constitutes a major tool for analyzing time series.
  In particular, computing a mean series of a given sample of series in dynamic time warping spaces (by minimizing the Fr\'echet function) is a challenging computational problem, so far solved by several heuristic and inexact strategies.
  We spot some inaccuracies in the literature on exact mean computation in dynamic time warping spaces.
  Our contributions comprise an exact dynamic program computing a mean (useful for benchmarking and evaluating known heuristics).
  Based on this dynamic program, we empirically study properties like uniqueness and length of a mean.
  Moreover, experimental evaluations reveal substantial deficits of state-of-the-art heuristics in terms of their output quality.
  We also give an exact polynomial-time algorithm for the special case of binary time series.
  
\end{abstract}

\section{Introduction}
Time series such as acoustic signals, electrocardiograms, and internet traffic data are time-dependent observations that vary in length and temporal dynamics. Given a sample of time series, to filter out the corresponding variations, one major direction to time series averaging applies \emph{dynamic time warping} (dtw). Different dtw-averaging approaches have been applied to improve nearest neighbor classifiers and to formulate centroid-based clustering algorithms in dtw-spaces \cite{Abdulla2003,HNF08,Morel2018,Oates1999,PFWNCK16,Rabiner1979,Soheily-Khah2016}.

The most successful approaches pose time series averaging as an optimization problem \cite{CB17,HNF08,PKG11,SJ17,Soheily-Khah2016}: Suppose that $\S{X} = \args{x^{(1)}, \dots, x^{(k)}}$ is a sample of $k$ time series $x^{(i)}$. 
Then a (weighted Fr\'echet) \emph{mean} in dtw-spaces is usually defined as any time series~$z$ that minimizes the \emph{weighted Fr\'echet function}~\cite{Frechet1948}
\[
F_w(z) := \sum_{i=1}^k w_i\cdot\left.\dtw\!\args{z, x^{(i)}}\right.^2,
\]
where $\dtw(x,y)$ denotes the dtw-distance between time series~$x$ and~$y$ and $w=(w_1,\ldots,w_k)$ is a given~$k$-dimensional weight vector (usually assumed to satisfy $0\le w_i\le 1$ for all~$i\in[k]$ and~$\sum_{i=1}^kw_i=1$).
We refer to the problem of minimizing~$F_w$ over the set $\S{T}$ of all finite time series as the  \wDTW problem.
In the special case of~$w_i=1/k$ for all~$i\in[k]$, we write~$F(z)$ and refer to the problem of minimizing~$F$ over the set $\S{T}$ as \DTW.

A variant of the \DTW problem constrains the solution set $\S{T}$ to the subset $\S{T} _m \subseteq \S{T}$ of all length-$m$ time series. For both the constrained and unconstrained \DTW variant, solutions are guaranteed to exist, but are not unique in general~\cite{SJ17}. 

As has been shown very recently, \DTW is NP-hard, W[1]-hard with respect to the sample size~$k$, and not solvable in~$\rho(k)\cdot n^{o(k)}$ time for any computable function~$\rho$ (where~$n$ is the maximum length of any input time series) assuming a plausible complexity-theoretic hypothesis~\cite{BFN18}.
Algorithms with exponential time complexity have been falsely claimed to provide optimal solutions \cite{HNF08,PG12,PKG11}. Existing heuristics approximately solve the constrained  \DTW problem \cite{CB17,PKG11,SJ17}. Thereby, the length~$m$ of feasible solutions is specified beforehand without any knowledge about whether the subset~$\S{T} _m$ contains an optimal solution of the unconstrained \DTW problem. 
In summary, the development of nontrivial exact algorithms are to be considered widely open. 

\paragraph{Our Contributions.}
We discuss several problematic statements in the literature concerning the computational complexity of exact algorithms for \DTW. 
We refute (supplying counterexamples) some false claims from the literature and clarify the known state of the art with respect to computing means in dtw-spaces (\Cref{sec:probstat}). 
We show that in case of binary time series (both input and mean) there is an exact polynomial-time algorithm for mean computation in dtw-spaces (\Cref{sec:binaryDTW}).
We develop a dynamic program as an exact algorithm for the (unconstrained) \wDTW problem on rational time series. 
The worst-case time complexity of the proposed dynamic program is $O(n^{2k+1}2^kk)$, where~$k$ is the sample size and~$n$ is the maximum length of a sample time series (\Cref{sec:exact}).  
We apply the proposed exact dynamic program on small-scaled problems as a benchmark of how well the state-of-the-art heuristics approximate a mean. Our empirical findings reveal that all state-of-the-art heuristics suffer from relatively poor worst-case solution quality in terms of minimizing the Fr\'echet function, and the solution quality in general may vary quite a lot (\Cref{sec:exp}).

\section{Preliminaries}
\label{sec:prelim}

Throughout this paper, we consider only finite univariate time series with rational elements.
A univariate \emph{time series} of length~$n$ is a sequence~$x = (x_1,\ldots, x_n)\in\Q^n$.
We denote the set of all univariate rational time series of length $n$ by~$\T_n$. Furthermore, $\T = \bigcup_{n \in \N} \T_n$ denotes the set of all univariate rational time series of finite length.
For every~$n\in\N$, let $[n]$ denote the set $\{1,\ldots,n\}$.

The next definition is fundamental for our central computational problem.

\begin{definition}
  A \emph{warping path of order~$m \times n$} is a sequence~$p=(p_1,\ldots,p_L)$
  with~$p_\ell\in [m]\times[n]$ for all~$\ell\in[L]$ such that
  \begin{compactenum}[i)]
    \item $p_1 = (1,1)$, 
	\item $p_L=(m,n)$, and
    \item $p_{\ell+1} - p_\ell \in \{(1,0),(0,1),(1,1)\}$ for all~$\ell\in[L-1]$.
    \end{compactenum}
  \end{definition}

  \noindent Note that~$\max\{m,n\} \le L \le m+n$.
  We denote the set of all warping paths of order~$m\times n$ by~$\mathcal{P}_{m,n}$.
  For two time series~$x=(x_1,\ldots,x_m)$ and~$y=(y_1,\ldots,y_n)$, a warping path~$p=(p_1,\ldots,p_L)\in\mathcal{P}_{m,n}$ defines an alignment of~$x$ and~$y$:
  each pair~$p_\ell=(i_\ell,j_\ell)$ of~$p$ \emph{aligns} element~$x_{i_\ell}$ with~$y_{j_\ell}$.
  The \emph{cost} $C_p(x,y)$ for aligning~$x$ and~$y$ along warping path~$p$ is defined as $C_p(x,y)=\sum_{\ell=1}^L (x_{i_\ell} - y_{j_\ell})^2$.
  The \emph{dtw-distance} between~$x$ and~$y$ is defined as
  
    \[\dtw(x,y):= \min_{p\in\mathcal{P}_{m,n}}\left\{\sqrt{C_p(x,y)}\right\}.\]
    A warping path~$p$ with~$C_p(x,y) = \left(\dtw(x,y)\right)^2$ is called an~\emph{optimal} warping path for~$x$ and~$y$.

We remark that~$\dtw(x,y)$ for two time series~$x,y$ of length~$n$ can be computed in subquadratic time~$O(n^2 \log\log\log{n}/\log\log{n})$~\cite{GS17}.
The existence of a strongly subquadratic-time (that is,~$O(n^{2-\epsilon})$ for some~$\epsilon > 0$) algorithm is considered unlikely~\cite{BK15}.

Our central computational problem is defined as follows. 
\Optproblemdef{\wDTW}
{Sample $\S{X} = \args{x^{(1)}, \dots, x^{(k)}}$ of $k$ univariate rational time series and rational nonnegative weights~$w_1,\ldots,w_k$.}
{A univariate rational time series~$z$ that minimizes the weighted Fr\'echet function $F_w(z)$.}

\noindent The special case of uniform weights (that is, $w_i=\frac{1}{k}$ for all $i\in[k]$) is called \DTW.
It is known that a weighted mean always exists~\cite[Remark~2.13]{JS18} (however, it is not necessarily unique).
Note that for rational inputs, every weighted mean is also rational (see \Cref{lem:opt_mean_value}), that is, \wDTW always has a solution.

\section{Problematic Statements in the Literature}
\label{sec:probstat}
In this section we discuss misleading and wrong claims and errors in the literature.

\subsection{NP-Hardness}
\label{ssec:nphardness}
The \DTW problem is often related to the \STS (\STSacr) problem \cite{ASW15,HNF08,Marteau16, FBELP17, PG17,PG12,PKG11,PFWNCK16}. A \emph{Steiner string} \cite{G97} (or \emph{Steiner sequence}) for a set~$S$ of strings is a string~$t$ that minimizes
$\sum_{s \in S} D(t,s)$,
where $D$ is a distance measure between two strings often assumed to fulfill the triangle inequality (e.g., the weighted edit distance).
Computing a Steiner string is equivalent to solving the~\MSA (\MSAacr) problem~\cite{G97}.
Both, \STSacr and \MSAacr, are known to be NP-hard for several distance measures even on binary alphabets~\cite{BD01, NR05}. 
Interestingly, as we show in \Cref{sec:binaryDTW}, \wDTW is solvable in polynomial time for the space of binary time series.

Several papers mention the NP-hardness results for \MSAacr and \STSacr in the context of \DTW~\cite{HNF08, Marteau16, PG17, FBELP17}, thereby suggesting that \DTW is also NP-hard.
However, it is not clear (and not shown in the literature) how to reduce from \MSAacr (or \STSacr) to \DTW since the involved distance measures are significantly different. For example, the $\dtw$-distance lacks the metric property of the edit distance.
Very recently, devising an intricate reduction from the \textsc{Multicolored Clique} problem, \citet{BFN18} showed that \DTW is NP-hard, W[1]-hard with respect to the sample size~$k$, and not solvable in time~$\rho(k)\cdot n^{o(k)}$ for any function~$\rho$, where~$n$ is the maximum length of any sample time series (unless the \emph{Exponential Time Hypothesis} fails).

\subsection{Computation of Exact Solutions}
\label{ssec:exact}

Two exponential-time algorithms were proposed to exactly solve \DTW. 
The first approach is based on multiple (sequence) alignment in bioinformatics~\cite{G97} and the second is a brute-force method \cite{PKG11}. 
We show that neither algorithm is guaranteed to return an optimal solution for every \DTW{} instance.

\paragraph{Multiple Alignment.}
It is claimed that \DTW can be solved by averaging a (global) multiple alignment of the~$k$ input series~\cite{HNF08,PKG11,PG12,PFWNCK16}. A multiple alignment of~$k$ time series in the context of dynamic time warping is described as computing a $k$-dimensional warping path in a $k$-dimensional matrix (a precise formal definition is not given).
Concerning the running time, it is claimed that computing a multiple alignment requires~$\Theta(n^k)$ time~\cite{PKG11, PG12} ($\Oh(n^k)$ time~\cite{PFWNCK16}), where~$n$ is the maximum length of an input time series. Neither the upper bound of~$\Oh(n^k)$ nor the lower bound of~$\Omega(n^k)$ on the running time are formally proven.

Given a multiple alignment, it is claimed that averaging the~$k$ resulting aligned time series column-wise yields a mean \cite[Definition 4]{PG12}.
We show that this is not correct even for two time series.
Note that for two time series, a multiple alignment is simply obtained by an optimal warping path.
However, the column-wise average of two aligned time series obtained from an optimal warping path is not always an optimal solution for a \DTW{} instance as 
the following example shows.

\begin{example}\label{ex:MAandMean}\normalfont 
Let $x^{(1)} = (1,4,2,3)$ and $x^{(2)} = (4,2,4,5)$. 
Using an exhaustive search, we obtain the unique optimal warping path
\[p = ((1,1), (2,1), (3,2), (4,3), (4,4))\] of length five.
The two corresponding aligned length-5 time series are 
\begin{align*}
 x_p^{(1)} &= (1,4,2,3,3), \\
 x_p^{(2)} &= (4,4,2,4,5).
\end{align*}
The arithmetic mean of these time series is $\bar{x} = (2.5,4,2,3.5,4)$.
However, for $z = (2.5, 4, 2, 4)$, we have $F(z) = 6.5 < 7 = F(\bar{x})$,
which shows that $\bar{x}$ is not a mean (see \Cref{fig:ex1}).
\end{example}

\begin{figure}[t]
 \centering
 \begin{tikzpicture}
    \usetikzlibrary{calc}

    \def\ysc{0.66}
    \node (x1) at (0,0)[]{$x^{(1)}=(1,4,2,3)$};
    \node (x) at (0,-1.5*\ysc)[]{$\bar{x}=(2.5,4,2,3.5,4)$};
    \node (x2) at (0,-3*\ysc)[]{$x^{(2)}=(4,2,4,5)$};

    \draw[thick] ($(x) -(1-0.66,-0.25)$) -- ($(x1) -(0.6-0.6,0.25)$);
    \draw[thick] ($(x) -(0.4-0.4,-0.25)$) -- ($(x1) -(0.2-0.5,0.25)$);
    \draw[thick] ($(x) -(0-0.3,-0.25)$) -- ($(x1) -(-0.2-0.4,0.25)$);
    \draw[thick] ($(x) -(-0.5-0.2,-0.25)$) -- ($(x1) -(-0.6-0.3,0.25)$);
    \draw[thick] ($(x) -(-1-0.1,-0.25)$) -- ($(x1) -(-0.6-0.3,0.25)$);

    \draw[thick] ($(x) -(1-0.66,0.25)$) -- ($(x2) -(0.6-0.6,-0.25)$);
    \draw[thick] ($(x) -(0.4-0.4,0.25)$) -- ($(x2) -(0.6-0.6,-0.25)$);
    \draw[thick] ($(x) -(0-0.3,0.25)$) -- ($(x2) -(0.2-0.5,-0.25)$);
    \draw[thick] ($(x) -(-0.5-0.2,0.25)$) -- ($(x2) -(-0.2-0.4,-0.25)$);
    \draw[thick] ($(x) -(-1-0.1,0.25)$) -- ($(x2) -(-0.6-0.3,-0.25)$);


    \def\xsh{4.5}

    \node (x1) at (0+\xsh,0)[]{$(1,4,2,3)=x^{(1)}$};
    \node (x) at (0+\xsh-0.25,-1.5*\ysc)[]{$(2.5,4,2,4)=z$};
    \node (x2) at (0+\xsh,-3*\ysc)[]{$(4,2,4,5)=x^{(2)}$};

    \draw[thick] ($(x) -(1-0.8+0.5,-0.25)$) -- ($(x1) -(0.6-0.7+1,0.25)$);
    \draw[thick] ($(x) -(0.4-0.7+0.6,-0.25)$) -- ($(x1) -(0.2-0.6+1,0.25)$);
    \draw[thick] ($(x) -(0-0.7+0.75,-0.25)$) -- ($(x1) -(-0.2-0.6+1.1,0.25)$);
    \draw[thick] ($(x) -(-0.5-0.55+0.8,-0.25)$) -- ($(x1) -(-0.6-0.45+1.1,0.25)$);

    \draw[thick] ($(x) -(1-0.8+0.5,0.25)$) -- ($(x2) -(0.6-0.6+1,-0.25)$);
    \draw[thick] ($(x) -(0.4-0.7+0.6,0.25)$) -- ($(x2) -(0.6-0.6+1,-0.25)$);
    \draw[thick] ($(x) -(-0.0-0.7+0.75,0.25)$) -- ($(x2) -(0.2-0.6+1,-0.25)$);
    \draw[thick] ($(x) -(-0.5-0.55+0.8,0.25)$) -- ($(x2) -(-0.-0.2-0.6+1.1,-0.25)$);
    \draw[thick] ($(x) -(-0.5-0.55+0.8,0.25)$) -- ($(x2) -(-0.6-0.45+1.1,-0.25)$);


  \end{tikzpicture}
  \caption{Illustration of \Cref{ex:MAandMean}. Shown are two time series $x^{(1)}$ and $x^{(2)}$, as well as the time series $\bar{x}$ obtained by the multiple alignment approach and a mean~$z$. Lines indicate optimal alignments between the time series. Note that $\dtw(x^{(1)},\bar{x})^2=\dtw(x^{(2)},\bar{x})^2=3.5$ whereas $\dtw(x^{(1)},z)^2=\dtw(x^{(2)},z)^2=3.25$.}
  \label{fig:ex1}
\end{figure}
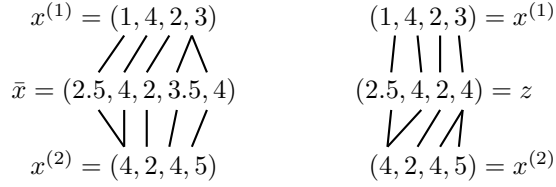

In \cref{ex:MAandMean}, the time series~$\bar{x}$ is also not an optimal choice among all time series of length five since also $z' = (2.5, 4, 4, 2, 4)$ satisfies~$F(z')=6.5 < F(\bar{x})$.
In fact, by computer-based exhaustive search we found that no warping path~$p\in\mathcal{P}_{4,4}$ yields a mean for~$x^{(1)}$ and~$x^{(2)}$ by averaging the aligned time series~$x_p^{(1)}$ and~$x_p^{(2)}$.
We conclude that a multiple alignment as defined by \citet[Definition~5]{PG12} that shall produce an averaged time series that minimizes the Fr\'echet function does not exist in general.
\Cref{ex:MAandMean} implies that incremental pairwise averaging strategies such as NLAAF~\cite{GMTS96} or PSA~\cite{NR09} are based on a wrong mean computation for two time series.

We finish with another erroneous example from the literature for three time series~\cite[Figure~2]{PG12}.
  \begin{example}\label{ex:PG_TCS}\normalfont
    For the three time series
    \begin{align*}
      x^{(1)} &= (1,10,0,0,4),\\
      x^{(2)} &= (0,2,10,0,0),\\
      x^{(3)} &= (0,0,0,10,0),
    \end{align*}
\noindent the multiple alignment is given as
    \begin{align*}
      x^{(1)'} &= (1,1,1,10,0,0,4),\\
      x^{(2)'} &= (0,0,2,10,0,0,0),\\
      x^{(3)'} &= (0,0,0,10,0,0,0),
    \end{align*}
    yielding the arithmetic mean~$\bar{x} = (\frac{1}{3}, \frac{1}{3},1,10,0,0,\frac{4}{3})$ with~$F(\bar{x}) = 14/3 \ge 4.66$.
    However, for~$z = (\frac{1}{4}, 1, 10, 0, \frac{4}{3})$, it holds~$F(z) \le 4.48 < F(\bar{x})$.
  \end{example}

\paragraph{Brute-Force Algorithm.}
  Another approach to solve \DTW is based on a brute-force algorithm~\cite{PKG11} working as follows: Suppose an optimal mean is of length~$m$.
  Consider for each input time series a partition into~$m$ consecutive non-empty parts and
  align the~$i$-th element in the mean with all elements in the~$i$-th part of each time series.
It is claimed that a mean can be found by trying out all possible partitions into~$m$ consecutive non-empty parts for each input time series.
  
  This approach is problematic since not all possible solutions are considered (as two elements in the mean can be aligned with the same element of an input time series).
  \Cref{ex:MAandMean} depicts this problem.

\section{Polynomial-time Solvability for Binary Data}
\label{sec:binaryDTW}
By restricting the values in the time series (input and mean) 
to be binary (0 or~1),
we arrive at the special case \BDTW{}.

\Optproblemdef{\BDTW}
{Sample $\S{X} = \args{x^{(1)}, \dots, x^{(k)}}$ of~$k$ time series with elements in~$\{0,1\}$ and rational nonnegative weights~$w_1,\ldots,w_k$.}
{Find a time series~$z\in\{0,1\}^*$ that minimizes $F_w(z)$.}

We prove that \BDTW is polynomial-time solvable.

\begin{theorem}
  \label{thm:bdtwpolytime}
  \BDTW for $k$~input time series is solvable with $\Oh(kn^3)$ arithmetic operations, where~$n$ is the maximum length of any input time series.
\end{theorem}

  To show polynomial-time solvability of \BDTW, we first prove some preliminary results about
the dtw-distance of binary time series and properties of a binary mean.
  We start with the following general definition.

  \begin{definition}
    A time series~$x=(x_1,\ldots, x_n)$ is \emph{condensed} if no two consecutive elements are equal, that is, $x_i\neq x_{i+1}$ holds for all~$i\in[n-1]$.
    We denote the \emph{condensation} of a time series~$x$ by~$\tilde{x}$ and define it
    to be the time series obtained by repeatedly removing one of two equal consecutive elements in~$x$ until the remaining series is condensed.
  \end{definition}

  The following proposition implies that a mean can always be assumed to be condensed.
  Note that this holds for arbitrary time series (not only for the binary case) but not for the
  constrained mean.

  \begin{proposition}
    \label{prop:condensed}
    Let~$x$ be a time series and let~$\tilde{x}$ denote its condensation.
    Then, for every time series~$y$, it holds that~$\dtw(\tilde{x},y) \le \dtw(x,y)$.
  \end{proposition}

  \begin{proof}
    Let~$y$ have length~$m$ and assume that~$x=(x_1,\ldots,x_n)$ is not condensed. Then, $x_i=x_{i+1}$ holds for some $i\in[n-1]$.
    Let~$p=((i_1,j_1),\ldots,(i_L,j_L))$ be an optimal warping path for~$x$ and~$y$.
    Now, consider the time series~$x'=(x_1,\ldots,x_i,x_{i+2},\ldots,x_n)$ that is obtained from~$x$ by deleting the element~$x_{i+1}$.
    We construct a warping path~$p'$ for~$x'$ and~$y$ such that~$C_{p'}(x',y) = C_p(x,y)$.
    To this end, let $p_a=(i_a,j_a)$, $2\le a \le L$, be the first index pair in~$p$ where~$i_a=i+1$ (hence, $i_{a-1}=i$). Now, we consider two cases.
    
    If $j_a = j_{a-1}+1$, then we define the order-$((n-1)\times m)$ warping path $$p':=((i_1,j_1),\ldots,(i_{a-1},j_{a-1}),(i_a-1,j_a),\ldots,(i_L-1,j_L))$$
    of length~$L$.
    Thus, each element of~$y$ that was aligned to~$x_{i+1}$ in~$p$ is now aligned to~$x_i$ instead.
    To check that~$p'$ is a valid warping path, note first that~$(i_1,j_1)=(1,1)$ and
    $(i_L-1,j_L)=(n-1,m)$ holds since~$p$ is a warping path.
    Also, it holds
    \begin{align*}
      \forall 1\le \ell \le a-2 \colon  &(i_{\ell+1},j_{\ell+1})-(i_\ell,j_\ell)\in\{(1,0),(0,1),(1,1)\},\\
      &(i_a-1,j_a) - (i_{a-1},j_{a-1}) = (0,1),\\
      \forall a \le \ell \le L-1 \colon &(i_{\ell+1}-1,j_{\ell+1})-(i_\ell-1,j_\ell)\in\{(1,0),(0,1),(1,1)\}.
    \end{align*}
    The cost of~$p'$ is
    \begin{align*}
      C_{p'}(x',y) &= \sum_{\ell=1}^{a-1} (x'_{i_\ell} - y_{j_\ell})^2 + \sum_{\ell=a}^L (x'_{(i_\ell-1)} - y_{j_\ell})^2\\
      &= \sum_{\ell=1}^{a-1} (x_{i_\ell} - y_{j_\ell})^2 + \sum_{\ell=a}^L (x_{i_\ell} - y_{j_\ell})^2 = C_p(x,y).
    \end{align*}
    
    Otherwise, if $j_a=j_{a-1}$, then we define the warping path
    $$p':=((i_1,j_1),\ldots, (i_{a-1},j_{a-1}),(i_{a+1}-1,j_{a+1}),\ldots,(i_L-1,j_L))$$
    of length~$L-1$.
    Again,~$(i_1,j_1)=(1,1)$ and $(i_L-1,j_L)=(n-1,m)$ holds since~$p$ is a warping path.
    Clearly, also
    \begin{align*}
      \forall 1\le \ell \le a-2 \colon  &(i_{\ell+1},j_{\ell+1})-(i_\ell,j_\ell)\in\{(1,0),(0,1),(1,1)\} \text{ and}\\
      \forall a+1 \le \ell \le L-1 \colon  &(i_{\ell+1}-1,j_{\ell+1})-(i_\ell-1,j_\ell)\in\{(1,0),(0,1),(1,1)\}
    \end{align*}
    holds. Finally, we have
    $(i_{a+1}-1,j_{a+1}) - (i_{a-1},j_{a-1}) \in \{(1,0),(0,1),(1,1)\}$
    since $i_{a+1}-1 - i_{a-1}= i_{a+1} - i_{a}$ holds and also $j_{a+1}-j_{a-1} = j_{a+1} - j_a$ holds.
    Thus, $p'$ is a valid warping path and its cost is
    \begin{align*}
      C_{p'}(x',y) &= \sum_{\ell=1}^{a-1} (x'_{i_\ell} - y_{j_\ell})^2 + \sum_{\ell=a+1}^L (x'_{(i_\ell-1)} - y_{j_\ell})^2\\
      &= \sum_{\ell=1}^{a-1} (x_{i_\ell} - y_{j_\ell})^2 + \sum_{\ell=a+1}^L (x_{i_\ell} - y_{j_\ell})^2 \\
      &= C_p(x,y) - (x_{i_a}-y_{j_a})^2.
    \end{align*}
    
    Since in both cases above, the cost does not increase, we obtain
    \begin{align*}
    \dtw(x',y) \leq C_{p'}(x',y) \le C_p(x,y) = \dtw(x,y).
    \end{align*} 
    Repeating this argument until~$x'$ is condensed finishes the proof.
  \end{proof}

  \Cref{prop:condensed} implies that we can assume a mean to be condensed.
  Next, we want to prove an upper bound on the length of a binary mean.
  To this end, we analyze the dtw-distances of binary time series.
  Note that a binary condensed time series is fully determined by its first element and its length.
  We use this property to give a closed expression for the dtw-distance of two condensed binary time series.

  \begin{lemma}
    \label{lem:condensed_binary}
    Let~$x=(x_1,\ldots,x_n),\; y=(y_1,\ldots,y_m)\in\{0,1\}^*$ be two condensed binary time series
    with~$n \ge m$.
    Then, it holds \[\dtw(x,y)^2 = \begin{cases}\lceil(n-m)/2\rceil,& x_1 = y_1\\
      2,& x_1\neq y_1 \wedge n=m\\
      1+\lfloor(n-m)/2\rfloor,& x_1\neq y_1 \wedge n > m
    \end{cases}.\]
  \end{lemma}

  \begin{proof}
    We prove the statement by first giving a warping path that has the claimed cost and second proving that every warping path has at least the claimed cost.
    
    ``$\le$'':
    We show that there exists a warping path~$p$ between~$x$ and~$y$ that has the claimed cost.
    The warping path~$p$ is defined as follows:
    
    If $x_1=y_1$, then we have~$x_i=y_i$ for all~$i\in[m]$ (since~$x$ and~$y$ are condensed and binary) and we set \[p := ((1,1),(2,2),\ldots,(m,m),(m+1,m),\ldots,(n,m)).\]
    This warping path has cost~$C_p(x,y) = \sum_{i=m+1}^n|x_i-y_m| = \lceil(n-m)/2\rceil$.

    If~$x_1\neq y_1$, then we have~$x_i = y_{i-1}$ for all~$2\le i \le m$.
    Thus, for~$n=m$, the warping path \[p:=((1,1),(2,1),(3,2),\ldots,(n,m-1),(n,m))\]
    has cost~$C_p(x,y) = |x_1-y_1| + |x_n-y_m| = 2$.
    Finally, for~$n > m$, the warping path $p:=((1,1),(2,1),(3,2),\ldots,(m+1,m),(m+2,m),\ldots,(n,m))$
    yields cost
    \begin{align*}
      C_p(x,y) &= |x_1 - y_1| + \sum_{i=2}^{m+1}|x_i-y_{i-1}| + \sum_{i=m+2}^n|x_i-y_m|\\
	      &= 1 + \sum_{i=m+2}^n|x_i-y_m|
	      = 1 + \lceil(n-m-1)/2\rceil
	      = 1 + \lfloor(n-m)/2\rfloor.
    \end{align*}
    
    ``$\ge$'': We show that every warping path has at least the cost claimed above.
    Consider an optimal warping path~$p=(p_1,\ldots,p_L)$ for~$x$ and~$y$ and note that there are at least~$n-m$ different indices $\ell\in[L-1]$ such that~$p_{\ell+1}-p_\ell=(1,0)$ since~$n \ge m$. For every such pairs $p_{\ell+1}=(i_{\ell+1},j_{\ell+1})$, $p_\ell = (i_\ell,j_\ell)$ with~$j_{\ell+1}=j_\ell$, we have
    \[|x_{i_{\ell+1}} - y_{j_{\ell+1}}| + |x_{i_\ell} - y_{j_\ell}| = |x_{i_{\ell+1}} - y_{j_\ell}| + |x_{i_\ell} - y_{j_\ell}| = 1,\]
    since~$x_{j_{\ell+1}} \neq x_{j_\ell}$ (recall that $x$ is condensed and binary). Hence, at least for every second such index~$\ell$ (starting from the first one) a cost of~1 is induced.
    Hence, $\dtw(x,y)^2\ge \lceil(n-m)/2\rceil$.
    If~$x_1 = y_1$, then this lower bound matches the claimed cost.

    If~$x_1\neq y_1$ and~$n=m$, then also~$x_n\neq y_m$ and hence the cost is at least~2.

    If~$x_1\neq y_1$ and~$n>m$, then we can assume that~$p_2=(2,1)$.
    To see this, note that for~$p_2\neq (2,1)$ the cost is at least~$1+\lceil(n-m)/2\rceil$ by the above argument.
    If~$p_2=(2,1)$, then the subpath~$(p_2,p_3,\ldots,p_L)$ of~$p$ is an optimal warping path between~$(x_2,\ldots,x_n)$ and~$y$, where~$x_2=y_1$ and~$n-1\ge m$. As we have already shown above,
    this path has cost~$\lceil(n-1-m)/2\rceil$. Hence, in this case~$p$ has cost
    $1+\lceil(n-1-m)/2\rceil = 1 + \lfloor(n-m)/2\rfloor$.
    Thus, we can assume that~$p_2=(2,1)$ in which case the cost of~$p$ matches the claimed
    cost of the lemma. This finishes the proof.
  \end{proof}

  Note that according to~\Cref{lem:condensed_binary}, for a fixed condensed binary time series~$y$ of length~$m$, the value~$\dtw(x,y)^2$ is monotonically increasing in the length of~$x$ for all condensed binary time series $x$ of length~$n \ge m+1$.
  We use this property later in the proof of~\Cref{lem:bounded_length} where we derive an upper bound
  on the length of a binary mean.
  In order to prove \Cref{lem:bounded_length}, we also need the following lemma concerning the dtw-distances between condensed and non-condensed time series.

  \begin{lemma}
    \label{lem:equal_dist}
    Let~$x=(x_1,\ldots,x_n)$ be a condensed binary time series
    and let~$y=(y_1,\ldots,y_m)\in\{0,1\}^*$ with~$n \ge m$.
    Then, for the condensation~$\tilde{y}$ of~$y$ it holds $\dtw(x,y)^2=\dtw(x,\tilde{y})^2$.
  \end{lemma}

  \begin{proof}
    Assume that~$y$ is not condensed.
    Then, $y$ consists of~$\ell\in[m]$ \emph{blocks}, where a block is a maximal subsequence of consecutive 0's or consecutive 1's in~$y$. Let~$m_1,\ldots,m_\ell$ denote the lengths of these blocks where~$m_1+\ldots +m_\ell=m$. Note also that~$\tilde{y}$ has length~$\ell$ with~$\ell<m\le n$.
    We define a warping path~$p$ between~$x$ and~$y$ such that
    $C_p(x,y)=\dtw(x,\tilde{y})^2$.
    Note that, by~\Cref{lem:condensed_binary}, we have
    \[ \dtw(x,\tilde{y})^2 = \begin{cases}\lceil(n-\ell)/2\rceil,& x_1=\tilde{y}_1\\1+\lfloor(n-\ell)/2\rfloor,& x_1\neq \tilde{y}_1\end{cases}.\]
    
    If~$x_1 = y_1$, then we set $p:=((1,1),\ldots,(1,m_1),(2,m_1+1),\ldots,(2,m_1+m_2),\ldots,(\ell,m),(\ell+1,m),\ldots,(n,m))$
    and obtain cost~$C_p(x,y) = \sum_{i=\ell+1}^n|x_i-y_m|=\lceil(n-\ell)/2\rceil$.
    
    If~$x_1 \neq y_1$, then we set $p:=((1,1),(2,1),\ldots,(2,m_1),(3,m_1+1),\ldots,(3,m_1+m_2),\ldots,(\ell+1,m),(\ell+2,m),\ldots,(n,m))$
    and obtain cost~$$C_p(x,y) = 1+\sum_{i=\ell+2}^n|x_i-y_m|=1+\lfloor(n-\ell)/2\rfloor.$$
  \end{proof}

  We now have all ingredients to show that there always exists a binary mean of length at most one larger than the maximum length of any input time series.

  \begin{lemma}
    \label{lem:bounded_length}
    For binary input time series~$x^{(1)},\ldots,x^{(k)}\in\{0,1\}^*$ of maximum length~$n$, there exists a binary mean~$z\in\{0,1\}^*$ of length at most~$n+1$.
  \end{lemma}

  \begin{proof}
    Assume that~$z=(z_1,\ldots,z_m)\in\{0,1\}^*$ is a mean of length~$m > n+1$.
    By \Cref{prop:condensed}, we can assume that~$z$ is condensed, that is,~$z_i\neq z_{i+1}$ for all~$i\in[m-1]$.
    We claim that~$z':=(z_1,\ldots,z_{n+1})$ is also a mean.
    We prove this claim by showing that~$\dtw(z',x^{(i)})^2\le\dtw(z,x^{(i)})^2$ holds for all~$i\in[k]$.
    By \Cref{lem:equal_dist,lem:condensed_binary}, we have $\dtw(z',x^{(i)})^2=\dtw(z',\tilde{x}^{(i)})^2 \le \dtw(z,\tilde{x}^{(i)})^2=\dtw(z,x^{(i)})^2$, where the inequality follows from \Cref{lem:condensed_binary} since~$z'$ is of length~$n+1<m$ and the dtw-distance is monotonically increasing.
  \end{proof}
  
  Having established that a binary mean can always be assumed to be condensed and of bounded length,
  we now show that it can be found in polynomial time.

  \begin{proof}[Proof of~Theorem~\ref{thm:bdtwpolytime}]
    By~\Cref{prop:condensed} and \Cref{lem:bounded_length}, we can assume the desired mean~$z\in\{0,1\}^*$ to be a condensed series of length at most~$n+1$. Thus, there are at most~$2n+2$ many possible candidates for~$z$.
    For each candidate~$z$, we can compute the value~$F_w(z)$ with~$O(kn^2)$ arithmetic operations and select the one with the smallest value. 
    Overall, this yields~$O(kn^3)$ arithmetic operations.
  \end{proof}

\section{An Exact Algorithm Solving \DTW}
\label{sec:exact}
We develop a nontrivial exponential-time algorithm solving \wDTW exactly.
The key is to observe a certain structure of a mean and the corresponding alignments
to the input time series.
To this end, we define redundant elements in a mean. Note that this concept was already used
by~\citet[Theorem~2.7]{JS18} in order to prove the existence of a mean of bounded length (though \cite[Definition~3.20]{JS18} is slightly different).

\begin{definition}\label{def:redundant}
  Let~$x^{(1)},\ldots,x^{(k)}$ and $z=(z_1,\ldots,z_m)$ be time series and let~$p^{(j)}$, $j\in[k]$, denote an optimal warping path between $x^{(j)}$ and~$z$. 
  We call an element~$z_i$ of~$z$ \emph{redundant} if in every time series~$x^{(j)}$, $j\in[k]$, there exists an element that is aligned with~$z_i$ and with another element of~$z$ by~$p^{(j)}$.
\end{definition}

The next lemma states that there always exists a mean without redundant elements (similarly to \cite[Theorem~2.7]{JS18}).

\begin{lemma}
  \label{lem:no_redundant}
  There exist a mean~$z$ for time series $x^{(1)},\ldots,x^{(k)}$ and optimal warping paths~$p^{(j)}$ between~$z$ and~$x^{(j)}$ for each~$j\in[k]$ such that~$z$ contains no redundant element.
\end{lemma}

  \begin{proof}
    Let~$z$ be a mean of $x^{(1)},\ldots,x^{(k)}$ with optimal warping paths~$p^{(j)}$, $j\in[k]$, such that the element~$z_i$ is redundant (recall that a mean~$z$ always exists).
    We show that there also exists a mean~$z'$ and optimal warping paths~$p^{(j)'}$ such that no element in~$z'$ is redundant.

    Assume first that there exists a $j\in[k]$ such that the element $z_i$ is aligned by~$p^{(j)}$ with at least one element~$x_\ell^{(j)}$ in~$x^{(j)}$ that is not aligned with any other element in~$z$.
    Let~$L$ denote the length of~$p^{(j)}$.
    Then, $p^{(j)}$ is of the form $$p^{(j)}=(p_1,\ldots,(i-1,\ell_{t-1}),(i,\ell_t),\ldots,(i,\ell_{t+\alpha}),(i+1,\ell_{t+\alpha+1}),\ldots,p_L)$$
    for some~$\ell_t \le \ell \le \ell_{t+\alpha}$ and $\alpha \ge 1$. Since~$z_i$ is redundant, it follows that~$\ell_{t-1} = \ell_t$ or~$\ell_{t+\alpha}=\ell_{t+\alpha+1}$ holds.
    If~$\ell_{t-1}=\ell_t$, then we remove the pair~$(i,\ell_t)$ from~$p^{(j)}$.
    Also, if~$\ell_{t+\alpha}=\ell_{t+\alpha+1}$, then we remove the pair~$(i,\ell_{t+\alpha})$ from~$p^{(j)}$.
    Note that this yields a warping path~$p^{(j)'}$ between~$z$ and~$x^{(j)}$ since even if we removed both pairs~$(i,\ell_t)$ and~$(i,\ell_{t+\alpha})$, then we know by assumption that there still exists the pair~$(i,\ell)$ with $\ell_t < \ell < \ell_{t+\alpha}$ in~$p^{(j)'}$ since~$z_i$ is aligned with~$x_\ell^{(j)}$ which is not aligned with another element in~$z$.
    Since we only removed pairs from~$p^{(j)}$, it holds
    $C_{p^{(j)'}}(z,x^{(j)}) \le C_{p^{(j)}}(z,x^{(j)})$.
    Moreover, $z_i$ is not redundant anymore.

    Now, assume that for all~$j\in[k]$, $z_i$ is aligned only with elements in~$x^{(j)}$ which are also aligned with another element of~$z$ by~$p^{(j)}$ (that is, $z_i$ is redundant according to~\cite[Definition~3.20]{JS18}).
    Let~$z'$ denote the time series obtained by deleting the element~$z_i$ from~$z$.
    \citet[Proof of Theorem~2.7]{JS18} showed that in this case~$\dtw(z',x^{(j)}) \le \dtw(z,x^{(j)})$ holds for all~$j\in[k]$.

    Under both assumptions above, we reduced the number of redundant elements. Hence, we can repeat the above arguments until we obtain a mean~$z'$ without redundant elements.
  \end{proof}

  \Cref{lem:no_redundant} allows us to devise a dynamic program computing a mean without redundant elements.
  We compute a mean by testing all possibilities to align the last mean element to elements from the input time series while recursively adding an optimal solution for the remaining non-aligned elements in the input time series (see \Cref{fig:dynprog}).
  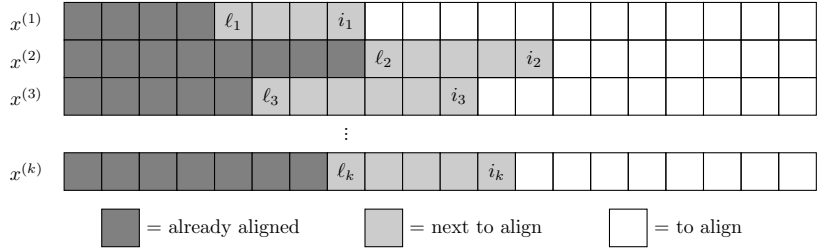
\begin{figure}
   \centering
   \begin{tikzpicture}

    \def\xw{width("1")+6pt}
    \def\xs{0.5}
    \def\yw{width("1")+6pt}
    \def\ys{1.6}
    \def\xsh{4}

    \tikzstyle{dnode}=[ultra thin,draw,minimum width=\xw,minimum height=\yw]
    \tikzstyle{dedge}=[thick]

    \def\ys{-0.5}
    \newcommand{\colarr}[5]{
      \node (x) at (-0.5*\xs,#4)[scale=3/4]{$x^{(#5)}$};
      \foreach \x in {1,2,...,#1}{
	      \ifnum\x <#2
		      \draw[fill=gray!100!white] (\x*\xs-0.5*\xs,#4-0.5*\xs) rectangle (\x*\xs+0.5*\xs,#4+0.5*\xs);
	      \else{}
		      \ifnum\x<#3\draw[fill=gray!40!white] (\x*\xs-0.5*\xs,#4-0.5*\xs) rectangle (\x*\xs+0.5*\xs,#4+0.5*\xs);
		      \ifnum\x=#2\node (x) at (\x*\xs,#4)[scale=3/4]{$\ell_{#5}$};\fi{}
		      \else{}
			      \ifnum\x=#3\node (x) at (\x*\xs-\xs,#4)[scale=3/4]{$i_{#5}$};\fi{}
			      \draw[] (\x*\xs-0.5*\xs,#4-0.5*\xs) rectangle (\x*\xs+0.5*\xs,#4+0.5*\xs);
		      \fi{}
	      \fi{}
      }
     }

    \colarr{20}{5}{9}{0*\ys}{1};
    \colarr{20}{9}{14}{1*\ys}{2};
    \colarr{20}{6}{12}{2*\ys}{3};
    \node (ld) at (4,2.9*\ys)[scale=0.5]{\bf $\vdots$};
    \colarr{20}{8}{13}{4*\ys}{k};

    \draw[fill=gray!100!white] (2*\xs-0.5*\xs,5.5*\ys-0.5*\xs) rectangle (2*\xs+0.5*\xs,5.5*\ys+0.5*\xs);
    \node at (2*\xs+0.5*\xs,5.5*\ys)[right,scale=3/4]{= already aligned};
    \draw[fill=gray!40!white] (9*\xs-0.5*\xs,5.5*\ys-0.5*\xs) rectangle (9*\xs+0.5*\xs,5.5*\ys+0.5*\xs);
    \node at (9*\xs+0.5*\xs,5.5*\ys)[right,scale=3/4]{= next to align};
    \draw[fill=gray!0!white] (15.5*\xs-0.5*\xs,5.5*\ys-0.5*\xs) rectangle (15.5*\xs+0.5*\xs,5.5*\ys+0.5*\xs);
    \node at (15.5*\xs+0.5*\xs,5.5*\ys)[right,scale=3/4]{= to align};

    \end{tikzpicture}
    \caption{Illustration of the dynamic program
      computing a mean~$(z_1,\ldots,z_q)$ for the subseries $(x_1^{(1)},\ldots,x_{i_1}^{(1)}),\ldots,(x_1^{(k)},\ldots,x_{i_k}^{(k)})$ (white elements are not considered at the current iteration).
      The light gray elements are aligned to the last mean element~$z_q$.
      The remaining mean elements $(z_1,\ldots,z_{q-1})$ form an optimal mean for the dark gray subseries.}
    \label{fig:dynprog}
  \end{figure}  
  We use the assumption that the mean does not contain redundant elements for this recursive approach.

  Before describing our dynamic program, we prove the following lemma concerning the optimal value of a mean element for given alignments.

  \begin{lemma}\label{lem:opt_mean_value}
    Let~$x^{(1)},\ldots,x^{(k)}$ be time series and let~$w_1,\ldots,w_k$ be nonnegative weights.
    Let~$n_j$ denote the length of~$x^{(j)}$, $j\in[k]$.
    Further, let~$p^{(j)}$ be a warping path of order~$m\times n_j$ for $m\in\N$ and let~$z = \argmin_{x\in\T_m}\sum_{j=1}^kw_jC_{p^{(j)}}(x,x^{(j)})$.
    For $i\in[m]$, let~$x^{(j)}_{\ell_{ij}},\ldots,x^{(j)}_{h_{ij}}$ denote the elements of~$x^{(j)}$ that are aligned with element~$z_i$ by~$p^{(j)}$.
    Then,
    \[z_i = \frac{\sum_{j=1}^kw_j\sum_{t=\ell_{ij}}^{h_{ij}}x^{(j)}_t}{\sum_{j=1}^{k}w_j(h_{ij}-\ell_{ij}+1)}.\]
    
  \end{lemma}

  \begin{proof}
    By the assumption of the lemma, we have
    \begin{align*}
      z = \argmin_{x\in\T_m}\sum_{j=1}^kw_jC_{p^{(j)}}(x,x^{(j)}) &= \argmin_{x\in\T_m}\sum_{j=1}^kw_j\sum_{i=1}^m\sum_{t=\ell_{ij}}^{h_{ij}}(x_i-x^{(j)}_t)^2\\
      &= \argmin_{x\in\T_m}\sum_{i=1}^m\sum_{j=1}^kw_j\sum_{t=\ell_{ij}}^{h_{ij}}(x_i-x^{(j)}_t)^2.
    \end{align*}
    Hence, for each~$i\in[m]$, it holds
    \begin{align*}
      z_i = \argmin_{\mu\in\Q}\sum_{j=1}^kw_j\sum_{t=\ell_{ij}}^{h_{ij}}(\mu-x^{(j)}_t)^2.
    \end{align*}
    Note that the above sum is a convex function in~$\mu$ (since all weights are nonnegative).
    Setting the first derivative with respect to~$\mu$ equal to zero yields
    \begin{align*}
      & &\sum_{j=1}^k\left(2w_j(h_{ij}-\ell_{ij}+1)z_i -2w_j\sum_{t=\ell_{ij}}^{h_{ij}}x^{(j)}_t\right) &= 0\\
      \Leftrightarrow & &2z_i\sum_{j=1}^k(w_j(h_{ij}-\ell_{ij}+1)) &= 2\sum_{j=1}^kw_j\sum_{t=\ell_{ij}}^{h_{ij}}x^{(j)}_t\\
      \Leftrightarrow & &z_i &= \frac{\sum_{j=1}^kw_j\sum_{t=\ell_{ij}}^{h_{ij}}x^{(j)}_t}{\sum_{j=1}^{k}w_j(h_{ij}-\ell_{ij}+1)}
    \end{align*}
    
  \end{proof}
  
  We now prove our main theorem.
  
  \begin{theorem}
    \label{thm:nto2kalgo}
    \wDTW for~$k$ input time series is solvable with $O(n^{2k+1}2^kk)$ arithmetical operations, where~$n$ is the maximum length of any input time series.
  \end{theorem}
  
  \begin{proof}
    Assume for simplicity that all time series have length~$n$ (the general case can be solved analogously). We find a mean using a dynamic programming approach.
    Let~$C$ be a~$k$-dimensional table, where for all~$(i_1,\ldots,i_k)\in[n]^k$, we define
    \[C[i_1,\ldots,i_k] = \min_{z \in \mathcal{T}}\left(\sum_{j=1}^kw_j\left(\dtw(z,(x_1^{(j)},\ldots,x_{i_j}^{(j)}))\right)^2\right),\]
    that is, $C[i_1,\ldots,i_k]$ is the value~$F_w(z)$ of the weighted Fr\'echet function of a mean~$z$ for the subseries~$(x_1^{(1)},\ldots,x_{i_1}^{(1)}),\ldots,(x_1^{(k)},\ldots,x_{i_k}^{(k)})$.
    Clearly, $C[n,\ldots,n]$ is the optimal value~$F_w(z)$ of a mean for the input instance.

    For $i_1=i_2=\ldots =i_k=1$, a mean~$z$ clearly contains just one element
    and each optimal warping path between~$z$ and~$(x_1^{(j)})$ trivially equals~$((1,1))$.
    By \Cref{lem:opt_mean_value}, we initialize
    \[C[1,\ldots,1] = \sum_{j=1}^kw_j(x_1^{(j)} - \mu)^2,\qquad \mu= \sum_{j=1}^kw_jx_1^{(j)}.\]
    Hence, the corresponding mean is~$z=(\mu)$.

    For the case that $i_j > 1$ holds for at least one~$j\in[k]$, consider a mean~$z$ for
    $(x_1^{(1)},\ldots,x_{i_1}^{(1)}),\ldots,(x_1^{(k)},\ldots,x_{i_k}^{(k)})$. By \Cref{lem:no_redundant},
    we can assume that there exist optimal warping paths~$p^{(j)}$ between~$z$ and~$(x_1^{(j)},\ldots,x_{i_j}^{(j)})$ such that~$z$ contains no redundant elements.
    Let~$z_q$ be the last element of~$z$. Then, for each~$j\in[k]$, $z_q$ is aligned by~$p^{(j)}$
    with some elements $x_{\ell_j}^{(j)},\ldots,x_{i_j}^{(j)}$ for $\ell_j\in[i_j]$.
    By \Cref{lem:opt_mean_value}, it follows that
    \[z_q = \frac{\sum_{j=1}^kw_j\sum_{t=\ell_j}^{i_j}x^{(j)}_t}{\sum_{j=1}^{k}w_j(i_j-\ell_j+1)}.\]
    Hence, the contribution of~$z_q$ to~$F_w(z)$ is
    $\sum_{j=1}^kw_j\sum_{t=\ell_j}^{i_j}(x^{(j)}_t - z_q)^2$.

    Now, assume that there exists another element~$z_{q-1}$ in~$z$. Clearly, for each~$j\in[k]$, $z_{q-1}$ is aligned only with elements of indices up to~$\ell_j$ since otherwise the warping path conditions are violated.
    Hence,~$F_w(z)$ can be obtained recursively from a mean of the subseries
    $(x_1^{(1)},\ldots,x_{\ell_1}^{(1)}),\ldots,(x_1^{(k)},\ldots,x_{\ell_k}^{(k)})$.
    Recall, however, that we assumed~$z$ not to contain any redundant element. 
    It follows that $z_{q-1}$ cannot be aligned with~$x_{\ell_j}^{(j)}$ for all~$j\in[k]$ since~$z_q$ is already aligned with each~$x_{\ell_j}^{(j)}$.
    Hence, we add the minimum value~$C[\ell_1',\ldots,\ell_k']$ over all~$\ell_j'\in\{\max(1,\ell_j-1),\ell_j\}$, where~$\ell_j' = \ell_j-1$ holds for at least one~$j\in[k]$.
    
    We arrive at the following recursion:
    \[C[i_1,\ldots,i_k] = \min\{c^*(\ell_1,\ldots,\ell_k) + \sigma(\ell_1,\ldots,\ell_k) \mid \ell_1\in[i_1],\ldots,\ell_k\in[i_k]\},\]
    where 
    \[
      \sigma(\ell_1,\ldots,\ell_k) := \sum_{j=1}^kw_j\sum_{t=\ell_j}^{i_j} (x_t^{(j)} - \mu)^2,\qquad
\mu := \frac{\sum_{j=1}^kw_j\sum_{t=\ell_j}^{i_j}x^{(j)}_t}{\sum_{j=1}^{k}w_j(i_j-\ell_j+1)},
    \]
    and
    \[c^*(\ell_1,\ldots,\ell_k) := \min\{C[\ell_1',\ldots,\ell_k'] \mid \ell_j'\in\{\max(1,\ell_j-1), \ell_j\},\, \sum_{j=1}^k(\ell_j - \ell_j')>0\}\]
    if $\ell_j > 1$ holds for some~$j\in[k]$, and~$c^*(1,\ldots,1):=0$.
    
      In order to compute~$C[i_1,\ldots,i_k]$, a minimum is computed over all possible choices~$\ell_j\in[i_j]$, $j\in[k]$. For each choice~$\ell_1,\ldots,\ell_k$, the value~$\mu$ corresponds to the last element of a mean and $\sigma(\ell_1,\ldots,\ell_k)$ is the induced cost of aligning this element with~$x_{\ell_j}^{(j)},\ldots,x_{i_j}^{(j)}$ for each~$j\in[k]$.
      The value~$c^*(\ell_1,\ldots,\ell_k)$ yields the value~$F_w(z')$ of a mean~$z'$ for the remaining subseries~$(x_1^{(j)},\ldots,x_{\ell_j'}^{(j)})$, $j\in[k]$, over all~$\ell_j'\in\{\max(1,\ell_j-1),\ell_j\}$ such that $\sum_{j=1}^k(\ell_j-\ell_j')>0$, which implies that~$\ell_j'=\ell_j-1$ holds for at least one~$j\in[k]$. This condition guarantees that we only find alignments which do not yield redundant elements in the mean (which we can assume by~\Cref{lem:no_redundant}). Note that~$\ell_j=1$ implies that~$\ell_j'=1$ (since index~$0$ does not exist in~$x^{(j)}$).

    The dynamic programming table~$C$ can be filled iteratively along the dimensions starting from $C[1,\ldots,1]$.
    The overall number of entries is~$n^k$. For each table entry, the minimum of a set containing~$O(n^k)$ elements is computed. Computing an element requires the computation of~$\sigma(\ell_1,\ldots,\ell_k)$ which can be done with~$O(kn)$ arithmetical operations plus the computation of~$c^*(\ell_1,\ldots,\ell_k)$ which is the minimum of a set of size at most~$2^k$ whose elements can be obtained by constant-time table look-ups. Thus, $C$ can be filled using~$O(n^k \cdot n^k \cdot2^k\cdot kn)$ arithmetical operations.
    A mean can be obtained by storing the values~$\mu$ for which the minimum in the above recursion is attained (\Cref{algo:DTWMean} contains the pseudocode).
\end{proof}
    
\begin{algorithm2e*}[t]
    \caption{Exact Dynamic Program (EDP) for \wDTW}
  \label{algo:DTWMean}
  \normalsize
  \DontPrintSemicolon
  \SetKwFunction{append}{append}
  \SetKw{and}{and}
  \KwIn{Time series $x^{(1)},\ldots,x^{(k)}$ of length~$n$ and weights~$w_1,\ldots,w_k>0$.}
  \KwOut{Mean~$z$ and~$F_w(z)$.}

  Initialize $C$ \tcp*{\small $k$-dimensional DP table storing $F_w$-values}
  Initialize $Z$ \tcp*{\small $k$-dimensional table storing means}
  \ForEach(\tcp*[f]{\small fill tables iteratively}){$(i_1,\ldots,i_k)\in [n]^k$}{
    $C[i_1,\ldots,i_k] := \infty$\;
    $Z[i_1,\ldots,i_k] := ()$\;
    \ForEach(\tcp*[f]{\small compute $C[i_1,\ldots,i_k]$}){$(\ell_1,\ldots,\ell_k)\in[i_1]\times\ldots\times[i_k]$}{
      $\mu := (\sum_{j=1}^kw_j\sum_{t=\ell_j}^{i_j} x_t^{(j)})/(\sum_{j=1}^k w_j(i_j - \ell_j+1))$\;
      $\sigma := \sum_{j=1}^kw_j\sum_{t=\ell_j}^{i_j} (x_t^{(j)} - \mu)^2$\;
      $c^* := \infty$\;
      $z := ()$\;
      \eIf{$\ell_1 = \ell_2 = \ldots = \ell_k = 1$}{
        $c^* := 0$\;
      }(\tcp*[f]{\small compute~$c^*(\ell_1,\ldots,\ell_k)$ based on table look-ups}){        \ForEach{$(\ell_1',\ldots,\ell_k')\in\{\ell_1-1,\ell_1\}\times\ldots\times\{\ell_k-1,\ell_k\}$}{
          \If{$\forall j\in[k]:\ell_j'\ge 1$ \and $\exists j\in[k]:\ell_j'<\ell_j$}{
            \If{$C[\ell_1',\ldots,\ell_k'] < c^*$}{
              $c^* := C[\ell_1',\ldots,\ell_k']$\;
              $z := Z[\ell_1',\ldots,\ell_k']$\;
            }
          }
        }
      }
      \If(\tcp*[f]{\small update mean and $F_w$-value}){$c^* + \sigma < C[i_1,\ldots,i_k]$}{
        $C[i_1,\ldots,i_k] := c^* + \sigma$\;
        $Z[i_1,\ldots,i_k] := \append(z,\mu)$\;
      }
    }
  }
  \Return $(Z[n,\ldots,n], C[n,\ldots,n])$\;
\end{algorithm2e*}
  
  \noindent We close with some remarks on the above result.
  
\begin{compactitem}
  \item The dynamic program also allows to compute all (non-redundant) means by storing all possible values for which the minimum in the recursion is attained.
  \item It is possible to incorporate a fixed length~$q$ into the dynamic program (by adding another dimension to the table~$C$) such that it outputs only optimal solutions of length~$q$.\footnote{Source code available at http://www.akt.tu-berlin.de/menue/software/.}
    The running time increases by a factor of~$q^2$.
  \item The dynamic program can easily be extended to multivariate time series with elements in~$\Q^d$
    and a cost function $C_p(x,y):=\sum_{\ell=1}^L \norm{x_{i_\ell} - y_{j_\ell}}_2^2$
    (with a running time increase by a factor of~$d$).
\end{compactitem}

\section{Experiments}
\label{sec:exp}

The goal of this section is twofold: First, we empirically study properties of a mean that are relevant for mean-based applications in data mining as well as for devising heuristics for mean computation. Second, we assess the performance of state-of-the-art heuristics.

\subsection{Test Data}

\begin{table}
\small
\centering
\begin{tabular}{lrrl}
\toprule
Data Set & \multicolumn{1}{c}{\#} & \multicolumn{1}{c}{$n$} & Type\\ 
\midrule
ItalyPowerDemand & 1096 & 24 & SENSOR\\ 
SyntheticControl & 600 & 60 & SIMULATED\\ 
SonyAIBORobotSurface2 & 980 & 65 & SENSOR\\ 
SonyAIBORobotSurface1 & 621 & 70 & SENSOR\\ 
ProximalPhalanxTW & 605 & 80 & IMAGE\\ 
ProximalPhalanxOutlineCorrect & 891 & 80 & IMAGE\\ 
ProximalPhalanxOutlineAgeGroup & 605 & 80 & IMAGE\\ 
PhalangesOutlinesCorrect & 2658 & 80 & IMAGE\\ 
MiddlePhalanxTW & 553 & 80 & IMAGE\\ 
MiddlePhalanxOutlineCorrect & 891 & 80 & IMAGE\\ 
MiddlePhalanxOutlineAgeGroup & 554 & 80 & IMAGE\\ 
DistalPhalanxTW & 539 & 80 & IMAGE\\ 
DistalPhalanxOutlineCorrect & 876 & 80 & IMAGE\\ 
DistalPhalanxOutlineAgeGroup & 539 & 80 & IMAGE\\ 
TwoLeadECG & 1162 & 82 & ECG\\ 
MoteStrain & 1272 & 84 & SENSOR\\ 
ECG200 & 200 & 96 & ECG\\ 
MedicalImages & 1141 & 99 & IMAGE\\ 
TwoPatterns & 5000 & 128 & SIMULATED\\ 
SwedishLeaf & 1125 & 128 & IMAGE\\ 
CBF & 930 & 128 & SIMULATED\\ 
FacesUCR & 2250 & 131 & IMAGE\\ 
FaceAll & 2250 & 131 & IMAGE\\ 
ECGFiveDays & 884 & 136 & ECG\\ 
ECG5000 & 5000 & 140 & ECG\\ 
Plane & 210 & 144 & SENSOR\\ 
GunPoint & 200 & 150 & MOTION\\
\bottomrule
\end{tabular}
\caption{List of~$27$ UCR time series data sets. 
Columns~\# and~$n$ show the number and length of time series, respectively. 
The last column refers to the respective application domains (e.g.~ECG stands for electrocardiography).}
\label{tab:ucr}
\end{table}

We used data sets derived from random walks and the 27 UCR data sets~\cite{Chen2015} listed in \Cref{tab:ucr}. Due to the  running time of \Cref{algo:DTWMean}, only data sets with short time series were considered. Random walks were used to conduct controlled experiments within the same problem domain in order to investigate mean properties and to assess the performance of heuristics more objectively under different conditions. 

A random walk~$x = (x_1, \ldots, x_n)$ is of the form
\begin{align*}
x_1 &= \varepsilon_1,\\
x_i &= x_{i-1} +  \varepsilon_i \quad \text{for all } 2 \leq i \leq n,
\end{align*}
where the~$\varepsilon_i$ are random numbers drawn from the normal distribution~$N(0,1)$. 
For the UCR data sets we merged the training and test sets.
Time series within each UCR data set have the same length~$n$.
We restricted the experiments to data sets whose time series have length~$n \leq 150$. 

Unless otherwise stated, we generated samples according to the following procedures:
\begin{itemize}
\item $\S{S}_{\rw}$:
For every~$n \in \{10, 20, \ldots, 100\}$, we generated~$1,000$ pairs of random walks of length~$n$ giving a total of~$10,000$ samples of size~$k = 2$. 

\item $\S{S}_{\rw}^k$:
For every~$k \in \{2, \ldots, 6\}$, we generated~$1,000$ samples consisting of~$k$ random walks of length~$n = 6$ giving a total of~$5,000$ samples. 

\item $\S{S}_{\ucr}$:
For every UCR data set, we randomly sampled~$1,000$ different pairs of time series giving a total of~$27,000$ samples of size~$k = 2$. 
\end{itemize}

\subsection{Mean Properties}\label{subsec:exp:properties}
The first set of experiments studies mean properties relevant for devising heuristics and mean-based applications. 
Due to the running time of~\cref{algo:DTWMean}, the majority of experiments focused on properties of condensed means. 
The analysis of properties of arbitrary means are confined to a subset of tiny scale problems.

\subsubsection{Uniqueness}

\begin{table}
\centering
\scriptsize
\begin{tabular}{l@{\qquad}rrrHH@{\quad}rH}
\textbf{Condensed means} \\
\toprule
data set 					  & $P_{\ucm}$&\multicolumn{1}{c}{$\overline{\alpha}$}&$\alpha^*$ &&&\multicolumn{1}{c}{$\overline{\delta}$}& $\delta_{\S{S}}^*$ \\
\midrule
$\S{S}_{\rw}$                   & 100.0 &   1.0 $(\pm  0.00)$ &   1.0 &   0 &   0.0 &  -30.3 $(\pm 25.79)$ &  40.0\\
$\S{S}_{\rw}^k$              & 100.0 &   1.0 $(\pm  0.00)$ &   1.0 &   0 &   0.0 &  -38.5 $(\pm 18.91)$ &  23.3\\
\midrule
ItalyPowerDemand               &  99.4 &   1.0 $(\pm  0.08)$ &   2.0 &  24 &  24.0 &  -5.8 $(\pm 11.29)$ &  16.7\\
synthetic\_control             & 100.0 &   1.0 $(\pm  0.00)$ &   1.0 &  60 &  60.0 &  -8.3 $(\pm 16.87)$ &  26.7\\
SonyAIBORobotSurface           &  43.9 &   2.2 $(\pm  2.01)$ &  16.0 &  70 &  70.0 & -13.9 $(\pm \phantom{0}5.53)$ &   2.9\\
SonyAIBORobotSurfaceII         &  48.2 &   2.1 $(\pm  2.03)$ &  16.0 &  65 &  65.0 & -12.9 $(\pm \phantom{0}7.15)$ &   6.2\\
ProximalPhalanxTW              &  80.2 &   1.2 $(\pm  0.40)$ &   2.0 &  80 &  80.0 & -16.0 $(\pm 13.79)$ &   8.8\\
ProximalPhalanxOutlineCorrect  &  69.5 &   1.3 $(\pm  0.46)$ &   2.0 &  80 &  80.0 & -19.9 $(\pm 13.48)$ &   8.8\\
ProximalPhalanxOutlineAgeGroup &  78.6 &   1.2 $(\pm  0.41)$ &   2.0 &  80 &  80.0 & -17.5 $(\pm 13.15)$ &   8.8\\
PhalangesOutlinesCorrect       &  82.5 &   1.2 $(\pm  0.38)$ &   2.0 &  80 &  80.0 & -17.3 $(\pm 13.24)$ &  11.3\\
MiddlePhalanxTW                &  82.8 &   1.2 $(\pm  0.38)$ &   2.0 &  80 &  80.0 & -15.6 $(\pm 13.67)$ &  11.3\\
MiddlePhalanxOutlineCorrect    &  88.3 &   1.1 $(\pm  0.32)$ &   2.0 &  80 &  80.0 & -13.8 $(\pm 13.76)$ &  10.0\\
MiddlePhalanxOutlineAgeGroup   &  85.2 &   1.1 $(\pm  0.36)$ &   2.0 &  80 &  80.0 & -14.7 $(\pm 13.64)$ &  11.3\\
DistalPhalanxTW                &  74.5 &   1.3 $(\pm  0.44)$ &   2.0 &  80 &  80.0 & -17.9 $(\pm 13.20)$ &  10.0\\
DistalPhalanxOutlineCorrect    &  81.7 &   1.2 $(\pm  0.39)$ &   2.0 &  80 &  80.0 & -16.4 $(\pm 13.46)$ &  11.3\\
DistalPhalanxOutlineAgeGroup   &  76.6 &   1.2 $(\pm  0.42)$ &   2.0 &  80 &  80.0 & -17.6 $(\pm 13.21)$ &  10.0\\
TwoLeadECG                     &  94.7 &   1.1 $(\pm  0.26)$ &   4.0 &  82 &  82.0 &  -9.3 $(\pm \phantom{0}8.07)$ &   9.8\\
MoteStrain                     & 100.0 &   1.0 $(\pm  0.00)$ &   1.0 &  84 &  84.0 & -30.9 $(\pm 16.55)$ &   7.1\\
ECG200                         & 100.0 &   1.0 $(\pm  0.00)$ &   1.0 &  96 &  96.0 &  -8.8 $(\pm \phantom{0}9.59)$ &  12.5\\
MedicalImages                  & 100.0 &   1.0 $(\pm  0.00)$ &   1.0 &  99 &  99.0 & -32.1 $(\pm 14.34)$ &   7.1\\
Two\_Patterns                  &   0.5 &   7.9 $(\pm  4.66)$ &  16.0 & 128 & 128.0 & -27.3 $(\pm \phantom{0}6.40)$ &   1.6\\
SwedishLeaf                    & 100.0 &   1.0 $(\pm  0.00)$ &   1.0 & 128 & 128.0 &   0.9 $(\pm \phantom{0}6.13)$ &  18.8\\
CBF                            & 100.0 &   1.0 $(\pm  0.00)$ &   1.0 & 128 & 128.0 &   2.8 $(\pm \phantom{0}6.82)$ &  19.5\\
FacesUCR                       &  99.7 &   1.0 $(\pm  0.05)$ &   2.0 & 131 & 131.0 &  -8.2 $(\pm \phantom{0}5.85)$ &   4.6\\
FaceAll                        &  99.8 &   1.0 $(\pm  0.04)$ &   2.0 & 131 & 131.0 &  -8.1 $(\pm \phantom{0}5.72)$ &   9.2\\
ECGFiveDays                    &  97.4 &   1.0 $(\pm  0.16)$ &   2.0 & 136 & 136.0 & -22.3 $(\pm 13.87)$ &   8.1\\
ECG5000                        & 100.0 &   1.0 $(\pm  0.00)$ &   1.0 & 140 & 140.0 & -21.5 $(\pm 16.04)$ &  10.0\\
Plane                          & 100.0 &   1.0 $(\pm  0.00)$ &   1.0 & 144 & 144.0 &  -2.3 $(\pm \phantom{0}4.68)$ &  11.1\\
Gun\_Point                     & 100.0 &   1.0 $(\pm  0.00)$ &   1.0 & 150 & 150.0 & -49.1 $(\pm 14.02)$ &   0.0\\
\bottomrule 
\\
\textbf{Non-condensed means} \\
\toprule
data set 					  			& $P_{\ncm}$\\
\midrule
$\S{S}_{\rw}$ with $10 \leq n \leq 40$	& 0.3\\
ItalyPowerDemand               			& 2.1\\
\bottomrule
\\
\textbf{Notation}\\
\toprule
\multicolumn{8}{l}{%
\begin{tabular}{l@{\qquad}l}
$P_{\ucm}$    & percentage of unique condensed means \\
$P_{\ncm}$    & percentage of non-condensed means\\
$\overline{\alpha}$    & average number of condensed means \\
$\alpha^*$  & maximum number of condensed means \\
$\overline{\delta}$    & average length-deviation 
\end{tabular}}\\
\bottomrule
\end{tabular}
\caption{The table on the top shows properties on uniqueness and length of condensed means. Numbers in parentheses show standard deviations. The table below shows the percentage of non-condensed means for a restricted set of samples. The bottom table describes the notation used as column identifiers.}
\label{tab:res_ucr_cm}
\end{table}

Non-unique means can cause problems in theory and practice. For example, proving that sample means are consistent estimators of population means becomes more complicated for non-unique means.
In addition, non-unique means can introduce undesired ambiguities into mean-based applications. For example, the performance of $k$-means clustering in DTW spaces 
\cite{HNF08,PFWNCK14,PFWNCK16,Soheily-Khah2016} does not only depend on the choice of initial means but also on the choice of recomputed means during optimization. Consequently, the extent of these difficulties depends---at least in principle---on the prevalence of non-unique means. 
In the following, we investigate the prevalence of non-unique condensed (and non-condensed) means.

\paragraph{Setup:}
We applied \Cref{algo:DTWMean} to all samples of type~$\S{S}_{\rw}$,~$\S{S}_{\rw}^k$, and~$\S{S}_{\ucr}$. The percentage~$P_{\ucm}$ of unique condensed means as well as the average and maximum number of condensed means, denoted by~$\overline{\alpha}$ and~$\alpha^*$, were recorded.
For samples consisting of pairs of time series of length $n \leq 40$, we tested the existence of non-condensed means. These samples are of type $\S{S}_{\rw}$ with $k = 2$ and $n \in \cbrace{10, 20, 30, 40}$ and the $1,000$ samples from the \emph{ItalyPowerDemand} data set.

\paragraph{Results and Discussion:}

\begin{figure*}[t]
\centering
\includegraphics[width=0.48\textwidth]{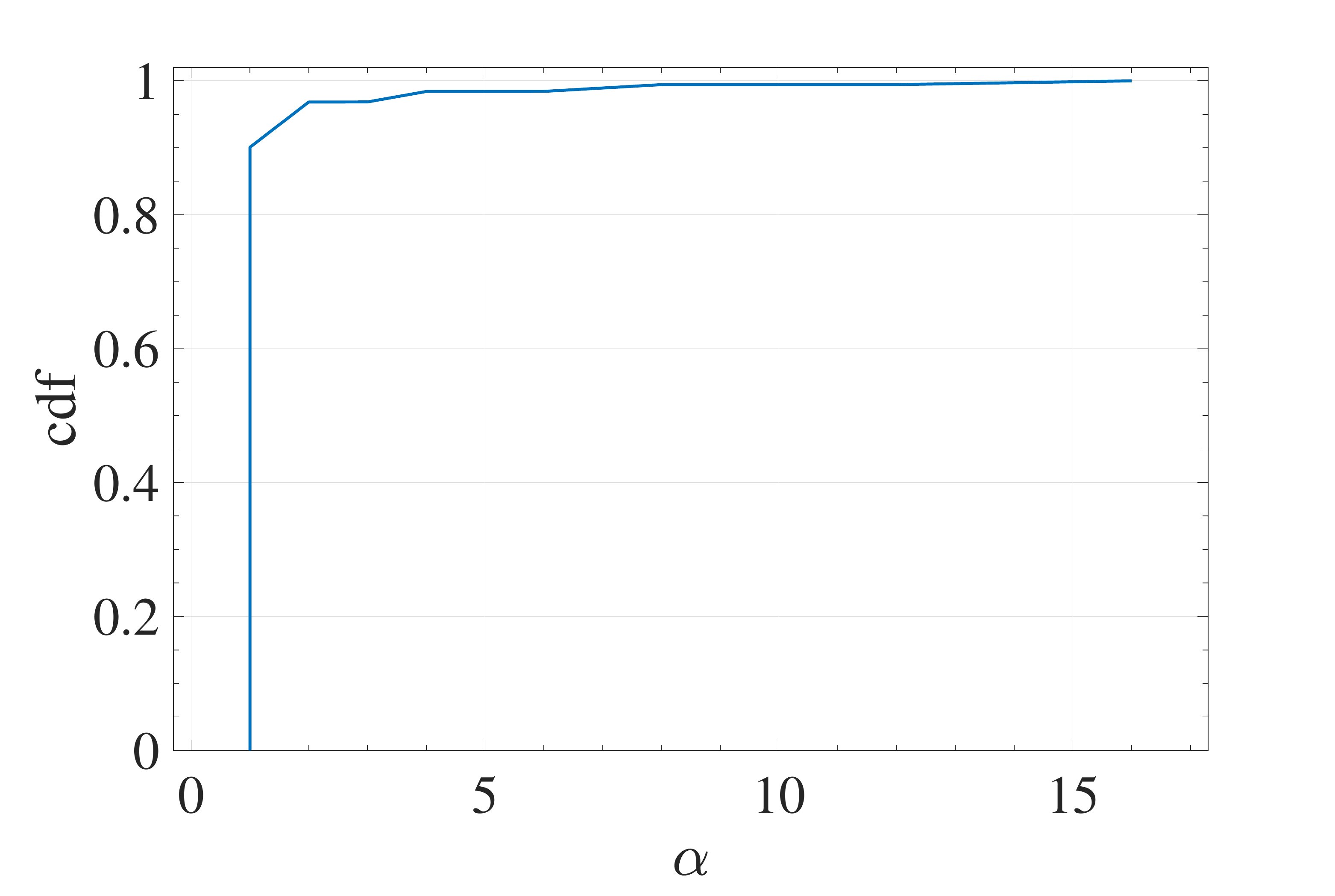}
\hfill
\includegraphics[width=0.48\textwidth]{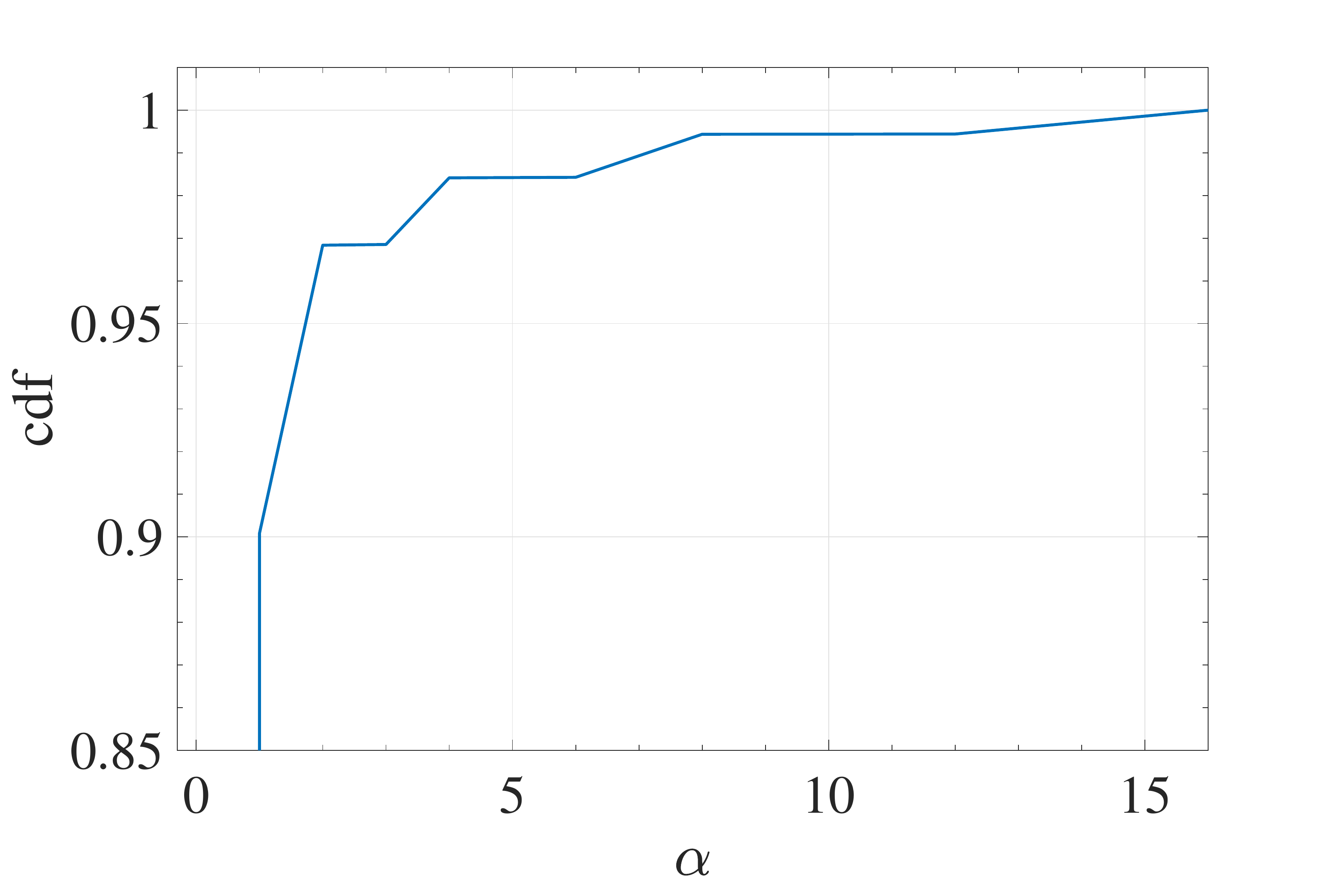}
\caption{Estimated cumulative distribution function (cdf) of the number~$\alpha$ of condensed means over all $42,000$ samples. The right plot shows an excerpt of the left plot with the cdf restricted to the interval $[0.85, 1]$.}
\label{fig:exp_cdf_P_cm}
\end{figure*}

\Cref{tab:res_ucr_cm} summarizes the results. \Cref{fig:exp_cdf_P_cm} shows the estimated cumulative distribution function of the number~$\alpha$ of condensed means over all $42,000$ samples. We made the following observations:
\begin{enumerate}
\item A condensed mean is unique for the majority of samples (see column $P_{\ucm}$ of \Cref{tab:res_ucr_cm} and \Cref{fig:exp_cdf_P_cm}). 
\item Non-condensed means occur exceptionally, that is, in less than $2.5 \%$ of all samples with short time series (see column $P_{\ncm}$ of \Cref{tab:res_ucr_cm}).
\item The average and maximum number of condensed means is between one and two for all but four UCR data sets (see columns $\overline{\alpha}$ and $\alpha^*$ of \Cref{tab:res_ucr_cm}).
\end{enumerate}
These findings indicate that unique (condensed) means are more likely than non-unique means. The implications of the proposed observations are twofold: First, mean-based methods such as $k$-means clustering are less likely to be prone to problems caused by ambiguities. Second, the observations give rise to the hypothesis that a (condensed) mean of a sample is unique \emph{almost everywhere} in a measure-theoretic sense. If this hypothesis is true, then the aforementioned problems caused by non-unique means are hopefully negligible in practice (we remark that optimal warping paths are unique almost everywhere~\cite{JS17} which might indicate that the above hypothesis holds).

\begin{figure*}[ht!]
\centering
\begin{subfigure}{1\textwidth}
\includegraphics[width=1\textwidth]{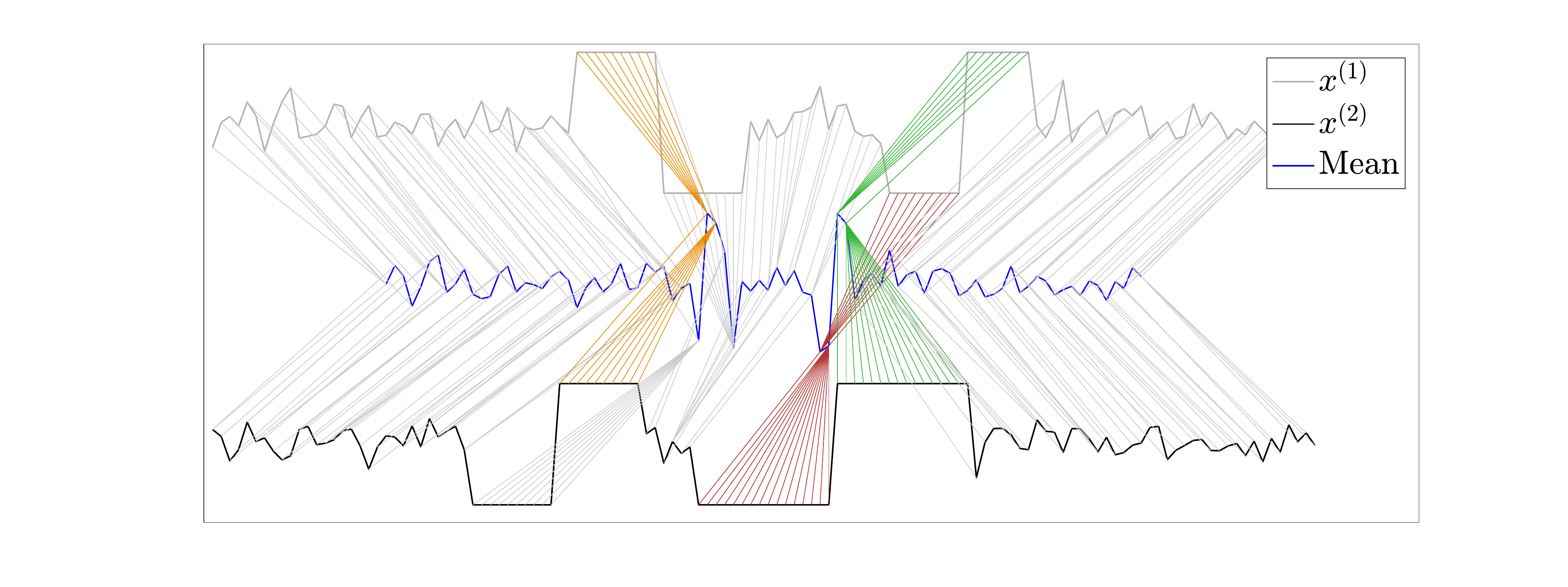}
\vspace*{-2em}
\caption{}
\end{subfigure}

\begin{subfigure}{0.49\textwidth}
\includegraphics[width=1\textwidth]{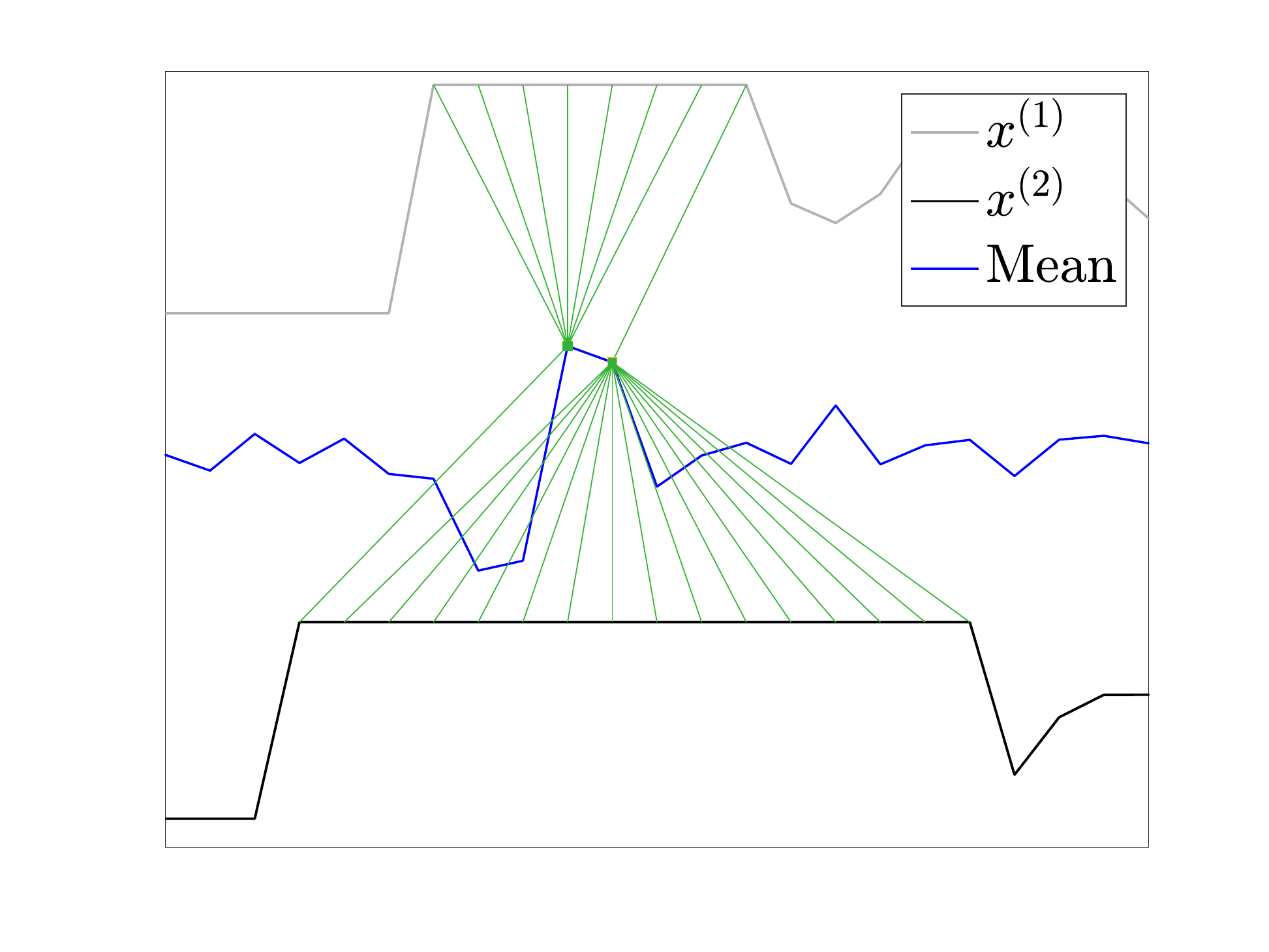}
\vspace*{-2em}
\caption{}
\end{subfigure}
\begin{subfigure}{0.49\textwidth}
\includegraphics[width=1\textwidth]{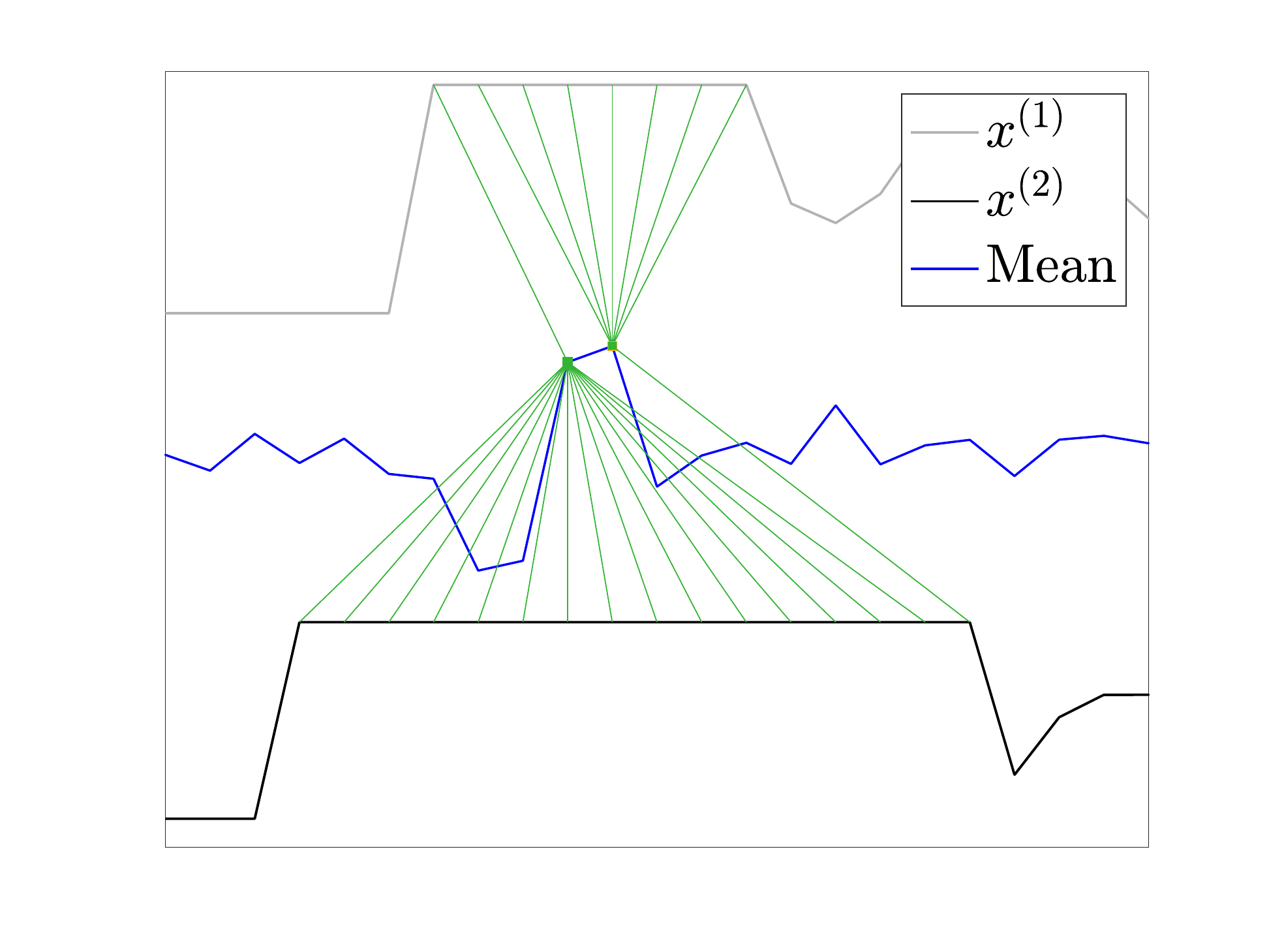}
\vspace*{-2em}
\caption{}
\end{subfigure}
\caption{Two time series $x^{(1)}$ and $x^{(2)}$ from the \emph{Two\_Patterns} data set with eight different condensed means. Plot (a) shows an example of a mean and corresponding optimal warping paths between the mean and both sample time series. The three colored warpings show regions that cause non-uniqueness. Plots (b) and (c) show a detailed view of the green warpings. In (b) the mean of (a) is depicted. In (c) another mean is shown. In (b) and (c) both plateaus of $x^{(1)}$ and $x^{(2)}$ are aligned with only two points in the mean; one point is aligned with all but one points of the plateau of $x^{(1)}$, the other point is aligned with all but one points of the plateau of $x^{(2)}$. The only difference in (b) and (c) is the order of the two mean points.
  Each pair of warped plateaus provides two possibilities to construct a mean. Consequently, the three warped pairs of plateaus in (a) induce $2^3 = 8$ condensed means.}
\label{fig:exp_TwoPatterns}
\end{figure*}

A notable exception is the (simulated) \emph{Two\_Patterns} data set with only $0.5\%$ of samples with a unique condensed mean. \Cref{fig:exp_TwoPatterns} shows two typical time series of this data set and describes why their condensed mean is not uniquely determined.
From the description it follows that warping of two plateaus can cause non-uniqueness. Such plateaus are common in the \emph{Two\_Patterns} data set which explains the low percentage of unique condensed means.

\subsubsection{Length}\label{sec:length}
For a given sample, the current state-of-the-art heuristics \cite{CB17,HNF08,PKG11,SJ17} approximate a mean by approximately solving a constrained \DTW problem. The constrained problem restricts the set of feasible solutions to the subset~$\S{T}_m$ of time series of length~$m$. The parameter~$m$ is typically chosen within the range of the lengths of the sample time series. 
The two main questions are whether the subset $\S{T}_m$ contains a mean at all and what the best possible approximation we can achieve is when constraining the solution set to~$\S{T}_m$?
We start with the first question.

\paragraph{Setup:}
We applied \Cref{algo:DTWMean} to all samples of type~$\S{S}_{\rw}$,~$\S{S}_{\rw}^k$, and~$\S{S}_{\ucr}$. For every sample $\mathcal{X}$, the set of all condensed means was computed and the lengths were recorded. We computed the \emph{length-deviation} 
\[
\delta_\mathcal{X} := 100\cdot \frac{m_{\cm}-n}{n},
\]
where~$n$ is the length of the sample time series and~$m_{\cm}$ is the average length of all condensed means of $\mathcal{X}$. Negative (positive) values of~$\delta_\mathcal{X}$ mean that~$m_{\cm}$ is~$|\delta_\mathcal{X}|$ percent smaller (larger) than~$n$.

\paragraph{Results and Discussion:}

\begin{figure*}[t]
\begin{subfigure}{0.49\textwidth}
\includegraphics[width=\textwidth]{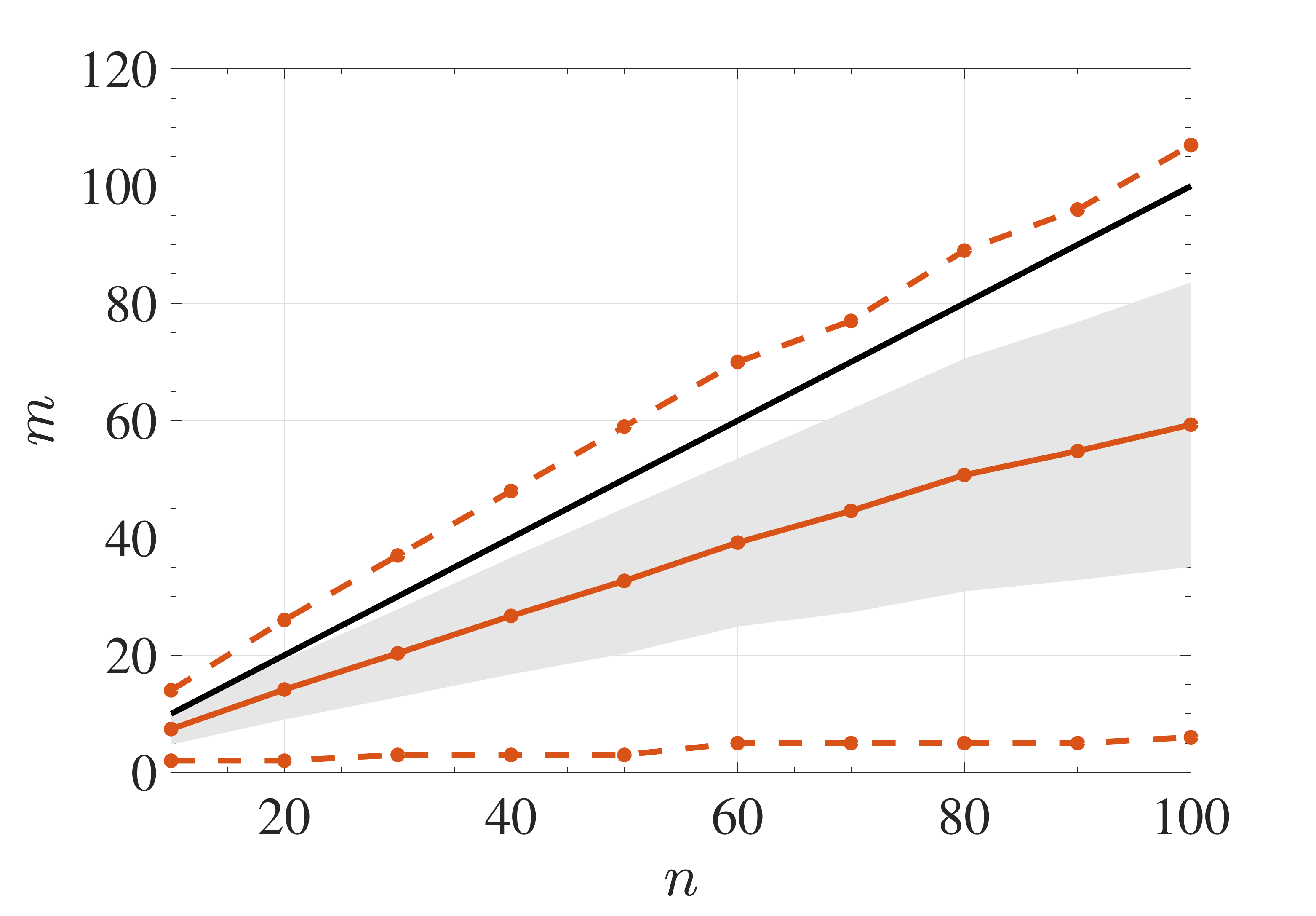}
\end{subfigure}
\hfill
\begin{subfigure}{0.49\textwidth}
\includegraphics[width=\textwidth]{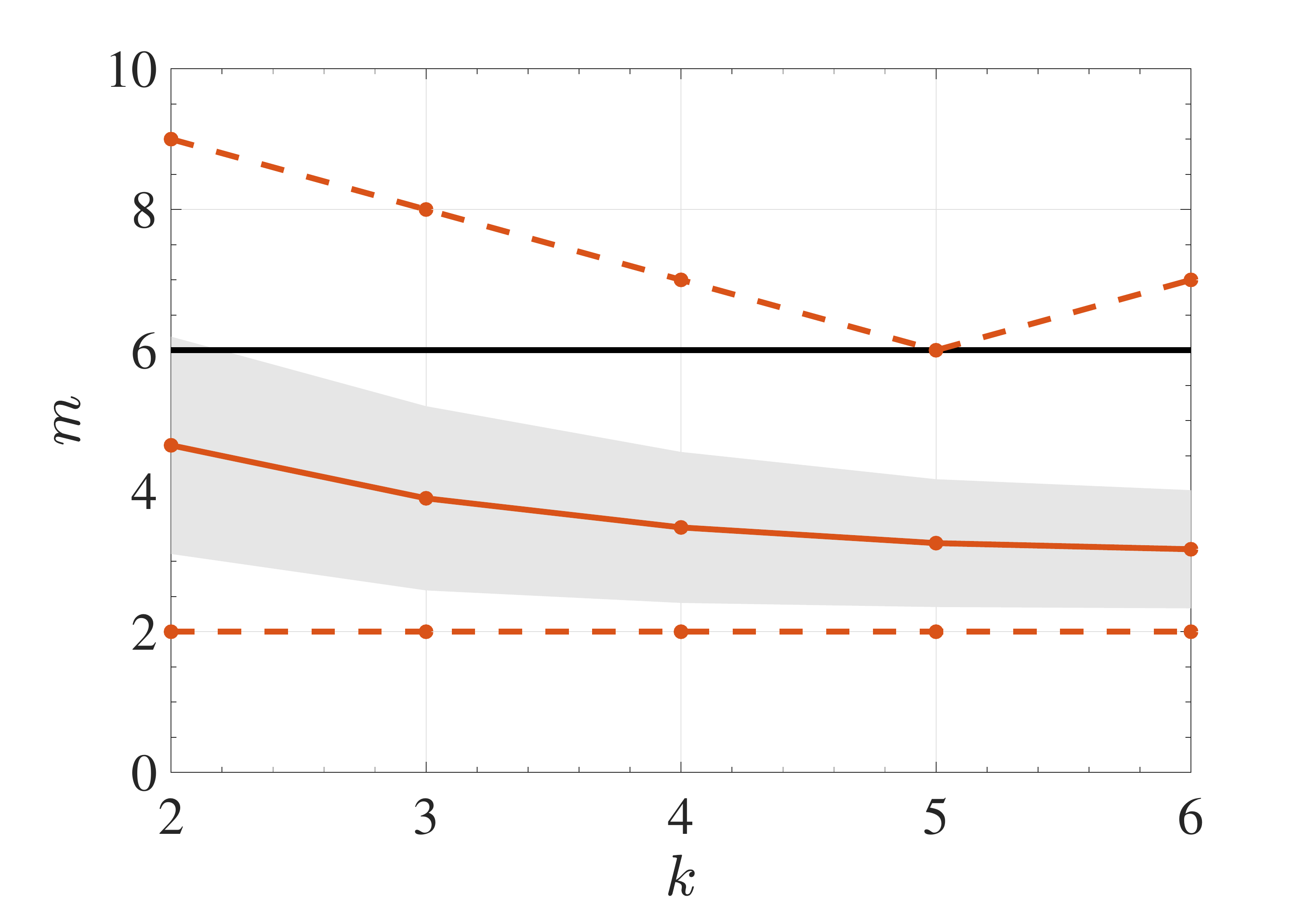}
\end{subfigure}
\caption{Average length $m$ of condensed means (red lines) as functions of the length of sample time series $n$ (left) and sample size $k$ (right). 
The gray shaded areas depict the standard deviations and the dashed red lines the minimum and maximum length of condensed means. 
The black lines show the length of the sample time series.}
\label{fig:exp_rw_length-deviation}
\end{figure*}

\Cref{tab:res_ucr_cm} and \Cref{fig:exp_rw_length-deviation} summarize the results. To discuss the results, we consider the \emph{constrained Fr\'echet variation}
\[
F_m^* =  \min_{z \in \S{T}_m} F(z)
\] 
for every length $m \in \N$.

The first observation to be made is that different condensed means of the same sample have the same length for all $42,000$ samples. Recall that non-condensed means occur only rarely (see column~$P_{\ncm}$ of \Cref{tab:res_ucr_cm}). Thus, means of a sample $\mathcal{X}$ are typically condensed with identical length~$m_{\cm}$. Under these conditions, the function~$F_m^*$ has a unique minimum at~$m_{\cm}$. Thus, for a majority of samples, the solution space $\S{T}_m$ does not contain a mean for any $m \neq m_{\cm}$. 
This gives rise to the question of how the length $m_{\cm}$ of condensed means is related to the length~$n$ of the sample time series. We observed the following:
\begin{enumerate}
\item The average length of condensed means is (substantially) less than the length of the sample time series for all but two UCR data sets (see column~$\overline{\delta}$ of \Cref{tab:res_ucr_cm}). 
\item The average length $m_{\cm}$ of condensed means increases with increasing length $n$ of random walks (see \Cref{fig:exp_rw_length-deviation} (left)). The best linear fit in a squared error sense has slope $0.58$ indicating that $m_{\cm}$ increases slower than $n$. No correlation between average length-deviations and length of sample time series could be observed across different UCR data sets.
\item The average length $m_{\cm}$ of condensed means decreases with increasing sample size $k$ of random walks (see \Cref{fig:exp_rw_length-deviation} (right)).
\end{enumerate}

The above observations have the following implications: Most state-of-the-art heuristics were tested on sample time series with identical length~$n$ using~$m = n$ as predefined mean length~\cite{CB17,PKG11,SJ17}. The empirical findings suggest that the length of a condensed mean is substantially shorter than the standard setting $m = n$. Consequently, in most cases, the solution space $\S{T}_n$ of state-of-the-art methods does not contain a mean. Thus, setting the mean length equal to the length of the input time series introduces a \emph{structural} error that can not be overcome by any solver of the corresponding constrained \DTW problem. 
The results suggest to consider mean lengths which are smaller than $n$.
Note that in this case the computation time of mean-algorithms should decrease.

We now address the second main question, that is, 
how the choice of a mean length $m$ affects the structural error
$\varepsilon_m := F_m^* - F_*$, where~$F_m^* = \min_{z \in \S{T}_m} F(z)$ is the constrained and~$F_* = \min_{z \in \S{T}} F(z)$ is the unconstrained Fr\'echet variation.
To estimate the structural error, we study the constrained Fr\'echet variation $F_m^*$ as a function
of $m$. 

\paragraph{Setup:}
For every~$n \in \cbrace{5, 10, 15, \ldots, 40}$ we generated~$1,000$ pairs of random walks. We applied \Cref{algo:DTWMean} to all samples and computed the constrained Fr\'echet variations~$F_m^*$ for all~$m \in \{1,\ldots,60\}$.  

\paragraph{Results and Discussion:}

\begin{figure*}[t]
\centering
\includegraphics[width=0.65\textwidth]{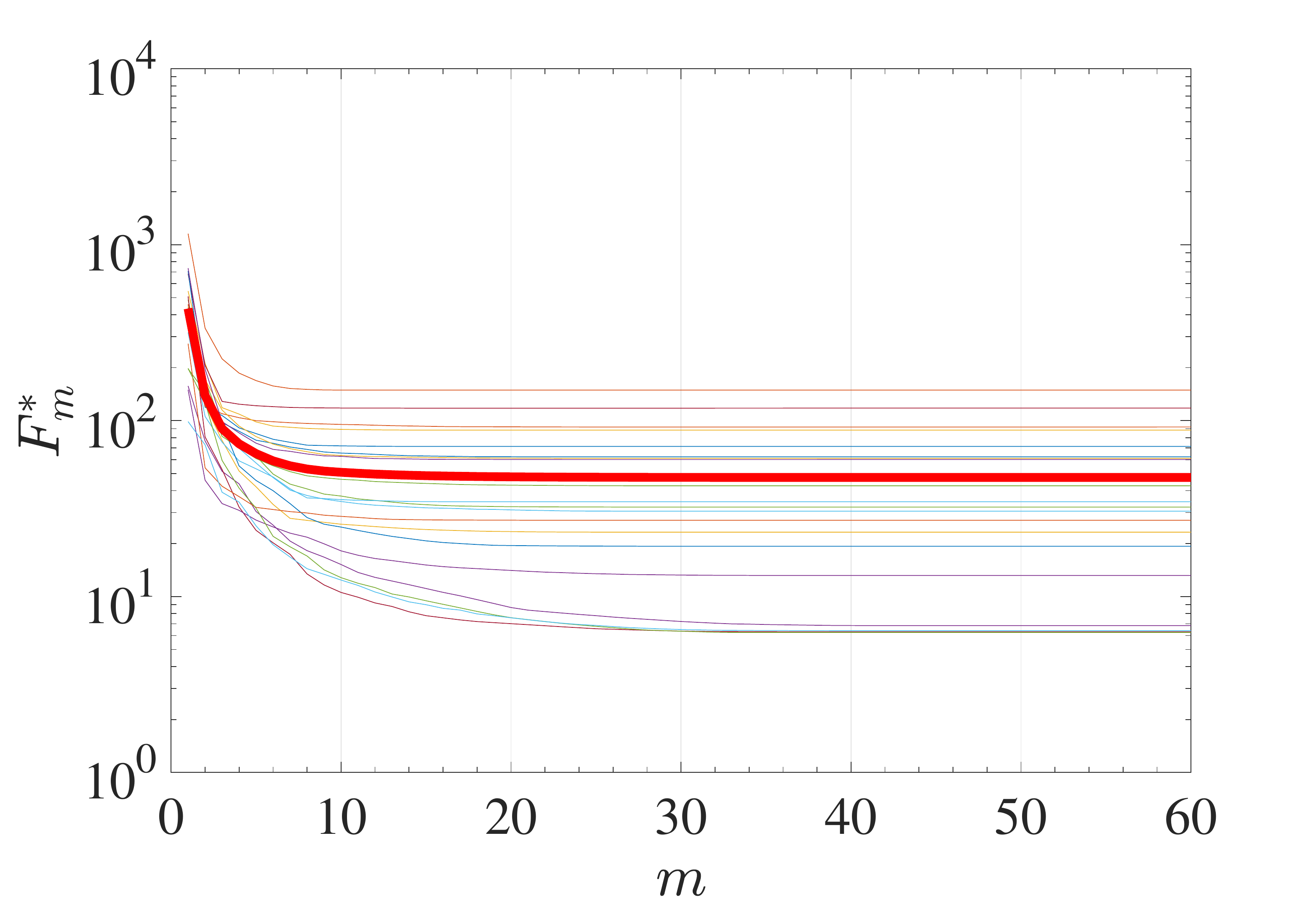}
\caption{Constrained Fr\'echet variations~$F_m^*$ as functions of length~$m\in\{1,\ldots,60\}$ shown in lin-log scale. 
The functions~$F_m^*$  were computed for~$1,000$ samples consisting of two random walks of length $n = 40$. The plot depicts~$20$ randomly selected functions~$F_m^*$. The highlighted red line is the average function~$\overline{F}_m^{\,*}$ over all~$1,000$ functions.}
\label{fig:exp_rwK2LFn_20}
\end{figure*}

\begin{figure*}[t!]
\centering
\begin{subfigure}{0.8\textwidth}
\includegraphics[width=\textwidth]{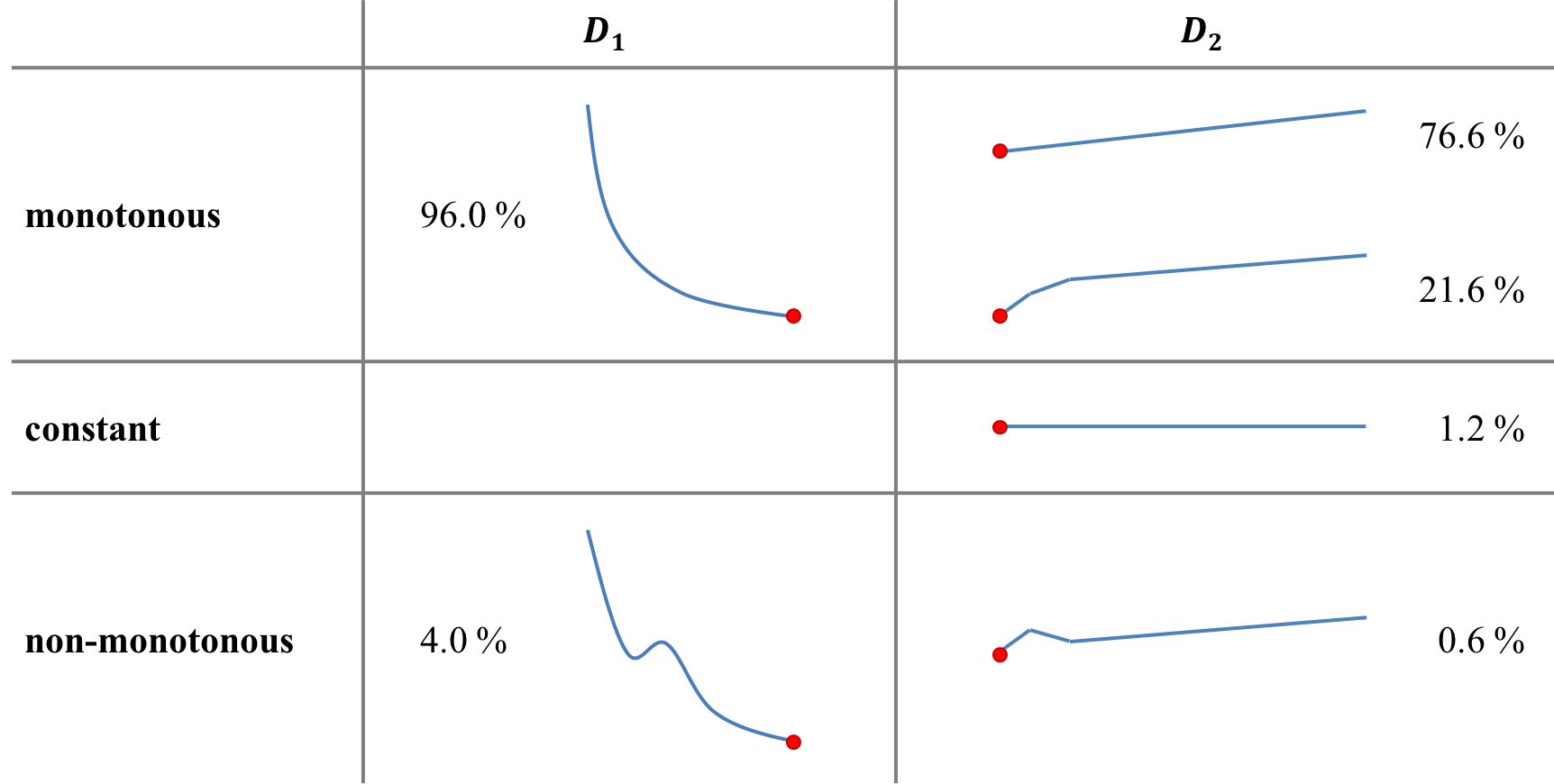}
\caption{Characteristic behavior of~$F_m^*$ on~$\S{D}_1$ and~$\S{D}_2$.}
\end{subfigure}

\vspace{0.5cm}

\begin{subfigure}{0.8\textwidth}
\includegraphics[width=\textwidth]{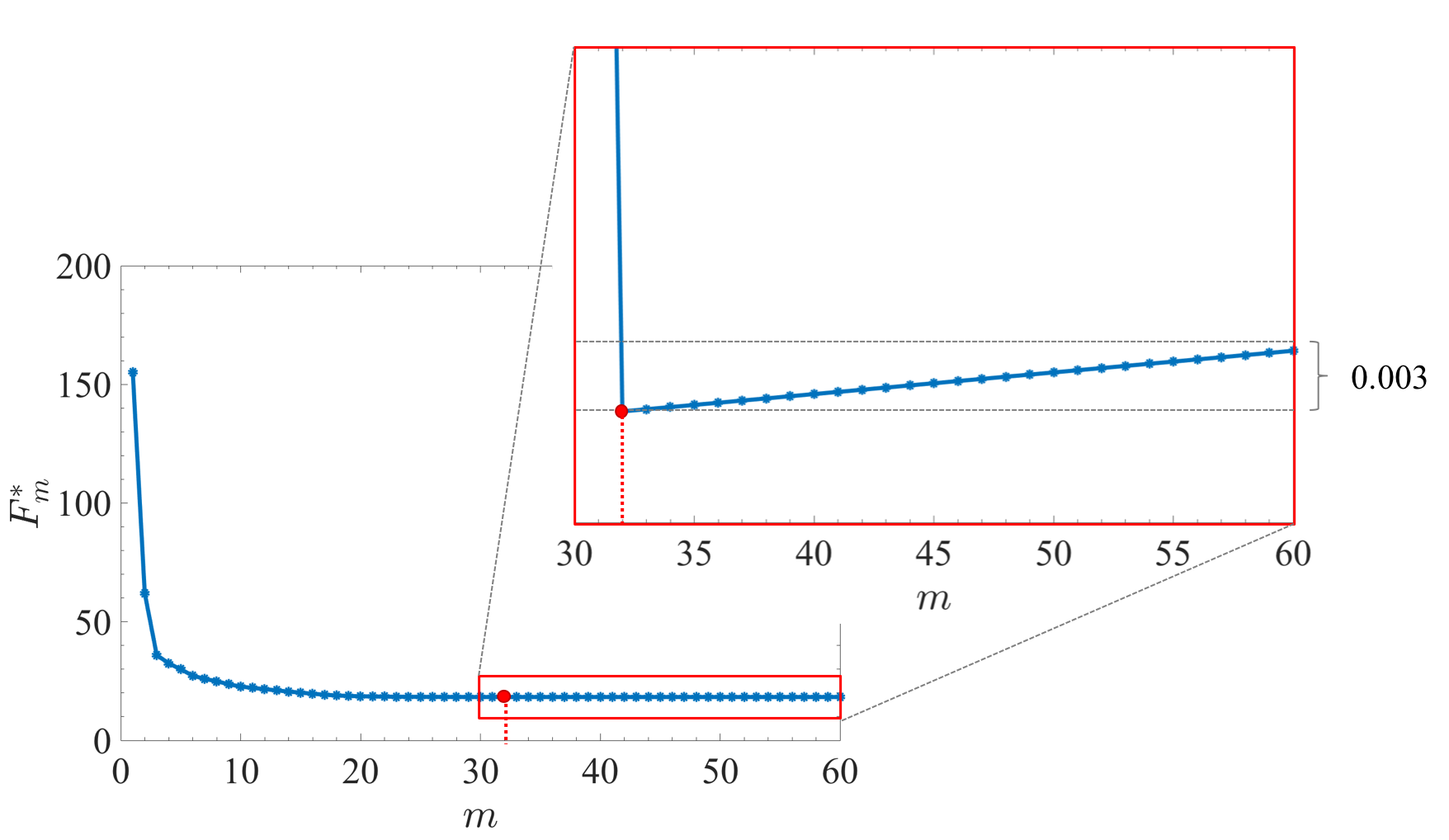}
\caption{Typical shape of $F_m^*$.}
\end{subfigure}
\caption{Plot (a) sketches the main characteristics of all observed shapes of~$F_m^*$ on~$\S{D}_1$ and~$\S{D}_2$ together with their percentage frequency. Sketches exaggerate bumps and slopes to better highlight the main features. Plot (b) shows a typical example of the most common shape of~$F_m^*$ for a sample of two random walks of length $n = 40$. The function~$F_m^*$ strictly decreases on~$\S{D}_1$ until it reaches its global minimum at $m_{\cm} = 32$ as highlighted by the red circle. Then the function~$F_m^*$ strictly increases on~$\S{D}_2$. The best linear fit of~$F_m^*$ on~$\S{D}_2$ has a very small slope of approximately~$10^{-5}$. }
\label{fig:exp_rwK2LFn}
\end{figure*}

\Cref{fig:exp_rwK2LFn_20,fig:exp_rwK2LFn} depict the main characteristics of the functions~$F_m^*$. 
As before, condensed means of the same sample always had identical length~$m_{\cm}$. We use this result for our discussion on the shape of~$F_m^*$ by decomposing its domain~$\S{D}=\{1,\ldots,60\}$ into two sub-domains $\S{D}_1 = \cbrace{1, \ldots, m_{\cm}}$ and $\S{D}_2 = \cbrace{m_{\cm}, \ldots, 60}$.

We observed the following typical behavior: The function $F_m^*$ first rapidly decreases on~$\S{D}_1$ until it reaches its global minimum at $m_{\cm}$ and then increases on~$\S{D}_2$ with a small slope. The shape of $F_m^*$ roughly resembles an exponential decay on~$\S{D}_1$ followed by a linear tail on~$\S{D}_2$. 
\Cref{fig:exp_rwK2LFn}~(a) categorizes all observed shapes of~$F_m^*$ on~$\S{D}_1$ and~$\S{D}_2$. The most common shape of~$F_m^*$ is strictly monotonous on~$\S{D}_1$ (96.0\%) and~$\S{D}_2$ (98.2 \%). \Cref{fig:exp_rwK2LFn}~(b) presents a typical example of~$F_m^*$. Non-monotonous or constant parts in a curve occur only rarely. These findings indicate that the function~$F_m^*$ is strictly convex in most cases. This implies that the further away $m$ is from the global minimum $m_{\cm}$, the larger the structural error~$\varepsilon_m$ is. 

\Cref{tab:exp_rwK2LFn} summarizes the slopes of the best linear fits of~$F_m^*$ on~$\S{D}_2$ and the error percentages $E = 100 \cdot(F_n^*-F_*)/F_*$ of the constrained Fr\'echet variation at parameter $m = n$, where $n$ is the length of the sample time series. 
The median slopes on~$\S{D}_2$ are low $a_{\med} \approx 0.0034$ for $n = 5$ and tend to further decrease with increasing length $n$ to  $a_{\med} \approx 0.0001$ for $n = 40$. As a consequence of the low median slopes, the median error percentages decline from $E_{\med} \approx 0.42$ to less than $E_{\med} \approx 0.01$ for increasing $n$. These findings indicate that the common practice of setting the parameter $m$ of the constrained \DTW problem to the length $n$ of the sample time series results in a low structural error on median. 
Due to the low median slopes on~$\S{D}_2$, solving the constrained \DTW problem instead of the unconstrained \DTW problem can in principle give good approximations for values $m > m_{\cm}$.
The average and maximum values for the slopes and error percentages, however, show that there are outliers that result in large structural error between $37 \%$ and $120 \%$.

\begin{table}
\centering
\begin{tabular}{r@{\qquad}rrrr@{\qquad}rrrr}
\toprule
 n & $a_{\med}$ & $a_{\avg}$ & $a_{\std}$ & $a_{\max}$ & $E_{\med}$ & $E_{\avg}$ & $E_{\std}$ & $E_{\max}$ \\
 \midrule
 5 & 0.0034 & 0.12 & 0.36 & 3.78 & 0.423 & 4.43 & 9.59 & 67.40 \\
10 & 0.0008 & 0.05 & 0.25 & 3.37 & 0.057 & 2.13 & 8.50 & 119.42\\
15 & 0.0005 & 0.05 & 0.24 & 2.80 & 0.028 & 1.49 & 7.30 & 92.65\\
20 & 0.0002 & 0.03 & 0.14 & 1.83 & 0.012 & 0.69 & 3.36 & 40.22\\
25 & 0.0002 & 0.02 & 0.14 & 3.15 & 0.008 & 0.47 & 2.77 & 37.70\\
30 & 0.0001 & 0.01 & 0.10 & 1.63 & 0.004 & 0.37 & 2.42 & 40.37\\
35 & 0.0001 & 0.01 & 0.06 & 1.35 & 0.003 & 0.20 & 2.22 & 52.44\\
40 & 0.0001 & 0.01 & 0.12 & 2.16 & 0.002 & 0.28 & 2.24 & 51.20\\
\bottomrule
\end{tabular}
\caption{Slopes $a$ of best linear fits of the second part of $F_m^*$ and error percentages $E$ of the structural error. Subscripts~$\med$, $\avg$, $\std$, $\max$ refer to the median, average, standard deviation, and maximum, respectively, of $1,000$ pairs of random walks of length $n$.}
\label{tab:exp_rwK2LFn}
\end{table}

\subsubsection{Shape}

One problem that is often associated with dynamic time warping is shape averaging. \citet{NR09} proposed a shape averaging algorithm that uses dynamic time warping, but differs from minimizing the Fr\'echet function. \citet{sun2017} claimed that their approach of minimizing the Fr\'echet function preserves shape characteristics. \citet{CB17} criticized that solutions found by DBA may have a shape that is not representative for the sample time series. 

It is well-known that dynamic time warping tends to warp similar shapes (subsequences) of two time series onto each other. 
Hence, at first glance one may expect that a mean preserves shapes. 
However, the commonly and here considered definition of a mean is statistically motivated.
We clarify in this section that a mean in DTW spaces does not preserve shapes in general.

\paragraph{Setup:}
We selected two sample time series from the UCR data set \emph{ECG200}. Both time series were normalized with zero mean and standard deviation one. We applied \Cref{algo:DTWMean} and plotted the sample time series and their mean.

\begin{figure}[t]
\centering
\begin{subfigure}[c]{0.495\textwidth}
\includegraphics[width=1\textwidth]{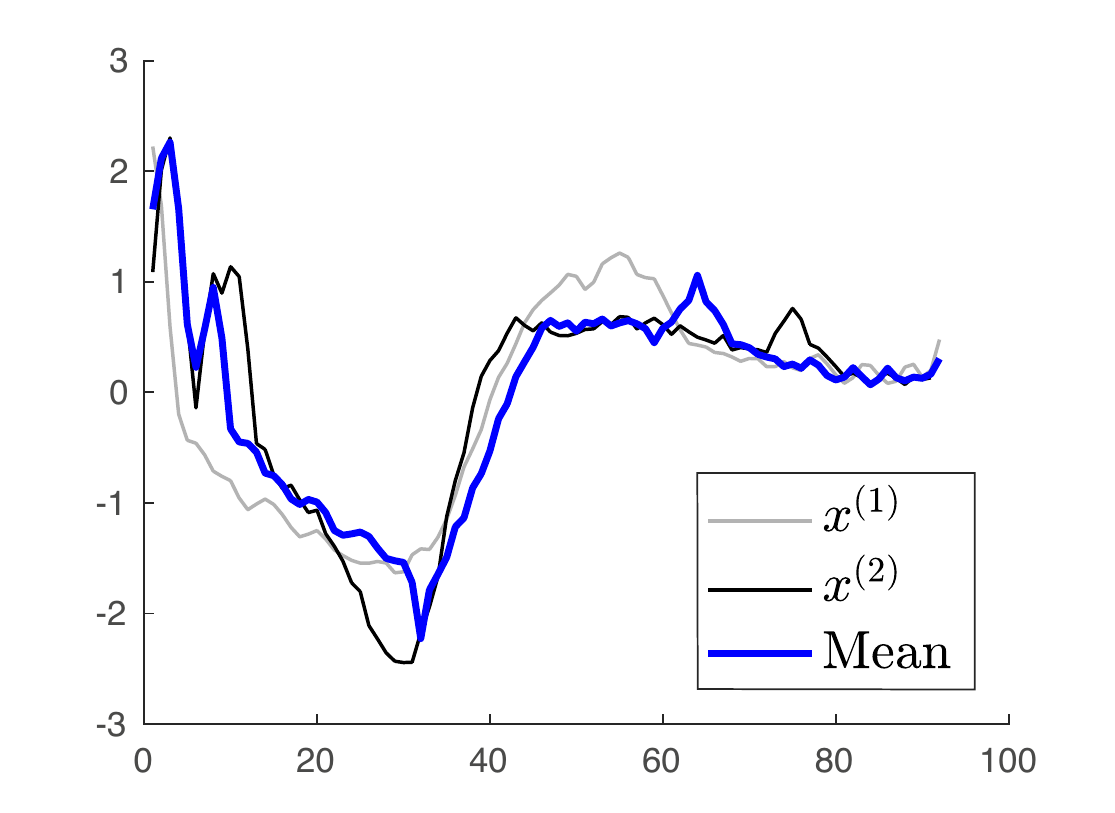}
\subcaption{Original time series $x^{(1)} = x^{(1)}_\text{o}$.}
\end{subfigure}
\hfill
\begin{subfigure}[c]{0.495\textwidth}
\includegraphics[width=1\textwidth]{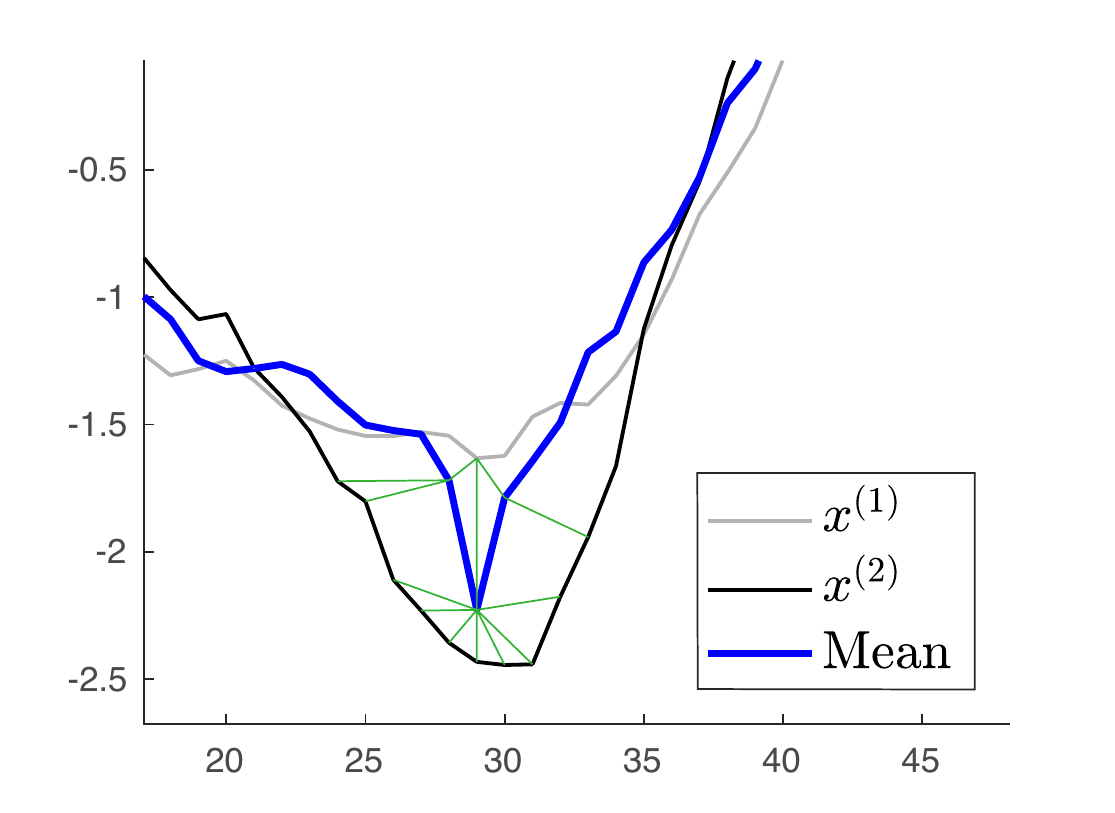}
\subcaption{Detailed view and warpings (green). }
\end{subfigure}
\caption{Mean with kink for two time series $x^{(1)}$ and $x^{(2)}$. Plot (b) shifts the mean along the time axis for clarification.}
\label{fig:kinks}
\end{figure}

\begin{figure}[t]
\begin{subfigure}[c]{0.495\textwidth}
\includegraphics[width=1\textwidth]{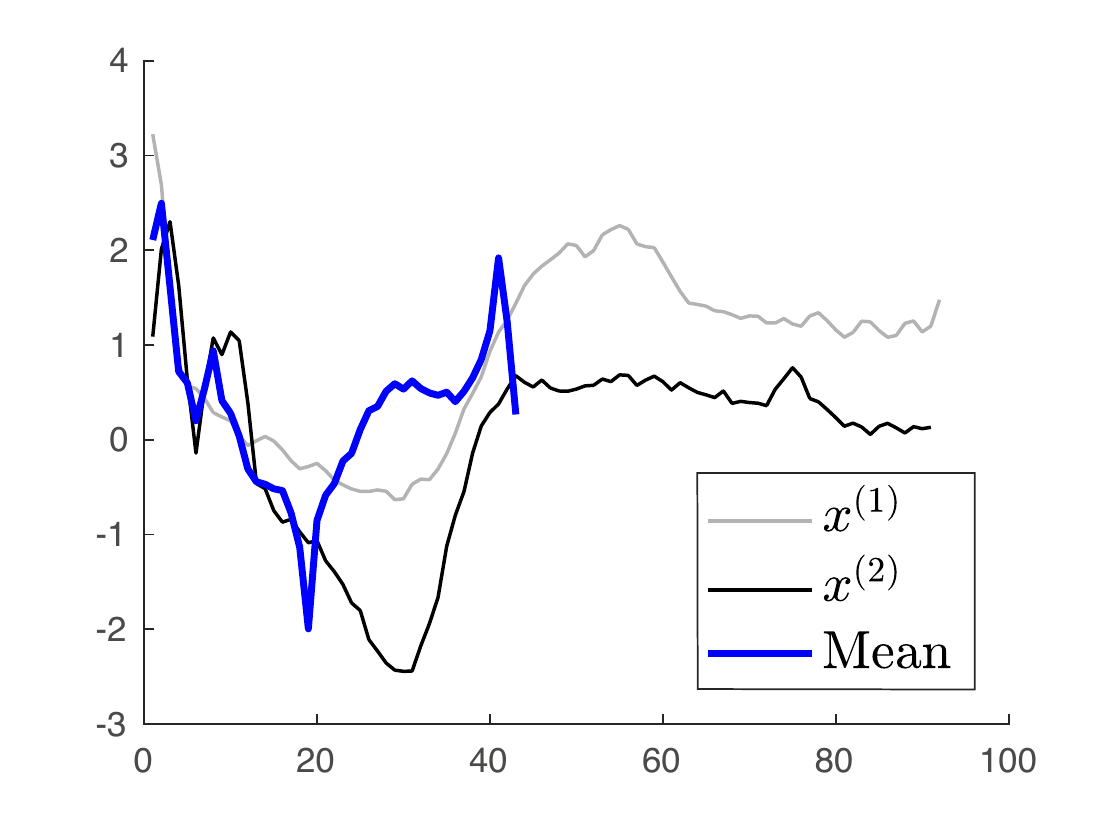}
\subcaption{Shifted time series $x^{(1)} =  x^{(1)}_\text{o} +1$.}
\end{subfigure}
\hfill
\begin{subfigure}[c]{0.495\textwidth}
\includegraphics[width=1\textwidth]{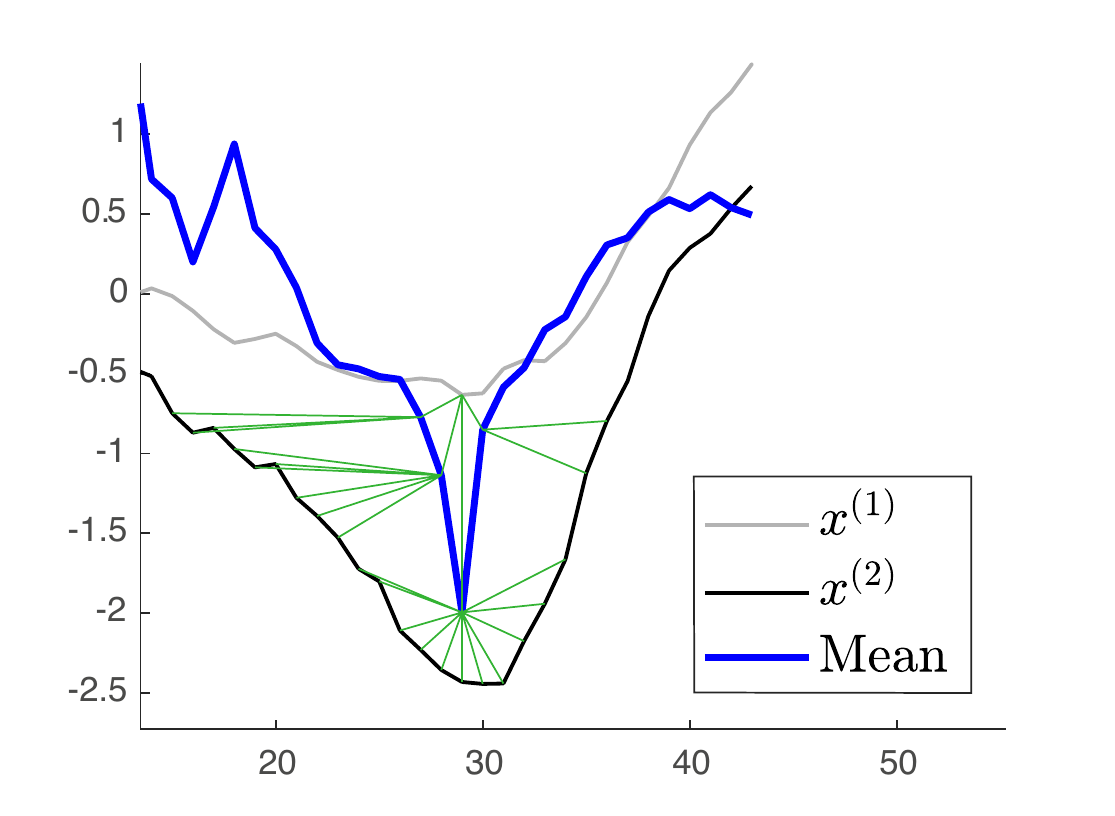}
\subcaption{Detailed view and warpings (green).}
\end{subfigure}
\caption{Mean for the two time series from \Cref{fig:kinks}, but $x^{(1)}$ is shifted along the $y$-axis. Plot (b) shifts the mean along the time axis for clarification.}
\label{fig:kinks2}
\end{figure}

\paragraph{Results and Discussion:}
\Cref{fig:kinks,fig:kinks2} show a mean of the two sample time series $x^{(1)}$ and $x^{(2)}$. In \Cref{fig:kinks}~(a) the mean has a small kink around time index $30$. This kink was also observed in the solution found by DBA in a similar experiment~\cite{CB17}. \citet{CB17} criticized that this kink is not present in any of the sample time series. We now can conclude that presence of kinks is a peculiarity of means in DTW spaces, rather than a property of the heuristics. \Cref{fig:kinks}~(b) illustrates the warpings near the kink. The kink can be explained as follows: Both sample time series have broad valleys but on a different scale. In this case, the lowest point of the upper time series~$x^{(1)}$ is warped through mean elements onto multiple points of the lower time series~$x^{(2)}$, which results in a kink in the mean.

To strengthen this effect, we shifted $x^{(1)}$ upwards in \Cref{fig:kinks2}. As expected, the kink becomes larger, because more points of $x^{(2)}$ are warped onto a single point of $x^{(1)}$ through mean elements. Moreover, the length of the mean reduces significantly. We conclude that a mean is not necessarily a suitable representative of the sample time series in terms of shape. Normalization of the sample time series may help to reduce this effect.

\subsubsection{Summary}\label{sec:summary1}

The main findings of our first series of experiments are:
\begin{enumerate}
\item A mean of a sample is usually condensed, unique and shorter than the sample time series.
\item Choosing a mean length that is larger than the optimal value induces only a small structural error in most cases.
\item A mean does not necessarily preserve shapes of the sample time series.
\end{enumerate}
The first two points show that in practice one might hope not to encounter critical problems
by non-unique means or by selecting a wrong mean length.
The third point suggests that smoother shapes of a mean require different objective functions to be minimized or different distance measures.

\subsection{Performance of Heuristics}

The goal of this series of experiments is to assess the performance of state-of-the-art heuristics in terms of minimizing the Fréchet variation. 

\subsubsection{Experimental Setup}\label{subsubsec:exp:algorithms}

\paragraph{Algorithms:}

\Cref{tab:algs} lists the mean-algorithms considered in this experiment.
\begin{table}[h!]
\centering
\begin{tabular}{llc}
\toprule
Algorithm & Acr. & Ref. \\
\midrule
Exact dynamic programming & \aedp & Alg.~\ref{algo:DTWMean}\\ 
Multiple alignment & \asym & \cite{HNF08,PG12}\\
DTW barycenter averaging & \adba & \cite{PKG11}\\
Soft-dtw & \asoft & \cite{CB17}\\
Batch subgradient & \absg & \cite{CB17,SJ17}\\
Stochastic subgradient & \assg & \cite{SJ17}\\
\bottomrule
\end{tabular}
\caption{List of algorithms compared in our experiments.}
\label{tab:algs}
\end{table}

We implemented \aedp and \asym (for two time series) in Java. For  \asoft  and \absg we used the python implementation provided by~\citet{BlondelCode17}. For \adba and \assg we used the implementation in Java \cite{SJcode16}.

\paragraph{Setup:}

All six algorithms were applied to $27,000$ samples of type $\S{S}_{\ucr}$ and to $10,000$ samples of type $\S{S}_{\rw}$. For the $5,000$ samples of type $\S{S}_{\rw}^k$ all mean algorithms but \asym were applied. 

\aedp and \asym required no parameter optimization. For the other algorithms, the following settings were used:
We optimized the algorithms on every sample using different parameter configurations and reported the best result. Every configuration was composed of an initialization method and an optional parameter selection. As initialization we used (i) the arithmetic mean of the sample, (ii) a random time series from the sample, and (iii) a random time series drawn from a normal distribution with zero mean and standard deviation one. \adba and \absg required no additional parameter selection. For \asoft, we selected the best smoothing parameter~$\gamma\in\{0.001, 0.01, 0.1, 1\}$ and for \assg the best initial step size~$\eta_0\in\{0.25, 0.2, 0.15, 0.1\}$. The step size of \assg  was linearly decreased from $\eta_0$ to $\eta_0/10$ within the first $100$~epochs and then remained constant for the remaining $100$ epochs (one epoch is a full pass through the sample). Thus, for \adba and \absg, we picked the best result from the three initialization methods for each sample, and for \asoft and \assg we picked the best result from twelve parameter configurations (three initializations $\times$ four parameter values). The length of the mean was set beforehand to the length of the sample time series (time series of a sample always have identical length). For every sample and every parameter configuration, all algorithms terminated after $200$ epochs at the latest. The tolerance parameter of \asoft  and \absg was set to~$\varepsilon = 10^{-6}$.

\paragraph{Performance Metric:}
We consider error percentages to assess the solution quality. The error percentage of algorithm $A$ for sample $\S{S}$ is defined by
\[
E = 100\cdot (F_A - F_*)/F_*,
\] 
where~$F_A$ is the solution obtained by algorithm~$A$ and~$F_* = \min_{z\in\T} F(z)$ is the optimal solution obtained by \aedp.

\subsubsection{Results and Discussion}

\Cref{fig:exp_ucr_k2,fig:exp_rw_k2} depict the performance profiles of the algorithms and quantitative summaries of their error percentages for samples of type $\S{S}_{\ucr}$ and $\S{S}_{\rw}$, respectively.\footnote{Appendix \ref{app:pp} describes performance profiles in more detail.}
The performance profiles exhibit the same pattern in general but differ by some shifts in the curves which is due to different data domains.
(Appendix \ref{app:results} presents the results in more detail.)

\begin{figure*}[t!]
\begin{subfigure}{0.54\textwidth}
\includegraphics[width=\textwidth]{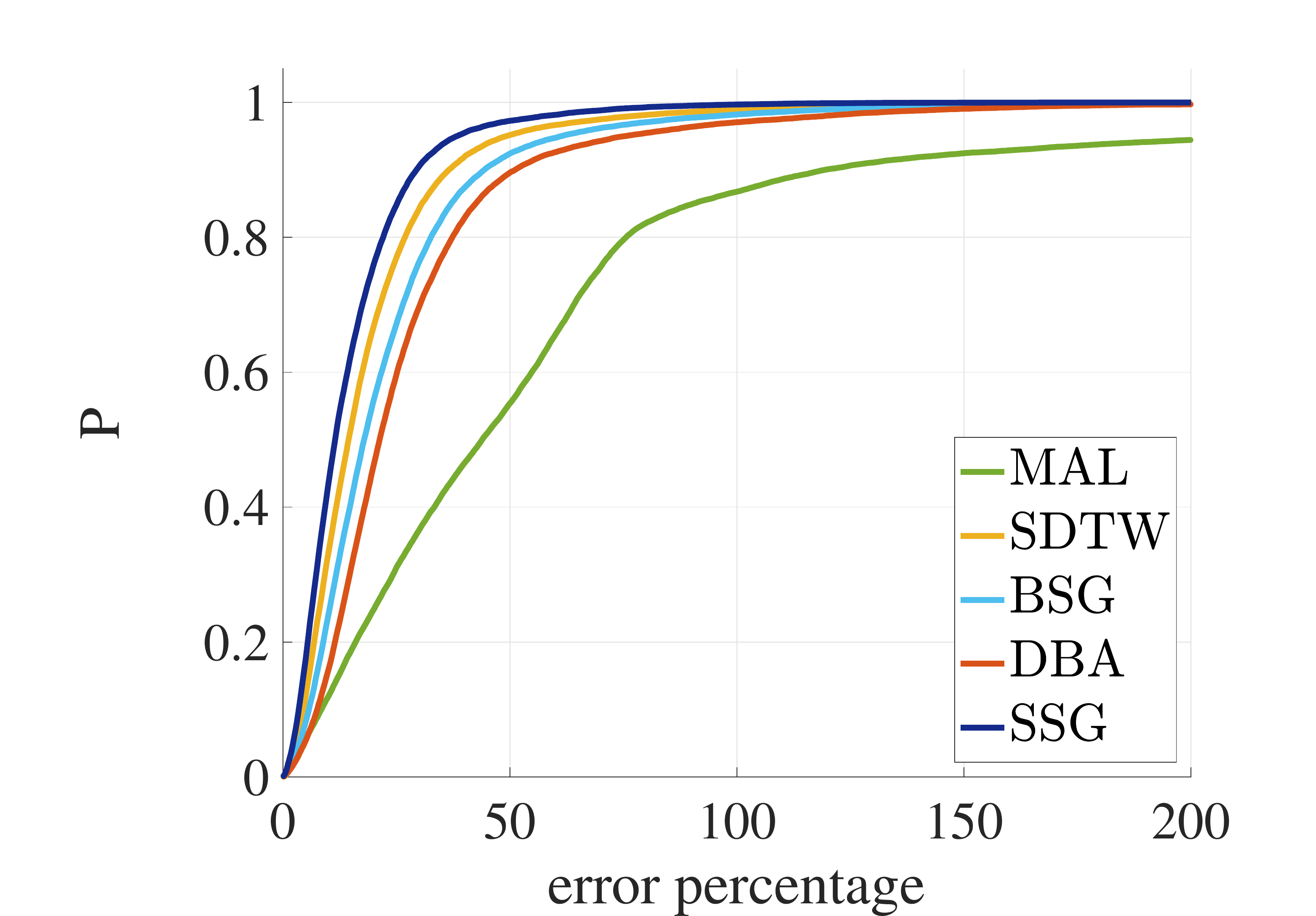}
\end{subfigure}
\begin{subfigure}{0.44\textwidth}
\scriptsize
\begin{tabular}[b]{l@{\qquad}rrrr}
\toprule
 	   &  avg &   std & max & eq \\
\midrule
\asym  & 68.3 & 107.8 & 1105.5 & 3 \\
\asoft & 19.1 &  18.5 &  381.9 & 0 \\
\absg  & 23.5 &  22.4 &  394.1 & 2 \\
\adba  & 28.1 &  27.3 &  437.0 & 0 \\
\assg  & 15.1 &  14.0 &  214.1 & 1 \\
\bottomrule
\end{tabular}
\end{subfigure}
\caption{Results for $\S{S}_{\ucr}$-samples. \textbf{Left:} Performance profiles of all heuristics (a larger area under the curve indicates better performance). The performance profile of the exact algorithm \aedp is the constant line~$P = 1$ and therefore not highlighted. The other performance profiles are truncated at $200\%$ error for the sake of presentation. 
A point~$(E_A, P_A)$ on the curve of an algorithm~$A$ states that solutions obtained by~$A$ deviate by at most~$E_A$ percent from the optimal solution with probability~$P_A$ (estimated over~$27{,}000$ different samples). \textbf{Right:} Average error percentage~(avg), standard deviation~(std), and maximum error percentage~(max) of the heuristics. The last column shows how often an optimal solution was found by the heuristic.}
\label{fig:exp_ucr_k2}
\end{figure*}

\begin{figure*}[t!]
\begin{subfigure}{0.54\textwidth}
\includegraphics[width=\textwidth]{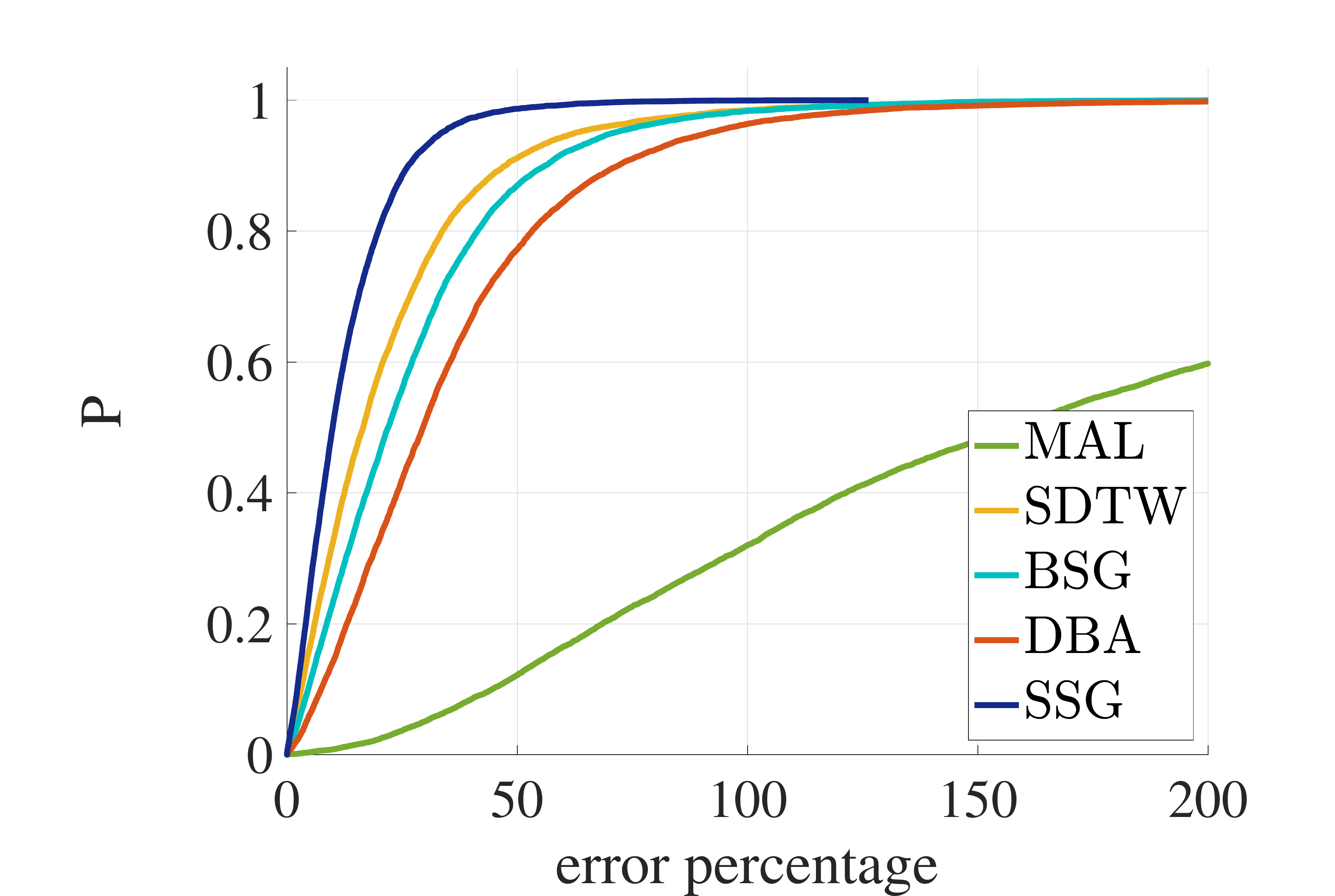}
\end{subfigure}
\begin{subfigure}{0.44\textwidth}
\scriptsize
\begin{tabular}[b]{l@{\qquad}rrrr}
\toprule
 	   &  avg &   std & max & eq \\
\midrule
\asym  & 210.7 & 174.7 & 1239.0 & 3 \\
\asoft &  22.5 &  22.7 &  254.9 & 4 \\
\absg  &  27.4 &  23.7 &  268.9 & 6 \\
\adba  &  36.0 &  29.5 &  361.9 & 3 \\
\assg  &  12.8 &  11.3 &  126.3 & 2 \\
\bottomrule
\end{tabular}
\end{subfigure}
\caption{Results for~$\S{S}_{\rw}$-samples. \textbf{Left:} Performance profiles of all heuristics. \textbf{Right:} Average error percentage~(avg), standard deviation~(std), and maximum error percentage~(max) of the heuristics. The last column shows how often an optimal solution was found by the heuristic.}
\label{fig:exp_rw_k2}
\end{figure*}

\paragraph{Performance of Multiple Alignment (MAL).}
\Cref{ssec:exact} refutes the claim that \DTW can be solved by a multiple alignment approach.
The remaining question is to which extent~\asym fails to solve \DTW.

\Cref{fig:exp_ucr_k2,fig:exp_rw_k2} show that~\asym exactly solved \DTW in only three out of~$27,000$ $\S{S}_{\ucr}$-samples and another three out of~$10,000$ $\S{S}_{\rw}$-samples. The average error percentages of~\asym on~$\S{S}_{\ucr}$ and~$\S{S}_{\rw}$ are~$68.3$ and~$210.7$, respectively. The error percentage of~\asym is larger than~$50 \%$ for more than~$40 \%$ of all~$\S{S}_{\ucr}$-samples and for roughly~$90 \%$ of all~$\S{S}_{\rw}$-samples. These observations suggest that~\asym is far from being optimal or competitive.
This may explain why algorithms based on pairwise mean computation via \asym such as NLAAF \cite{GMTS96} and PSA \cite{NR09} are not competitive~\cite{PKG11,Soheily-Khah2015}.

\paragraph{Performance of State-of-the-Art Heuristics.}
We discuss the performance of \asoft, \absg, \adba, and \assg. The best and most robust method is the stochastic subgradient method (\assg) with~$15.1 \%$ and $12.8 \%$ average error percentage on~$\S{S}_{\ucr}$- and~$\S{S}_{\rw}$-samples, respectively. Standard deviation and maximum error percentages of~\assg are lowest among all other heuristics.
Nevertheless, the solution quality of all heuristics is rather poor leaving much space for improvements.

\paragraph{Performance of Progressive Alignment (PSA).}
Progressive alignment methods combine pairwise averages beginning with the two most similar time series. In each step, two time series are averaged and replaced by their weighted average. The weights correspond to the number of sample time series involved in creating the time series. The currently best performing progressive alignment method for time series averaging is called prioritized shape averaging (PSA)~\cite{NR09}. This method applies a variant of \asym for pairwise averaging. Empirical results suggest that PSA is not competitive to \adba~\cite{PKG11,Soheily-Khah2015}. We evaluate the performance of a modified variant of PSA using \aedp for pairwise averaging.

\begin{figure}
\begin{subfigure}{0.49\textwidth}
\includegraphics[width=\textwidth]{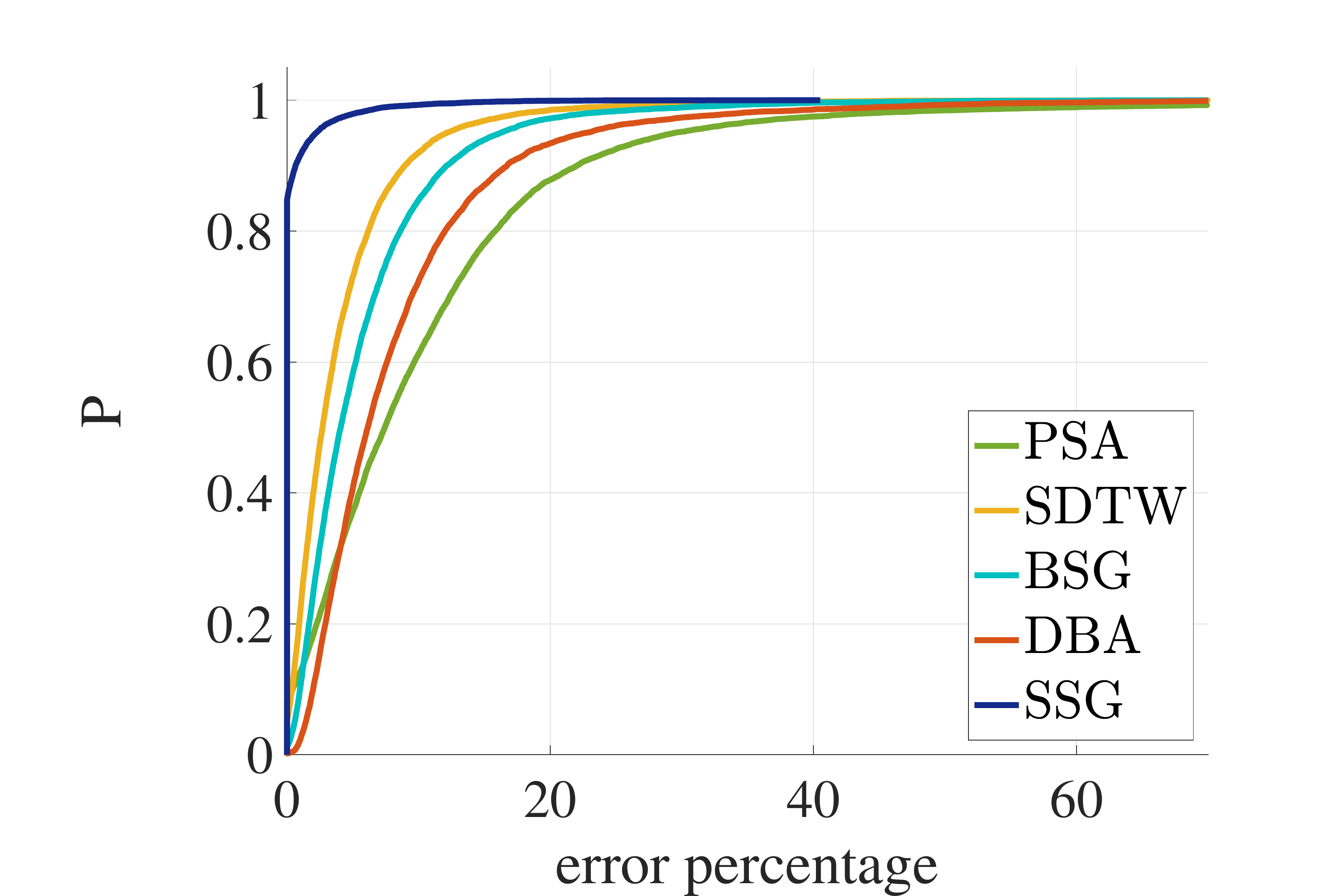}
\caption{Performance profiles}
\end{subfigure}
\hfill
\begin{subfigure}{0.49\textwidth}
\includegraphics[width=\textwidth]{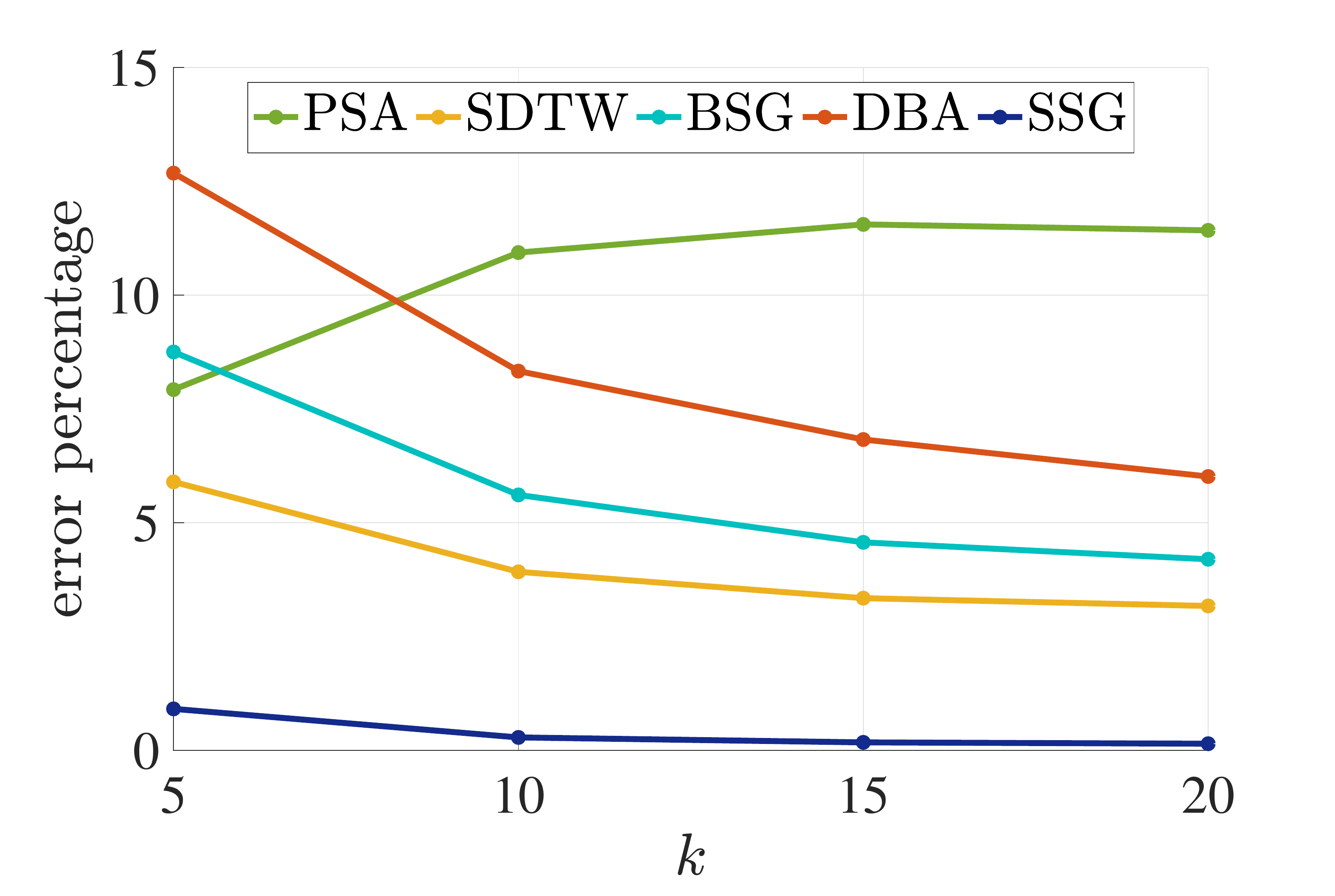}
\caption{Error percentage as a function of $k$}
\end{subfigure}
\caption{Results on samples~$\S{S}_{\ucr}^k$. Plot (a) shows the performance profiles and plot (b) the average error percentages as a function of the sample size $k$.}
\label{fig:exp_ucr_K}
\end{figure}

As test data, we used samples of type $\S{S}_{\ucr}^k$: For every every UCR data set and every sample size $k \in \cbrace{5, 10, 15, 20}$, we randomly selected~$100$ samples of~$k$ time series giving a total of~$10,800$ samples. We applied the modified PSA algorithm and the algorithms \asoft, \absg, \adba, and \assg on all~$10,800$ samples (the setup is the same as described in \Cref{subsubsec:exp:algorithms}). 

\Cref{fig:exp_ucr_K} summarizes the results. The performance profiles exhibit similar patterns for $\S{S}_{\ucr}^k$-samples as those obtained for $\S{S}_{\ucr}$-samples (\Cref{fig:exp_ucr_k2}).
PSA is the worst performing heuristic. \Cref{fig:exp_ucr_K}~(b) shows that the average error percentage of PSA increases with increasing sample size~$k$, whereas the average error percentage of all other heuristics exhibit the opposite trend.
These results show that the weak performance of progressive alignment methods is not explained by
suboptimal pairwise mean computation.

\paragraph{Error Decomposition.}
The previous results revealed a poor solution quality of state-of-the-art heuristics.
In \Cref{sec:length}, we have seen that choosing a fixed mean length~$m$ introduces a certain structural error.
The error~$\varepsilon$ of an algorithm~$A$ can thus be written as
\[
\varepsilon =  F_A \;-\; F_* =  \underbrace{\args{F_A \;-\; F_m^*}}_{\text{approximation error}} + \quad \underbrace{\args{F_m^* \;-\; F_*}}_{\text{structural error } \varepsilon_m},
\]
where $F_A$ is the value returned by algorithm~$A$, $F_* = \min_{z\in\T} F(z)$ is the Fr\'echet variation, and~$F_m^* = \min_{z\in\T_m} F(z)$ is the constrained Fr\'echet variation.
The approximation error~$F_A \;-\; F_m^*$ represents the inability of algorithm~$A$ to solve the constrained \DTW problem.

We conducted some experiments to assess the contribution of the approximation and the structural error to the total error of the considered heuristics.
For every length~$n \in \cbrace{5, 10, 15, \ldots, 40}$, we generated~$1,000$ samples of pairs of random walks of length $n$ giving a total of~$8,000$ samples. We applied~\asoft, \absg, \adba, and \assg~on all samples using the same setting as described in \Cref{subsubsec:exp:algorithms}.

\begin{figure}[t]
\begin{subfigure}{0.49\textwidth}
\includegraphics[width=\textwidth]{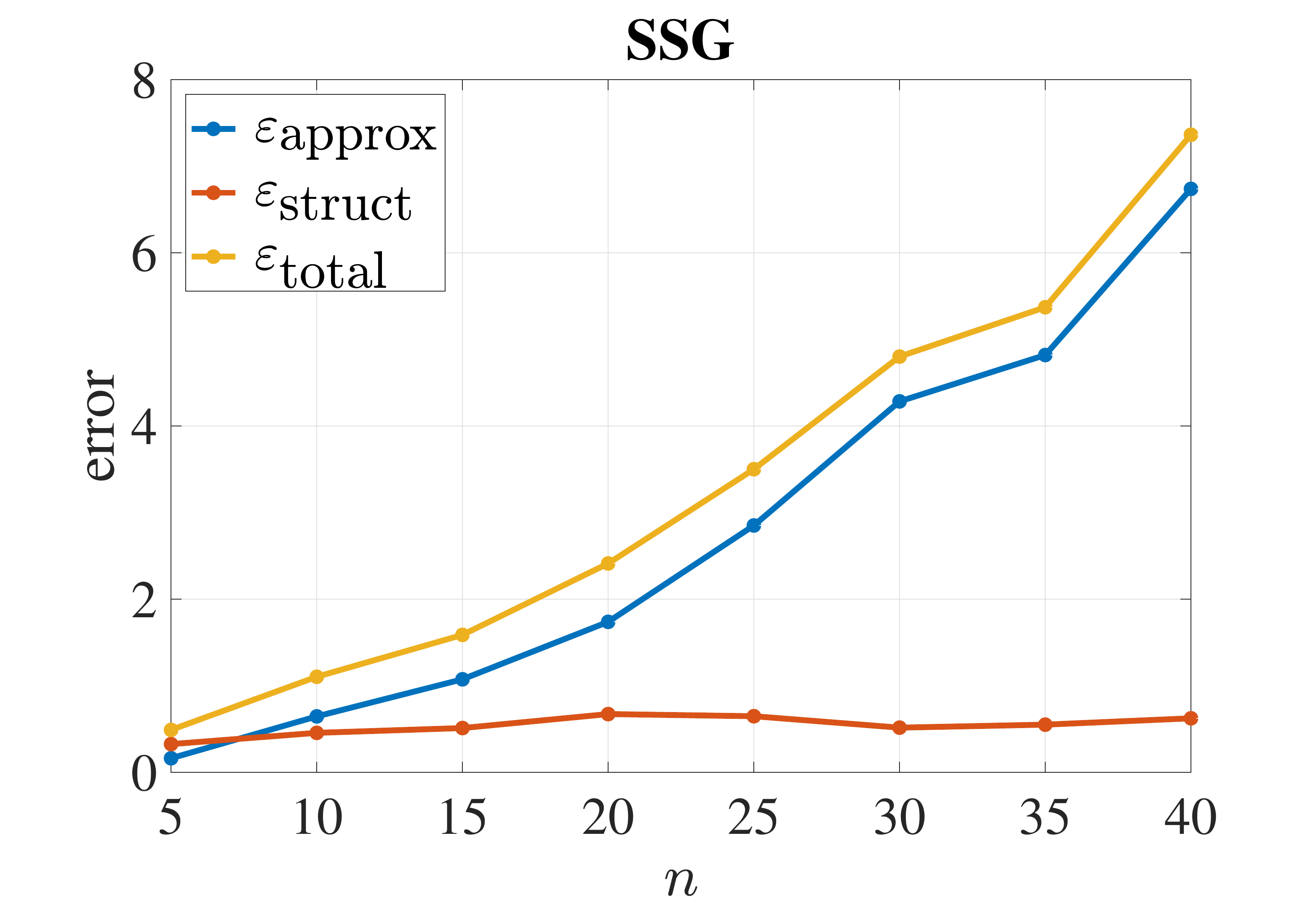}
\caption{Average errors}
\end{subfigure}
\hfill
\begin{subfigure}{0.49\textwidth}
\includegraphics[width=\textwidth]{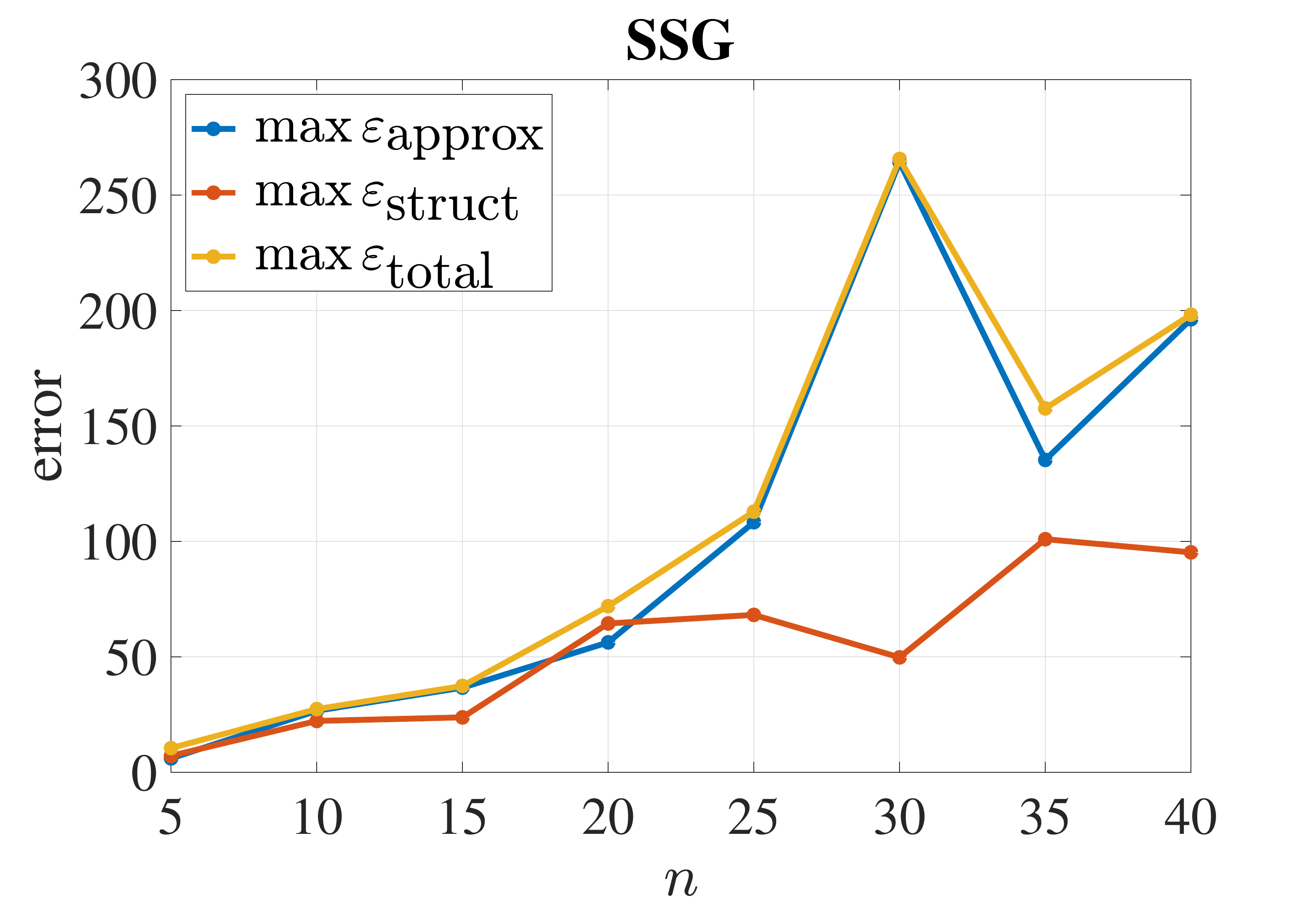}
\caption{Maximum errors}
\end{subfigure}
\caption{Error decomposition of~\assg on pairs of random walks with $5 \leq n \leq 40$. 
Both plots show the total error~$\varepsilon_{\total}$ (yellow), the approximation error~$\varepsilon_{\approx}$ (blue), and the structural error~$\varepsilon_{\struct}$ (red). 
Plot (a) averages and plot (b) maximizes the errors over~$1,000$ samples at each~$n$.}
\label{fig:exp_rw_K2_L_err_decomp}
\end{figure}

\Cref{fig:exp_rw_K2_L_err_decomp} depicts the results of the best performing heuristic (\assg). The other three heuristics exhibited the same behavior with identical structural error but different approximation errors. We made the following observations:
\begin{enumerate}
\item The total error is dominated by the approximation error (for larger~$n$).
\item Both total and approximation error increase with increasing length $n$.
\item The structural error is independent of the length $n$. 
\item In exceptional cases the structural error can be large. 
\end{enumerate}
These observations indicate that the total error is mainly caused by the approximation error. In contrast, the structural error introduced by fixing the mean length becomes negligible for longer sample time series and is large only in very few cases (as concluded in \Cref{sec:summary1}). 
The results suggest that devising heuristics should mainly focus on improving the approximation error.

\begin{figure}[t]
\begin{subfigure}{0.49\textwidth}
\includegraphics[width=\textwidth]{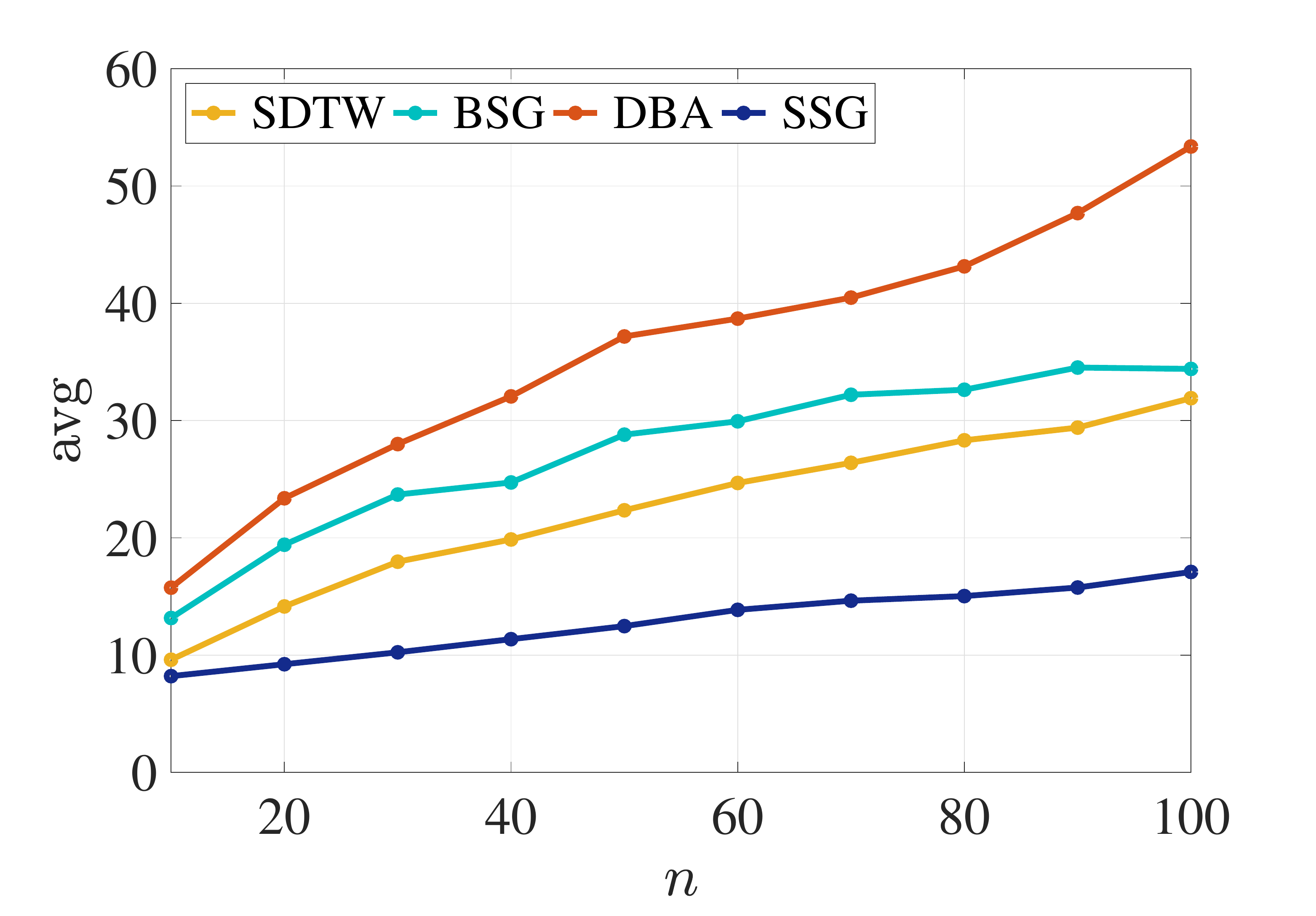}
\end{subfigure}
\hfill
\begin{subfigure}{0.49\textwidth}
\includegraphics[width=\textwidth]{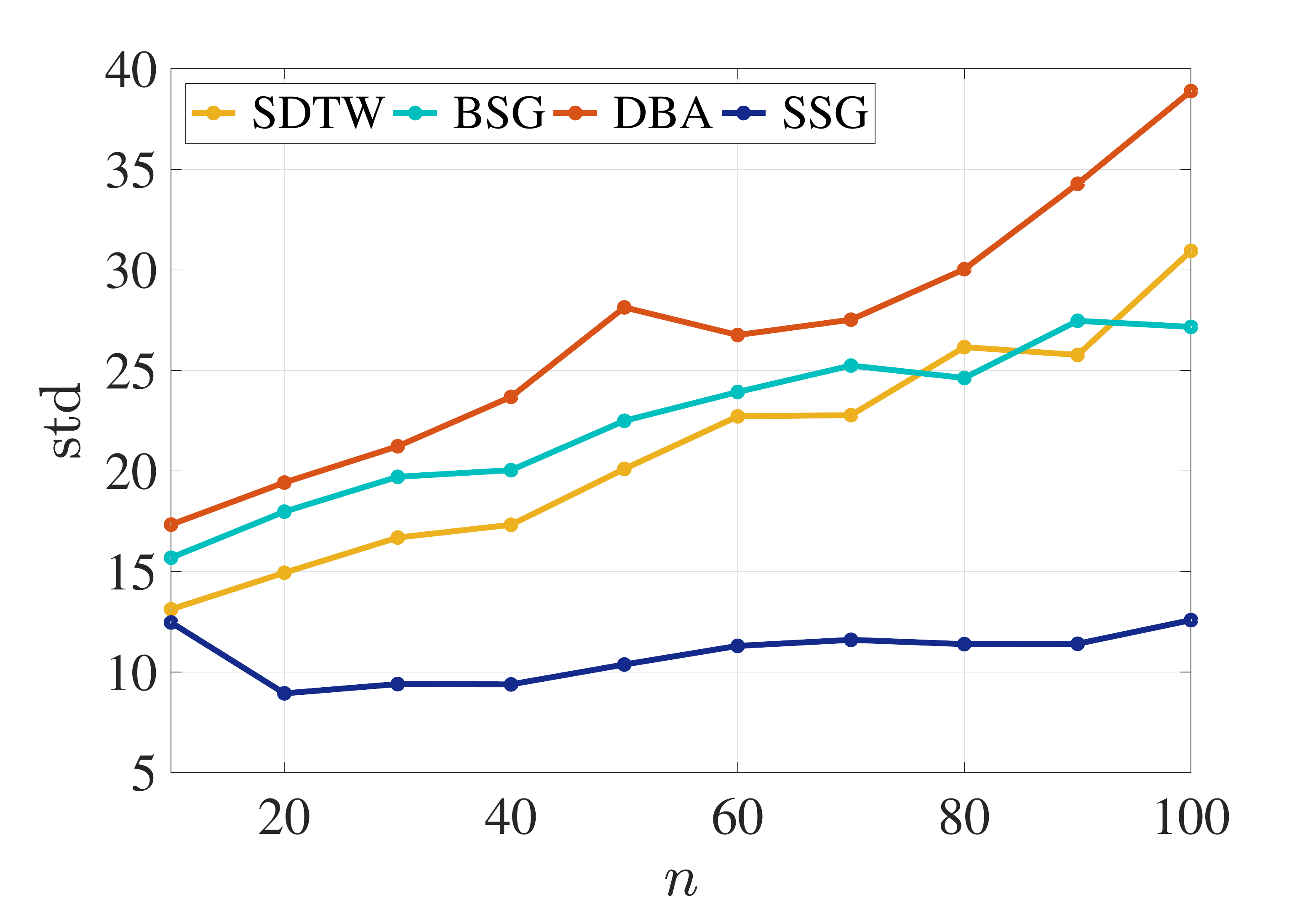}
\end{subfigure}
\begin{subfigure}{0.49\textwidth}
\includegraphics[width=\textwidth]{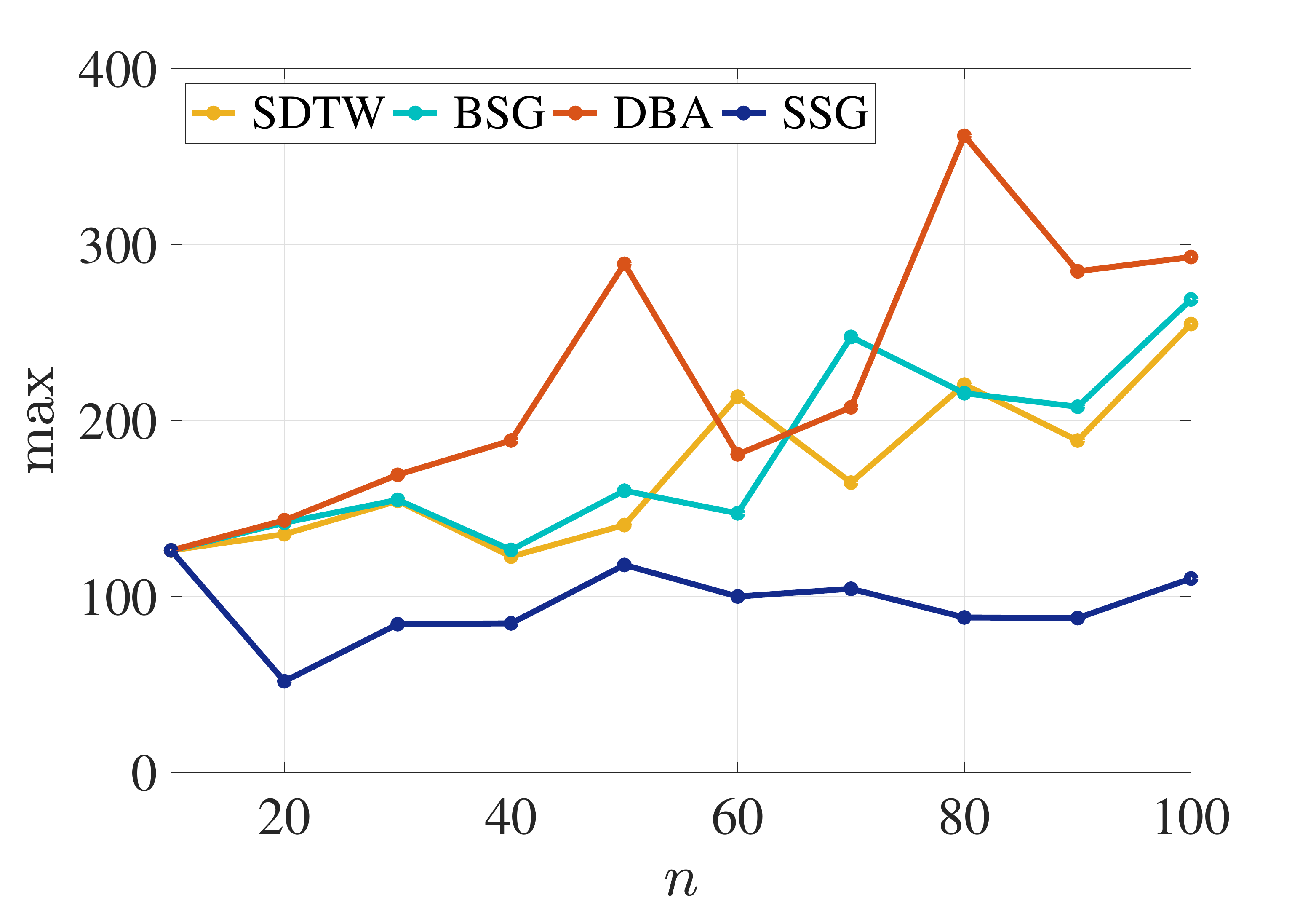}
\end{subfigure}
\hfill
\begin{subfigure}{0.49\textwidth}
\includegraphics[width=\textwidth]{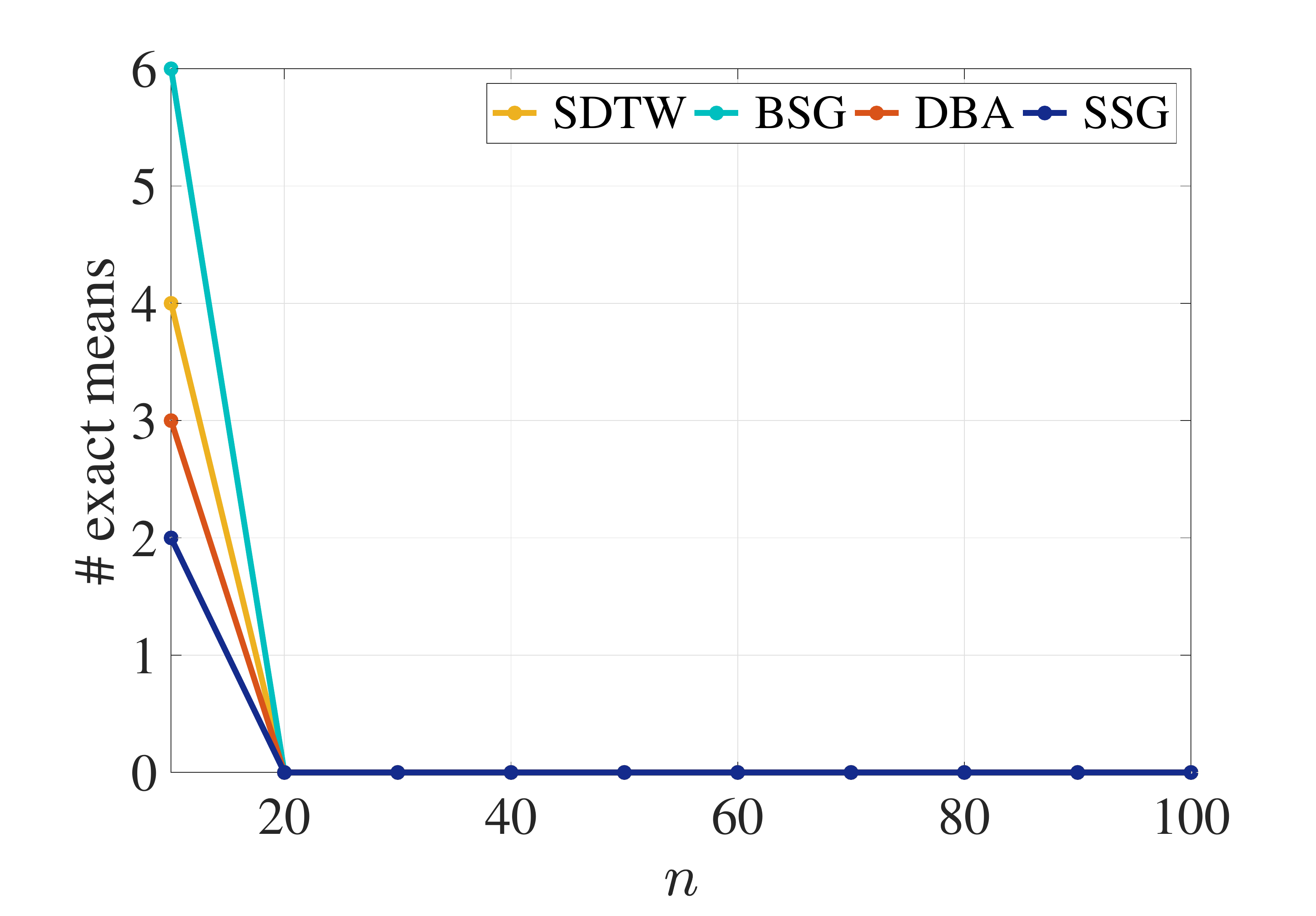}
\end{subfigure}
\caption{Results on samples of type $\S{S}_{\rw}$ consisting of $k = 2$ random walks of length $n$. Average error percentage~(avg), standard deviation~(std), maximum error percentage~(max) and number of exact solutions are plotted as a function of the length $n$.}
\label{fig:exp_rw_K2_Lengths}
\end{figure}

\paragraph{Effect of Sample Time Series Length~$n$.}
We investigated how the performance of~\asoft, \absg, \adba, and \assg~is related to the length $n$ of the sample time series for random walk samples of type $\S{S}_{\rw}$.

\Cref{fig:exp_rw_K2_Lengths} summarizes the results. The trend is that the performance of all heuristics with respect to solution quality (avg) and robustness (std, max) deteriorates with increasing length~$n$. Since we observed before that the structural error is independent of $n$ and likely to be negligible, we conclude that the heuristics perform worse in solving the constrained \DTW problem for larger~$n$. 
No state-of-the-art method found an optimal solution for length $n \geq 20$.

\begin{figure}[t]
\begin{subfigure}{0.49\textwidth}
\includegraphics[width=\textwidth]{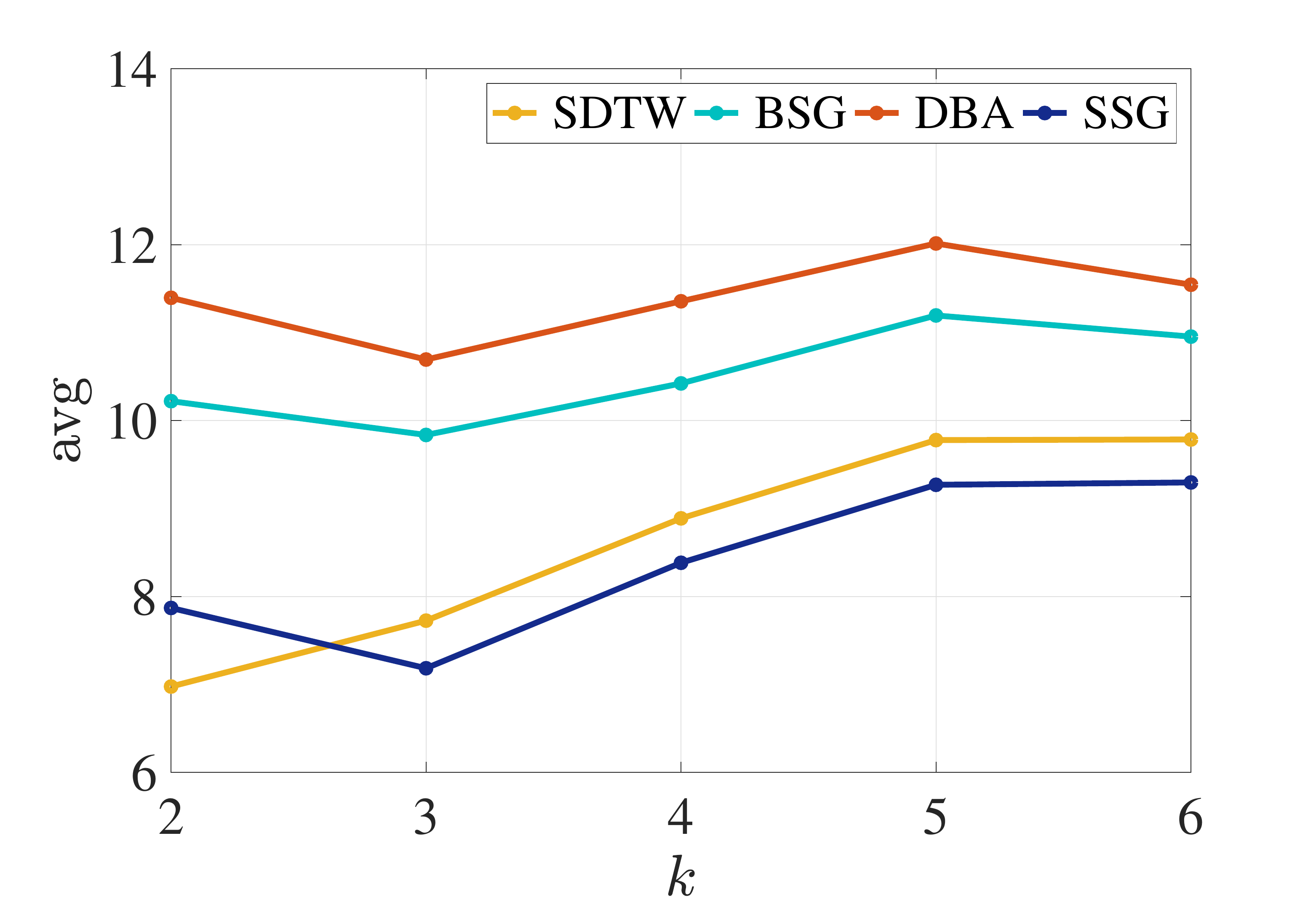}
\end{subfigure}
\hfill
\begin{subfigure}{0.49\textwidth}
\includegraphics[width=\textwidth]{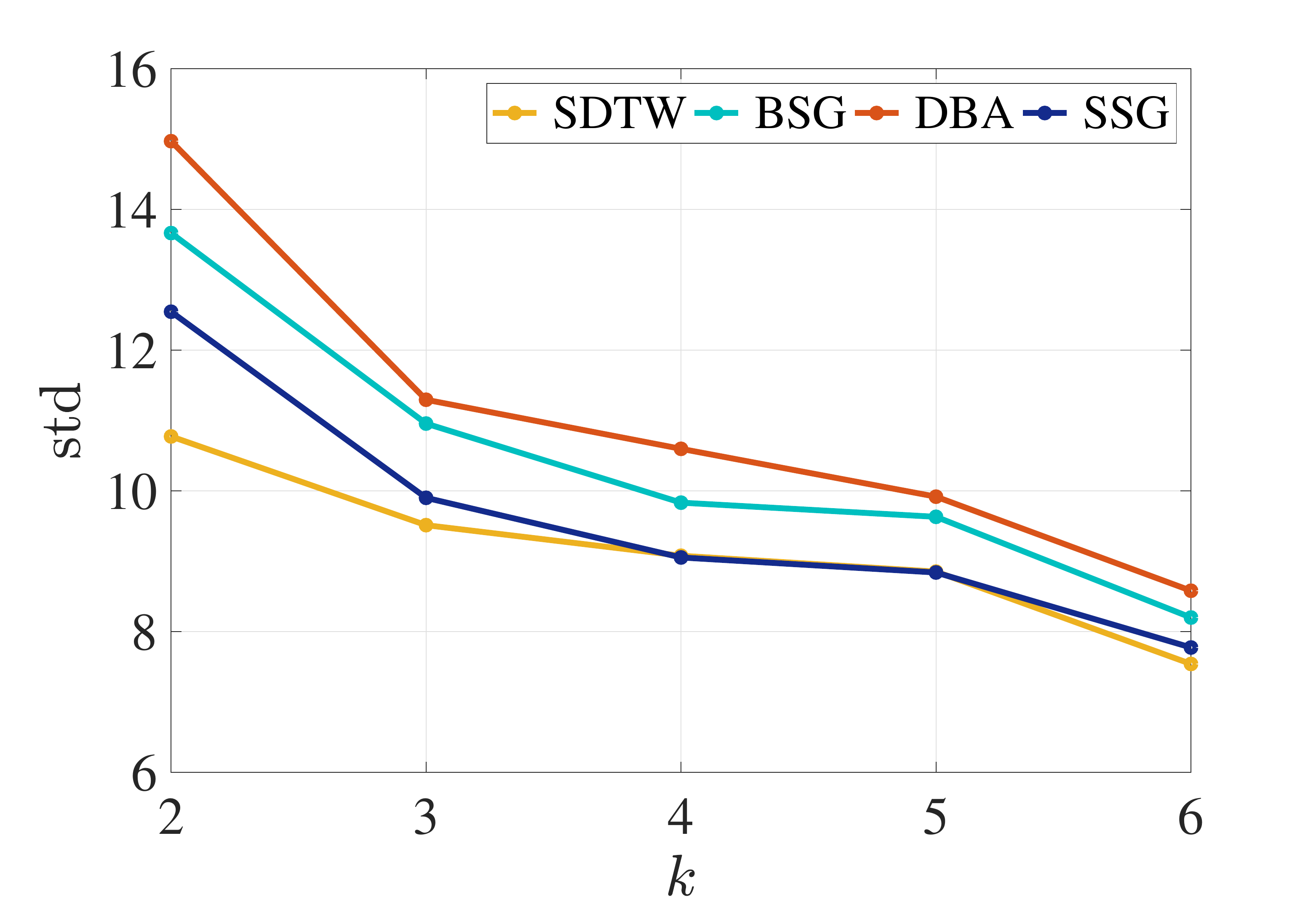}
\end{subfigure}
\begin{subfigure}{0.49\textwidth}
\includegraphics[width=\textwidth]{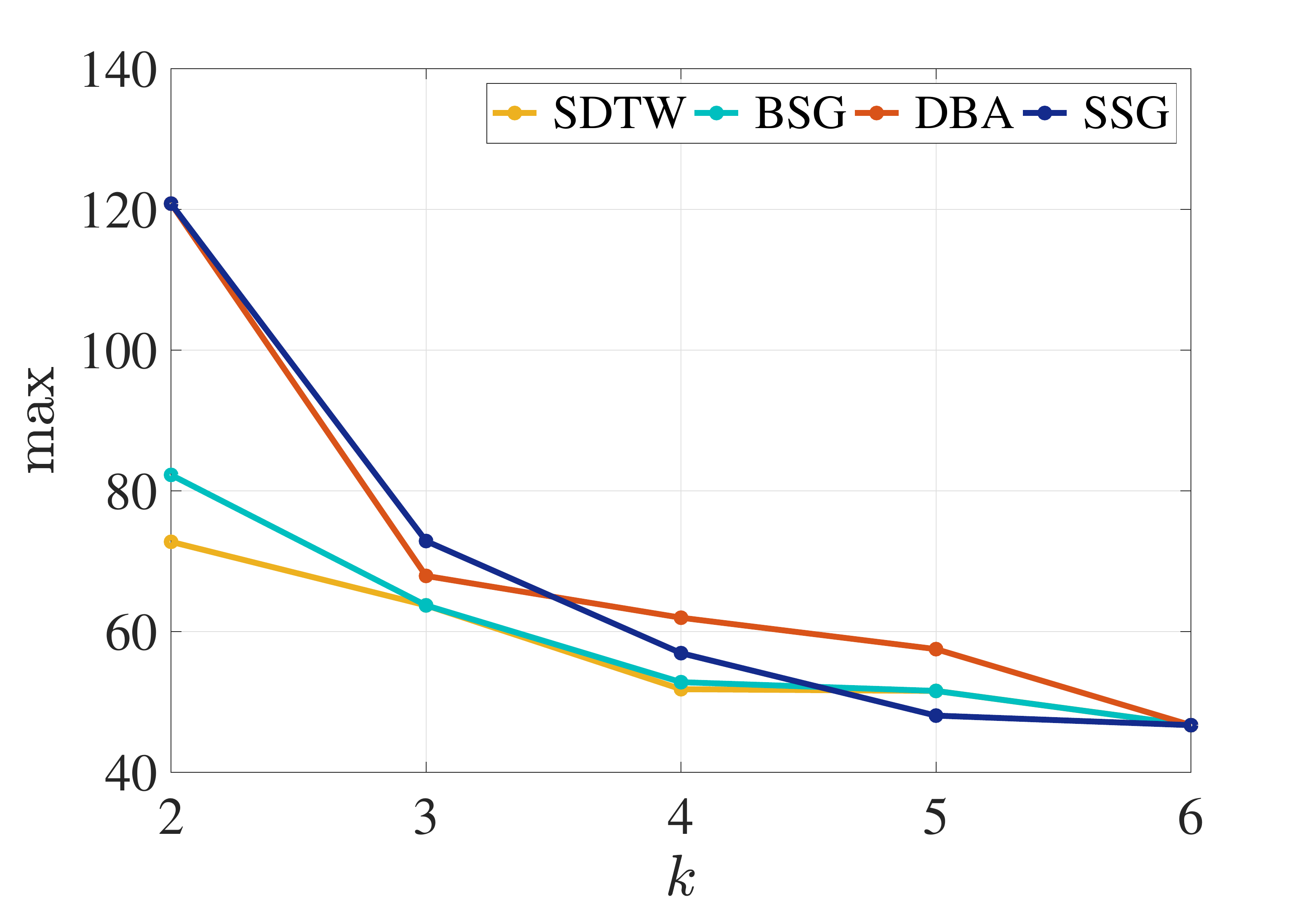}
\end{subfigure}
\hfill
\begin{subfigure}{0.49\textwidth}
\includegraphics[width=\textwidth]{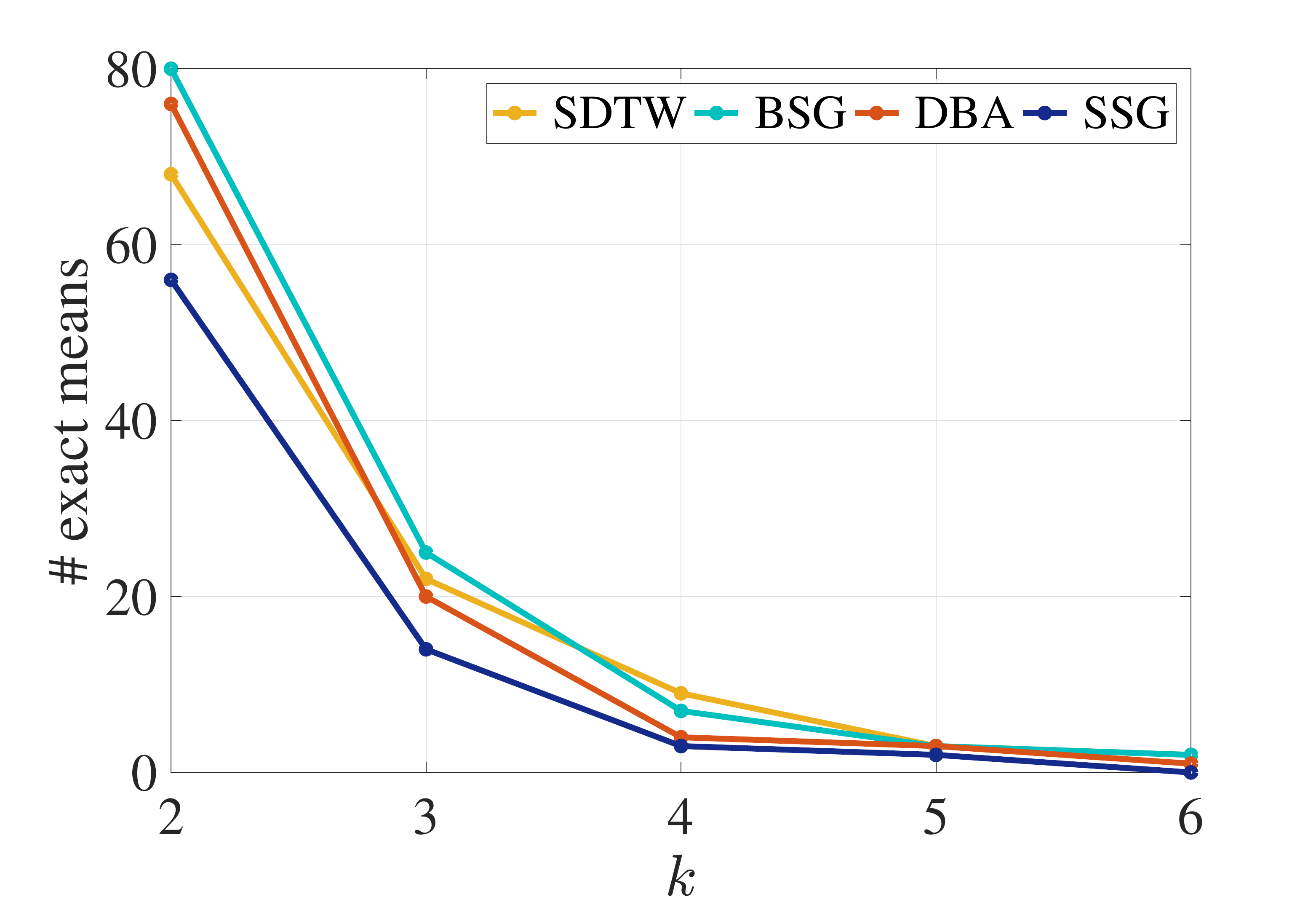}
\end{subfigure}
\caption{Results on samples $\S{S}_{\rw}^k$ each of which consists of $k$ random walks. Average error percentage~(avg), standard deviation~(std), maximum error percentage~(max) and number of exact solutions are plotted as a function of the sample size $k$.}
\label{fig:exp_rw_L6_sample-sizes}
\end{figure}

\paragraph{Effect of Sample Size~$k$.}
We analyzed how the performance of \asoft, \absg, \adba, and \assg~depends on the sample size $k$ for samples of type $\S{S}_{\rw}^k$. 

\Cref{fig:exp_rw_L6_sample-sizes} shows the results. We observed no clear trend with respect to average error-percentage and a  decline of the number of exact solutions found by all heuristics. In contrast, robustness (that is, standard deviation and maximum error percentage) of all four heuristics substantially improved with increasing $k$. 

\subsubsection{Summary}

The main findings of our second series of experiments are:
\begin{enumerate}
  \item State-of-the-art heuristics do not perform well on average due to a large approximation error.
  \item \asym is not competitive to state-of-the-art heuristics (for $k = 2$). 
  \item PSA using exact weighted means is not competitive to state-of-the-art heuristics.
  \end{enumerate}
  Concerning running times, our exact dynamic program is clearly slower than state-of-the-art heuristics (for example, our implementation required 66 seconds on average to compute a mean of two time series of length 100, which is slower by a factor of $10^5$ compared to \adba).

\section{Conclusion}
We developed an exact exponential-time algorithm for \wDTW and
conducted extensive experiments (on small samples containing up to 6 time series of different lengths up to 150) to investigate characteristics of an exact mean and to benchmark existing state-of-the-art heuristics in terms of their solution quality.
On the positive side, we found that a mean is often unique and that fixing the length of a mean in advance does usually not affect the solution quality much.
On the negative side, we showed that basically all heuristics perform poorly in the worst case.
These findings urge for devising better algorithms to solve \DTW.
Our empirical observations (e.g. concerning the typical mean length) might prove useful for this.
As a side result, we also developed a polynomial-time algorithm for binary time series.

We conclude with some challenges for future research.
From an algorithmic point of view it is interesting to investigate whether one can extend the polynomial-time solvability of \BDTW to larger ``alphabet'' sizes; already the case of 
alphabet size three is open.
Another open question is whether \wDTW is polynomial-time solvable for time series of constant lengths (maybe it is even fixed-parameter tractable with respect to the maximum length).
Finally, we wonder whether there are other practically relevant restrictions 
of \DTW that make the problem more tractable, for example, fixing the length of a mean.
Another example, also motivated from a practical point of view, is to compute means with a given size of time windows (also known as the Sakoe-Chiba band~\cite{SC78}).
On a high level, a time window constrains the warping path to only select tuples within a specified range.
On an applied side, our experimental results strongly motivate the 
search for further improved, more ``robust'' heuristics.

\paragraph{Acknowledgements.}
This work was supported by the Deutsche Forschungsgemeinschaft under grants {JA~2109/4-1} and {NI~369/13-2}, and by a Feodor Lynen return fellowship of the Alexander von Humboldt Foundation.  
The work on the theoretical part of this paper started at the research retreat
of the Algorithmics and Computational Complexity group, TU Berlin, 
held at Boiensdorf, Baltic Sea, April~2017, 
with MB, TF, VF, and~RN participating.

\makeatletter
\renewcommand\bibsection%
{
  \section*{\refname
    \@mkboth{\MakeUppercase{\refname}}{\MakeUppercase{\refname}}}
}
\makeatother

{ 
\bibliographystyle{abbrvnat}
\bibliography{./dmkd/ref}
}

\newpage
\appendix
\section{Performance Profiles}\label{app:pp}

To compare the performance of the mean algorithms, we used a slight variation of the performance profiles proposed by \citet{Dolan2002}. A performance profile is a cumulative distribution function for a performance metric. Here, the chosen performance metric is the error percentage from the exact solution. 

To define a performance profile, we assume that $\mathbb{A}$ is a set of mean algorithms and $\mathbb{S}$ is a set of samples each of which consists of $k$ time series. For each sample $\S{X} \in \mathbb{X}$ and each mean algorithm $A \in \mathbb{A}$, we define $E_{A,\S{S}}$ as the error percentage obtained by applying algorithm $A$ on sample $\S{S}$. The performance profile of algorithm $A$ over all samples $\S{S} \in \mathbb{S}$ is the empirical cumulative distribution function defined by
\[
P_A(\tau) = \frac{1}{\abs{\mathbb{S}}} \abs{\cbrace{\S{S} \in \mathbb{S} \,:\, E_{A,\S{S}} \leq \tau}}
\]
for all $\tau \geq 0$. Thus, $P_A(\tau)$ is the estimated probability that the error percentage of algorithm $A$ is at most $\tau$. The value $P_A(0)$ is the estimated probability that algorithm $A$ finds an exact solution.

\section{Results}\label{app:results}

\begin{table}[t]
\centering
\scriptsize
\begin{tabular}{l@{\qquad}rrrrr}
\toprule
UCR data set                    & \asym & \asoft & \absg & \adba & \assg\\
\midrule
ItalyPowerDemand               &  26.1 &  13.8 &  17.6 &  19.2 &  \B 10.5 \\
synthetic\_control             &  54.7 &  22.1 &  26.5 &  29.4 &  \B 17.9 \\
SonyAIBORobotSurface           &  33.6 &  20.6 &  26.0 &  28.9 &  \B 17.2 \\
SonyAIBORobotSurfaceII         &  28.4 &  20.0 &  23.8 &  26.7 &  \B 17.8 \\
ProximalPhalanxTW              &  35.2 &  11.8 &  15.3 &  17.1 &  \B 10.6 \\
ProximalPhalanxOutlineCorrect  &  35.1 &  13.4 &  16.9 &  18.4 &  \B 11.9 \\
ProximalPhalanxOutlineAgeGroup &  36.9 &  11.7 &  15.3 &  17.4 &  \B 10.8 \\
PhalangesOutlinesCorrect       &  49.7 &  15.2 &  19.5 &  22.3 &  \B 13.2 \\
MiddlePhalanxTW                &  36.8 &  12.5 &  14.8 &  17.3 &  \B 10.0 \\
MiddlePhalanxOutlineCorrect    &  38.7 &  13.3 &  16.2 &  18.4 &  \B 10.8 \\
MiddlePhalanxOutlineAgeGroup   &  37.0 &  12.8 &  15.1 &  17.6 &  \B  9.9 \\
DistalPhalanxTW                &  41.9 &  12.0 &  15.3 &  18.6 &  \B 10.5 \\
DistalPhalanxOutlineCorrect    &  48.4 &  13.4 &  17.4 &  21.2 &  \B 11.6 \\
DistalPhalanxOutlineAgeGroup   &  43.6 &  12.4 &  15.5 &  18.6 &  \B 10.6 \\
TwoLeadECG                     &  44.4 &  12.0 &  15.3 &  17.7 &  \B  9.6 \\
MoteStrain                     &  79.7 &  17.5 &  23.3 &  29.2 &  \B 12.6 \\
ECG200                         &  51.4 &  20.5 &  23.5 &  26.5 &  \B 14.2 \\
MedicalImages                  &  84.6 &  16.8 &  21.0 &  27.6 &  \B 10.6 \\
Two\_Patterns                  & 184.0 &  73.5 &  82.8 &  92.8 &  \B 56.4 \\
SwedishLeaf                    &  56.2 &  18.0 &  23.2 &  30.6 &  \B 13.4 \\
CBF                            &  28.5 &  31.2 &  33.8 &  36.9 &  \B 26.4 \\
FacesUCR                       &  38.4 &  22.5 &  27.6 &  30.2 &  \B 18.0 \\
FaceAll                        &  37.7 &  22.6 &  27.6 &  30.2 &  \B 18.3 \\
ECGFiveDays                    &  41.7 &  11.0 &  13.0 &  15.9 &  \B  8.5 \\
ECG5000                        &  77.8 &  15.3 &  18.0 &  22.4 &  \B  9.9 \\
Plane                          &  83.1 &  13.6 &  18.6 &  30.0 &  \B  9.7 \\
Gun\_Point                     & 489.8 &  36.2 &  52.0 &  77.0 &  \B 27.4 \\
\midrule
\textbf{Total}                 &  68.3 &  19.1 &  23.5 &  28.1 &  \B 15.1 \\
\bottomrule
\end{tabular}
\caption{Average error percentage of the five heuristics on samples $\S{S}_{\ucr}$ of size $k = 2$ grouped by UCR data set. Averages were taken over $1,000$ samples randomly drawn from each data set.}
\label{tab:res_ucr_k2}
\end{table}

\begin{table}[t]
\centering
\scriptsize
\begin{tabular}{r@{\qquad}rrrrr}
\toprule
$n$  & \asym & \asoft & \absg & \adba & \assg \\
\midrule
10  &  64.3 &   9.6 &  13.2 &  15.8 &  \B 8.2  \\
20  & 107.4 &  14.2 &  19.4 &  23.4 &  \B 9.2  \\
30  & 152.1 &  18.0 &  23.7 &  28.0 &  \B 10.2 \\
40  & 179.1 &  19.9 &  24.7 &  32.1 &  \B 11.4 \\
50  & 206.5 &  22.4 &  28.8 &  37.2 &  \B 12.5 \\
60  & 231.2 &  24.7 &  29.9 &  38.7 &  \B 13.9 \\
70  & 253.2 &  26.4 &  32.2 &  40.5 &  \B 14.6 \\
80  & 269.7 &  28.3 &  32.6 &  43.2 &  \B 15.0 \\
90  & 303.9 &  29.4 &  34.5 &  47.7 &  \B 15.8 \\
100 & 339.8 &  31.9 &  34.4 &  53.4 &  \B 17.1 \\
\midrule
\textbf{Total} & 210.7 &  22.5 &  27.4 &  36.0 &  \B 12.8 \\
\bottomrule
\end{tabular}
\caption{Average error percentage of the five heuristics on $\S{S}_{\rw}$-samples of size $k = 2$ grouped by length $n$ of random walks. Averages were taken over $1,000$ samples for each length $n$.}
\label{tab:res_rw_k2}
\end{table}

\begin{table}
\centering
\begin{tabular}{ll@{\qquad}rrrr}
\toprule
&&  avg &   std & max & eq \\
\midrule
$k = 2$\\
&\asoft &   7.0 &  10.8 &  72.8 &  68.0 \\
&\absg  &  10.2 &  13.7 &  82.3 &  80.0 \\
&\adba  &  11.4 &  15.0 & 120.8 &  76.0 \\
&\assg  &   7.9 &  12.5 & 120.8 &  56.0 \\
\midrule
$k = 3$\\
&\asoft &   7.7 &   9.5 &  63.7 &  22.0 \\
&\absg  &   9.8 &  11.0 &  63.7 &  25.0 \\
&\adba  &  10.7 &  11.3 &  67.9 &  20.0 \\
&\assg  &   7.2 &   9.9 &  72.9 &  14.0 \\
\midrule
$k = 4$\\
&\asoft &   8.9 &   9.1 &  51.8 &   9.0 \\
&\absg  &  10.4 &   9.8 &  52.8 &   7.0 \\
&\adba  &  11.4 &  10.6 &  62.0 &   4.0 \\
&\assg  &   8.4 &   9.1 &  57.0 &   3.0 \\
\midrule
$k = 5$\\
&\asoft &   9.8 &   8.9 &  51.6 &   3.0 \\
&\absg  &  11.2 &   9.6 &  51.6 &   3.0 \\
&\adba  &  12.0 &   9.9 &  57.5 &   3.0 \\
&\assg  &   9.3 &   8.8 &  48.1 &   2.0 \\
\midrule
$k = 6$\\
&\asoft &   9.8 &   7.5 &  46.7 &   2.0 \\
&\absg  &  11.0 &   8.2 &  46.7 &   2.0 \\
&\adba  &  11.5 &   8.6 &  46.7 &   1.0 \\
&\assg  &   9.3 &   7.8 &  46.7 &   0.0 \\
\bottomrule
\end{tabular}
\caption{Average error percentage of the five heuristics on $\S{S}_{\rw}^k$-samples of varying sample size $k$. Length of random walks was $n = 6$. Averages were taken over $1,000$ samples for each sample size $k$.}
\label{tab:res_rw_L6_k}
\end{table}

\end{document}